\documentclass[aps,prx,twocolumn,superscriptaddress]{revtex4-2}
\pdfoutput=1

\usepackage[page,toc,titletoc,title]{appendix}

\usepackage{footmisc}
\usepackage{subcaption}
\usepackage{hyperref}
\usepackage{mathtools}
\usepackage{amsmath,amssymb,graphicx,xcolor,braket,makeidx,comment,amsthm}

\usepackage{calc}
\newlength\myheight
\newlength\mydepth
\settototalheight\myheight{Xygp}
\usepackage{cancel}
\usepackage{xfrac}
\usepackage{lipsum} 

\theoremstyle{theorem}

\newtheorem{theorem}{Theorem}
\newtheorem{prop}{Proposition}
\newtheorem{corollary}[prop]{Corollary}
\newtheorem{lemma}[prop]{Lemma}
\newtheorem{definition}[prop]{Definition}
\theoremstyle{remark}
\newtheorem{remark}[prop]{Remark}
\newtheorem{conj}[prop]{Conjecture}

\usepackage{dsfont}
\usepackage{bbm} 
\usepackage[normalem]{ulem} 
\usepackage{color} 
\usepackage{braket} 
\usepackage{float}  

\newcommand{\ketbra}[2]{\ket{#1}\!\!\bra{#2}}
\newcommand{\proj}[1]{\ketbra{#1}{#1}}
\newcommand{\tr}{\textup{tr}}
\newcommand{\be}{\begin{equation}}
	\newcommand{\ee}{\end{equation}}
\newcommand{\nn}{{\mathbbm{N}}}
\newcommand{\rr}{{\mathbbm{R}}}
\newcommand{\cc}{{\mathbbm{C}}}

\newcommand{\me}{\mathrm{e}}
\newcommand{\mi}{\mathrm{i}}
\newcommand{\id}{{\mathbbm{1}}} 
\newcommand{\idv}{{{I}}}
\newcommand{\bo}{\mathcal{O}}

\theoremstyle{definition}

\newcommand{\Mspace}{\vspace{0.15cm}} 

\def\ba#1\ea{\begin{align}#1\end{align}} 

\newcommand{\doc}{\text{paper}}  
\newcommand{\App}{\text{Appendix}}
\newcommand{\app}{\text{appendix}}
\newcommand{\Supp}{\text{Supplementary}}
\newcommand{\supp}{\text{supplementary}}

\newcommand\mpwS[1]{{\let\helpcmd\sout\parhelp#1\par\relax\relax} }
\long\def\parhelp#1\par#2\relax{%
	\helpcmd{#1}\ifx\relax#2\else\par\parhelp#2\relax\fi%
}

\theoremstyle{definition}

\newcommand{\Sy}{\textup{S}}
\newcommand{\cat}{\textup{Cat}}
\newcommand{\D}{\textup{D}}
\newcommand{\cl}{\textup{Cl}}
\newcommand{\G}{\textup{G}}
\newcommand{\A}{\textup{A}}
\newcommand{\B}{\textup{B}}
\newcommand{\C}{\textup{C}}

\newcommand{\tauT}{{\tilde\tau}}

\newcommand{\ACDlong}{\text{clock}}
\newcommand{\Exp}{\textup{Exp}}
\newcommand{\ccc}{c}

\newcommand{\NO}{\underset{\textup{N\!O}}{\xrightarrow{\hspace*{1cm}}}}

\newcommand{\TO}{\underset{\textup{T\!O}}{\xrightarrow{\hspace*{1cm}}}}

\newcommand{\vertiii}[1]{{\left\vert\kern-0.25ex\left\vert\kern-0.25ex\left\vert #1 
		\right\vert\kern-0.25ex\right\vert\kern-0.25ex\right\vert}}

\newcommand{\psidelta}{{\delta_\psi}}
\newcommand{\potdelta}{{\delta_V}}

\newcommand{\cP}{\mathcal{P}} 

\newcommand{\mpwno}[1]{#1}

\newcommand\mpwSo[1]{MW:{\let\helpcmd\sout\parhelp#1\par\relax\relax} }
\long\def\parhelp#1\par#2\relax{%
	\helpcmd{#1}\ifx\relax#2\else\par\parhelp#2\relax\fi%
}

\newcommand{\vep}{\varepsilon}

\newcommand{\cblack}{\color{black}}

\def\alphazero{{\alpha_{\min}}}
\def\alphamax{{\alpha_{\max}}}
\def\epzero{\ep_{0}}
\def\epres{\ep_{res}}
\def\epemb{\ep_{emb}}

\def\epsig{\ep_{H}} 
\def\epSycatG{{\varepsilon_{\Sy\cat\G}}} 
\def\epcl{{\varepsilon_{\cl}}}
\def\tauGibb{{\tau}_\G}
\def\dt{{\rm d\,}}
\def\fa{\tilde f_{\alpha}}
\def\hatf{\theta}
\def\bed{\begin{definition}}
\def\eed{\end{definition}}
\def\bel{\begin{lemma}}
\def\eel{\end{lemma}}
\def\bet{\begin{theorem}}
\def\eet{\end{theorem}}
\def\be{\begin{equation}}
\def\ee{\end{equation}}
\def\bei{\begin{itemize}}
\def\eei{\end{itemize}}
\def\ben{\begin{eqnarray}}
\def\een{\end{eqnarray}}
\def\bea{\begin{array}}
\def\eea{\end{array}}

\def\ot{\otimes}
\def\ep{\epsilon}

\def\<{\langle}
\def\>{\rangle}
\def\hcal{{\cal H}}

\begin{document}

\title{Autonomous quantum devices: When are they realisable without additional thermodynamic costs?}

\begin{abstract}
	The resource theory of quantum thermodynamics has been a very successful theory and has generated much follow up work in the community. It requires energy preserving unitary operations to be implemented over a system, bath, and catalyst as part of its paradigm. So far, such unitary operations have been considered a ``free'' resource in the theory. However, this is only an idealisation of a necessarily inexact process. Here, we include an additional auxiliary control system which can autonomously implement the unitary by turning  ``on/off'' an interaction. However, the control system will inevitably be degraded by the back-action caused by the implementation of the unitary. We derive conditions on the quality of the control device so that the laws of thermodynamics do not change; and prove | by utilising a good quantum clock | that the laws of quantum mechanics allow the back-reaction to be small enough so that these conditions are satisfiable. Our inclusion of non-idealised control into the resource framework also  raises interesting prospects, which were absent when considering idealised control. Among other things, the emergence of a 3rd law | without the need for the assumption of a light-cone.
	
	Our results and framework unify the field of autonomous thermal machines with the thermodynamic quantum resource theoretic one, and lay the groundwork for all quantum processing devices to be unified with fully autonomous machines.
\end{abstract}

\author{Mischa P. Woods}
\affiliation{Institute for Theoretical Physics, ETH Zurich, Switzerland}
\affiliation{University Grenoble Alpes, Inria, Grenoble, France}
\email{mischa.woods@gmail.com}
\author{Micha\l{} Horodecki}
\affiliation{International Centre for Theory of Quantum Technologies, University of Gdansk, Poland}
\email{michal.horodecki@ug.edu.pl}
\maketitle

\section{Introduction}

Thermodynamics has been tremendously successful in describing the world around us. It has also been at the heart of developing new technologies, such as heat engines which powered the industrial revolution, jet and space rocket propulsion | just to name a few. In more recent times, scientists have been developing a theoretical understanding of thermodynamics for tiny systems for which often quantum effects cannot be ignored. These ongoing developments are influential in optimising current quantum technologies, or understanding important physical processes. Take for example molecular machines or nano-machines such as molecular motors \cite{howard1997molecular}, which are important in biological processes \cite{molecularBio}, or distant technologies such as nanorobots \cite{nanobots}, where quantum effects on the control mechanism, and the back-reaction they incur, are likely to be significant due to their small size.\Mspace

\begin{figure}[ht!]
	\includegraphics[width=\linewidth]{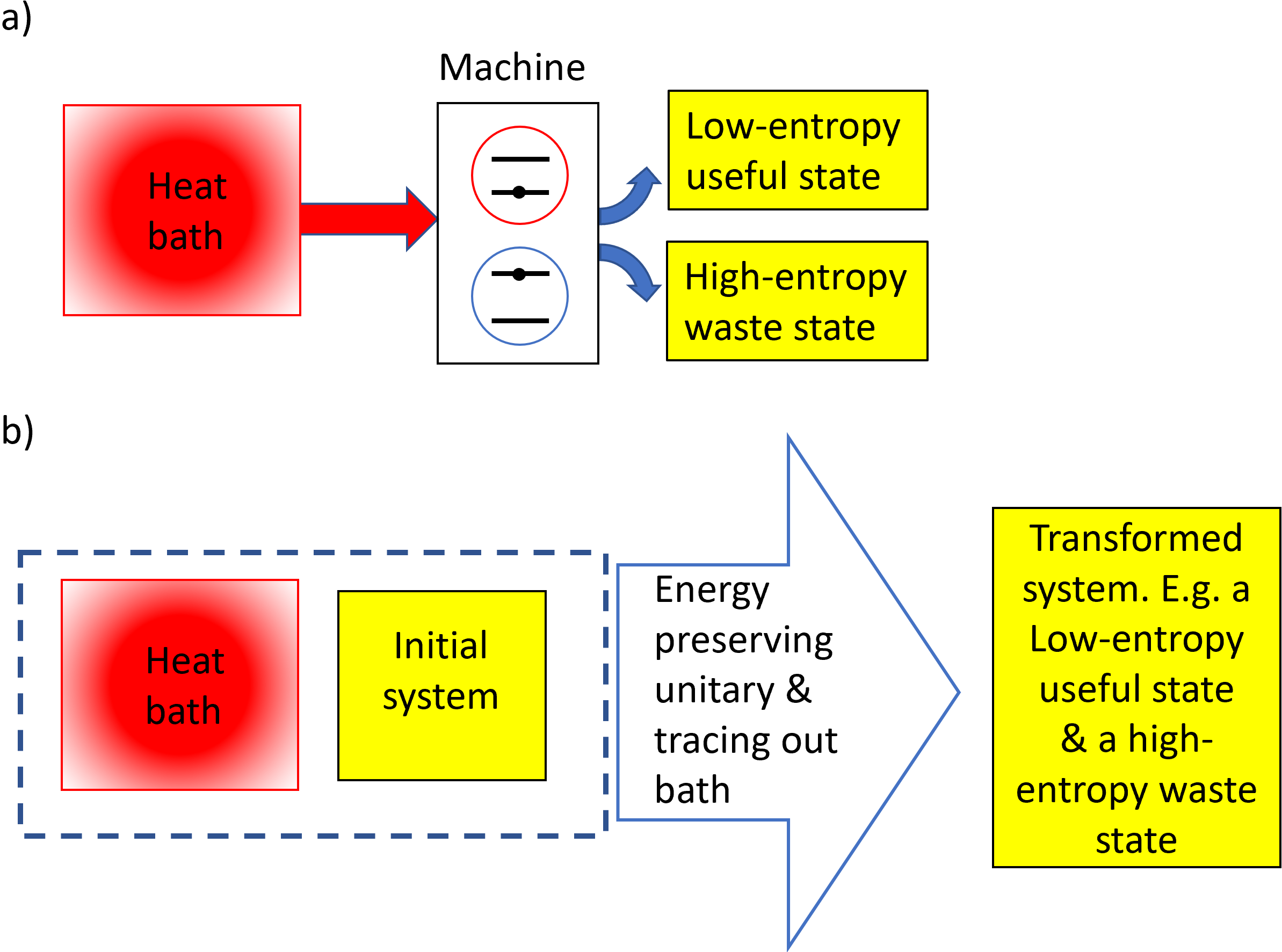}
	\caption{
	\textbf{Fully autonomous thermal machines vs. a type of non-autonomous cycle-based machines} \\{\bf a)} Depiction of a quantum thermal absorption machine: these devices do not need external control to operate (i.e. they are governed by time-independent Hamiltonians). Given enough time, they settle into a functioning steady state where heat from a heat bath is converted via a machine (composed of fine-tuned energy levels and couplings) into a low-entropy useful state (such as a charged battery) and a high-entropy ``waste'' state (such as a room temperature thermal state). See \cite{MitchisonReview,BENENTI20171} for reviews.
	{\bf b)}Schematic of a non-autonomous thermal machine.  In this resource-based framework~\cite{secondlaw}, an energy preserving unitary is performed over a heat bath and initial system state. The unitary is chosen so that the transformed system state is of high value (e.g. it could represent a charged battery). The control required to perform the energy-preserving unitary necessitates a time-dependent Hamiltonian and may not be thermodynamically cost-free. 
}\label{Fig:overview}
\end{figure}

The modern quantum thermodynamics literature tends to be about two types of processes: those which are fully autonomous (i.e. the processes described by time-independent Hamiltonians) and those which assume implicit external control at no extra cost (i.e. the processes described by  time-\emph{dependent} Hamiltonians). An example of processes described by a constant Hamiltonian  is the Brownian ratchet, popularised by Feynman \cite{brownianratchetfeynmann}, which simply sits between two thermal baths and extracts work in situ. There are many autonomous quantum thermal machines built on similar principles \cite{MitchisonReview,linden2010small,GAK-2013,brask2015autonomous,Mitchison_2015,Mitchison_2016,ExperimentAbsorptionFridge,AutoMachineTimeArrow2018,MarcusPauli,PhysRevX_Raam,Hewgill2019,AutonomousFridgeOnOff}. 
However, there are a number of processes, such as quantum Carnot cycles, that are described by time-dependent Hamiltonians  and thus require external control. This is true both in theory \cite{Alicki-1979,gelbwaser,Geusic1967quatum,MischaQHE,Ying_Ng_2017,PhysRevA.101.042116} and in 
experiment
\cite{SingleIon}. See Fig. \ref{Fig:overview} for a comparison of autonomous and non-autonomous processes.

The non-autonomous engines of the kind depicted in Fig. \ref{Fig:overview} require an external agent that
makes the changes. 
This does not happen in the engines used in our daily life. E.g. car engines do not require any external control | the passage via different strokes during the cycle is caused by suitable feedback mechanisms. 
An example of a thermal machine that requires switching between the strokes by an external agent  is the quantum heat engine of \cite{SingleIon}, where
alternating coupling to the hot and the cold bath is implemented by switching between two lasers | one producing thermal light of high temperature while the other one producing light at low temperature.

In this context the following problem appears. While the non-autonomous machines involve additional systems responsible for making the changes, those additional systems are by definition not considered explicitly. For  microscopic engines, 
such systems might actually be a place where a significant amount of entropy and/or energy is being deposited.
Such entropy production is actually likely to occur in microscopic regimes, due to the quantum back-reaction occurring between the controlling unit and the controlled system. There may be thus hidden thermodynamic costs, which are not accounted for.

Hence the following question can be posed: given a non-autonomous thermal machine, is it possible to provide an explicit control scheme, such that the overall (now autonomous) machine will exhibit no additional cost? 

This question is especially relevant in the context of the recently developed resource theory of thermodynamics \cite{janzing2000thermodynamic}, where any process is supposed to arise from the concatenation of basic operations which are energy preserving unitary transformations over a microscopic system of interest and a thermal bath. Thus we deal here with external control, represented by a time-dependent Hamiltonian 
that implements the subsequent 
unitaries. 
In such a microscopic regime, the hidden costs acquired by the control system may be indeed high, 
as is indicated by the phenomenon of so-called \emph{embezzling} \cite{Patrick_Embezzelment,Ng_2015} (see  Sec. \ref{Idealised control}). 

The problem of the cost of making the resource theoretic thermal machines autonomous 
was considered in \cite{Malabarba}.
 The control device was implemented by means of an idealized momentum clock.
Actually, any conceivable control system 
that enables one to 
go from a time-dependent Hamiltonian description to a time-independent one must  involve a clock as part of the control unit. I.e. a device for which the change  in  its state, due to time evolution, allows one to predict time. 

E.g. in a car engine the role of the clock is played by periodic motion of the piston  (arising via  so-called self-oscillation \cite{Jenkins-2013-self}); or in the already mentioned single-ion heat engine of \cite{SingleIon} the timing involved in the changing of the lasers is ultimately due to an external electronic device,  
which is a kind of clock. 


Unfortunately the clock used in \cite{Malabarba}  requires infinite energy.  It was first noted by Pauli that such clocks are unphysical \cite{pauli1958allgemeinen} and we will provide more weight to 
 Pauli's argument in this \doc.  
\Mspace

\begin{figure}[ht!]
	\includegraphics[width=0.7\linewidth]{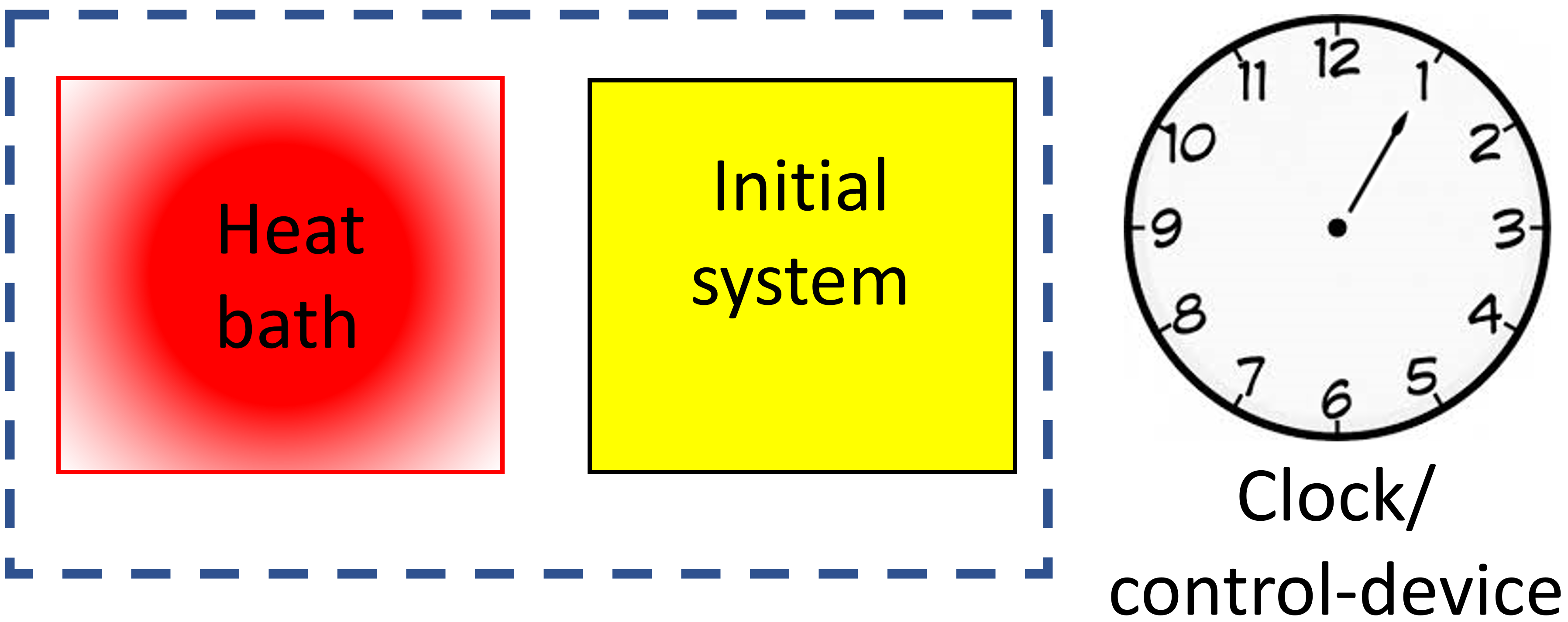}
	\caption{\textbf{Schematic of the autonomous quantum devices we focus on: A non-autonomous thermal machine complemented by a quantum clock.} The system in the dotted-line box is the same as in Fig. \ref{Fig:overview} b). It depicts the standard systems involved in the resource-theoretic approach to thermodynamics. If no other systems are involved, its dynamics are described by a time-dependent Hamiltonian. If one includes an additional quantum system whose state changes in a predictable fashion with the passing of time | i.e. a clock | then it can turn on and off interaction terms at specific times, leading to an autonomous implementation of the resource-theoretic approach to thermodynamics. Hence we use a clock as a control device.}\label{Fig:setup}
\end{figure}

In this paper the question of whether one can make the resource theoretic thermal machines autonomous, without incurring an extra thermodynamic cost
is reconsidered, and positively answered.


Namely, we start with non-autonomous scenario, where an external agent performs energy preserving
unitary on system plus bath. 
We will then examine the clock which turns on and off the interactions implementing the unitaries 
and derive conditions so that the change in the clock's state due to the back-reaction on it has a vanishingly small thermodynamic cost. We will then show that clocks exist which satisfy our criterion. In particular, we will find a family of clocks with different dimensions, for which there is no change in energy while the difference in entropy relative to before and after the unitary has been performed will be vanishingly small as the clock increases in size.
Importantly, since our clocks use finite energy, they avoid the issues of the clock of Ref. \cite{Malabarba}.
Our work thus demonstrates that the control needed to implement thermodynamic transformations in the resource-theoretic paradigm 
can indeed be neglected under certain achievable circumstances. 
In this way we show that non-autonomous resource theoretic thermal machines can be recast into autonomous ones, without additional cost. 


As a byproduct, our necessary conditions for the change in the clock to not have a significant additional thermodynamic cost 
reveal the emergence of a $3^\text{rd}$ law: if the clock implements the unitary too quickly (relative to the free dynamics of the system and clock), it will suffer a large back-reaction and \emph{will} represent a significant additional thermodynamic cost in addition to failing to implement correctly the required unitary.  The minimum time interval in which the unitary can be implemented without the clock suffering significant back-reaction is limited by the dimension of the clock. This demonstrates the emergence of a $3^\text{rd}$ law without the need to impose a light-cone or locality condition on how the unitary is implemented \cite{3rdLawLluis}. 
\Mspace

The rest of this \doc{} is divided into five main sections: Setting \ref{sec:Setting}, Results \ref{sec:Results}, Discussion \ref{sec:Discussion}, and Conclusions \ref{sec:Conclusions}. In Setting, we start by describing the thermodynamic transformations under the convention of idealised control.  This is summarized in definition  \ref{def:t-CTO}. 
Then in subsection \ref{Idealised control} 
we describe how to explicitly implement the control via  time-independent dynamics on the system of interest and an additional system called ``clock''. 
Finally, before moving to the Results section, we show why the cost of control can be counter-intuitive by showing how it is related to the established phenomenon of catalytic embezzlement and how idealised control requires infinite energy
 (see Prop. \ref{prop:idealised control}).  Our results discussed in Section \ref{sec:Results} start with the simplest case possible: the control of  so-called  noisy operations, in which baths are a source of entropy but not heat.  The result is quantified in Theorems \ref{thm:noemb physical} and \ref{Thm:Implementation with Quasi-Idela clock}.  The core of the Theorem \ref{thm:noemb physical} is what we can call ``no-embezzling conditions''. 
Namely, for the first time, we give a lower bound for the value of error on the catalyst that does not cause deviation 
from the $2^\text{nd}$ laws, i.e. from the limitations for transitions via noisy operations at zero error on the catalyst. 
We then move on to consider the full paradigm of control of thermodynamic operations in which the baths are a source of entropy and heat  | the so-called thermal operations.  This case is summarized in Theorems \ref{thm:noemb physical t CTO} and \ref{thm:4}.  In both cases  (i.e. noisy and thermal operations) we allow for catalysts and 
provide conditions under which the cost of control is neglectable. The situation is more nuanced in the case of thermal operations and has unforeseen consequences which we discuss.  Finally, in the last two sections, (Discussion \ref{sec:Discussion} and Conclusions \ref{sec:Conclusions}), we discuss in more detail the implications of our work followed by a summary.

The proofs of our results are given in the \App{} (section \ref{sec:appendix}). Additional technical details required for the proofs are relegated to the \Supp.

\section{Setting}\label{sec:Setting}

\subsection{Types of Thermodynamic Transformations}\label{sec:CTO def}
\subsubsection{Background: Thermal Operations and variants}
Resource theories have been applied to the study of quantum thermodynamics. In this setting, one considers transformations from a state $\rho_\A^0$ to $\rho_\A^1$ for which there exists a unitary $U_{\A\G}$ over system $\A$ and a Gibbs state $\tau_\G$ such that $\rho_\A^1= \tr_\G[U_{\A\G}\, (\rho_\A^0 \otimes \tau_\G)\, U_{\A\G}^\dag]$. This setup is entropy preserving since it is a unitary transformation. In order to call it a thermal operation (TO), we further require the process to be energy preserving, namely $[U_{\A\G}, \hat H_\A+ \hat H_\G]=0$, where $\hat H_\A$ is the local Hamiltonian of the $\A$ system and $\hat H_\G$ that of the thermal bath.\footnote{We often omit tensor products with the identity when adding operators on different spaces, e.g. $ \hat H_\A+ \hat H_\G\equiv \hat H_\A\otimes\id_\G+ \id_\A\otimes\hat H_\G$.} These operations can be extended to the strictly larger class of catalytic TOs (CTOs) by considering additional ``free'' objects called catalysts $\rho_\cat^0$. In this case the $\A$ system is bipartite with the requirement that the catalyst is returned to its initial state after the transformation;  $\rho_\Sy^1\otimes \rho_\cat^0 = \tr_\G[U_{\Sy\cat\G}\,( \rho_\Sy^0\otimes \rho_\cat^0 \otimes \tau_\G)\, U_{\Sy\cat\G}^\dag]$, with a Hamiltonian $\hat H_\A$ of the form $\hat H_\Sy+\hat H_\cat$. The bath provides a source of entropy and heat. In the special case in which its Hamiltonian is completely degenerate, its Gibbs state $\tau_\G$ becomes the maximally mixed state $\tau_\G\propto\id_\G$ and the bath can now only provide entropy. These are known as catalytic noisy operations (CNOs), or simply noisy operations (NOs) when there is no catalyst involved \cite{PhysRevA.67.062104,QmapsDuality}.
It is known that CNOs allow for transitions that are not possible by NOs \cite{Jonathan_1999,gour2015resource}.
\Mspace

In these frameworks, the operations (NOs, CNOs, TOs, CTOs) are considered to be \emph{free} from the resource perspective, since they preserve entropy and energy over system $\A$ and the bath $\G$ | the two resources in thermodynamics. However, note that there is the assumption that the external control (i.e. the ability to apply energy preserving unitaries over the setup) is ``perfect''. In order to challenge this perspective, we will now introduce an auxiliary system to represent explicitly the system which implements the external control, while aiming to show to what extent it can be free, from the resource theory perspective.\Mspace

\subsubsection{t-Catalytic Thermal Operations}

If the control system is a thermodynamically free resource, its final state after the transition must be as useful as the state it would have been in had it not implemented the unitary, and instead evolved unitarily according to its free Hamiltonian. One way to realise this within the resource theoretic paradigm, is to choose a control device whose free evolution is periodic and let the time taken to apply the unitary be an integer multiple of its period. In this scenario the control device fits nicely within the resource theory framework, since when viewed at integer multiples of the period, the control device is a catalyst according to CTOs.

The downside with this approach is that the times corresponding to multiples of the period are a measure zero of all possible times. Consequently, not only would one need an idealised clock which can tell the time with zero uncertainty to discern these particular times, but one would like to be able to say whether the transition was thermodynamically allowed during proper intervals of time. Fortunately, there is a simple generalisation of CTOs,\footnote{Note that this generalisation also generalises NOs, CNOs and CTOs by allowing for the inclusion of a catalyst in the initial and final state of the transition and or specialising to the case of a maximally mixed Gibbs state.} which naturally resolves this issue. We introduce t-CTO which take into account that the transition is not instantaneous, but moreover occurs over a finite time interval. In the following definition, one should think of the catalyst system as playing the role of the external control device. 


\begin{definition}[t-CTO and t-CNO]\label{def:t-CTO}
	A transition from $\rho_\Sy^0(t_1)$ to $\rho_\Sy^1(t_2)$ with $t_1\leq t_2$ is possible under t-CTO iff there exists a finite dimensional quantum state $\rho_\cat$ with Hamiltonian $\hat H_\cat$ such that
	\begin{equation}
	\rho_\Sy^0(0)\otimes \rho_\cat^0(0) \TO \bar\sigma_\Sy(t) \otimes \rho_\cat^0(t),
	\end{equation}
	where
	\be
	\bar\sigma_\Sy(t)=\begin{cases}
		\rho_\Sy^0(t) &\mbox{ if } t\in[0,t_1]\\
		\rho_\Sy^1(t) &\mbox{ if } t\in[t_2,t_3]\label{eq:sigma bar def}
	\end{cases}
	\ee
	 $\rho_\D^n(t):=\me^{-\mi t \hat H_\D} \rho_\D^n\, \me^{\mi t \hat H_\D}$, $\D\in\{\Sy,\cat\}, n\in\{0,1\}$, and $t_1$ is called ``the time when the TO began'' while $t_2$ ``the time at which the TO was finalised''. $[0,t_1)$ and $(t_2,t_3]$ are both proper intervals called ``the time before the TO began'' and ``the time after the TO was finalised'' respectiverly. In the special cases where the bath can only be maximally mixed, $\tau_\G\propto\id_\G$, it will be denoted $\tauT_\G$ and we will call the transition a t-CNO.
\end{definition}
Unless stated otherwise, we will always use the notation $\rho_\D^n(t)$, $n\in\{0,1\}$, to denote the free evolution of a normalised quantum state $\rho_\D^n$ on some Hilbert space $\mathcal{H}_\D$ according to its free Hamiltonian $\hat H_\D$; namely $\rho_\D^n(t)=\me^{-\mi t \hat H_\D} \rho_\D^n\, \me^{\mi t \hat H_\D}$. \Mspace

Definition \ref{def:t-CTO} captures two notions. On the one hand that the individual subsystems are effectively non interacting before and after the transition has taken place. On the other hand, that during the time interval $(t_1, t_2)$, in which the transition occurs, arbitrarily strong interactions could be realised. Note that there are two special cases for which t-CTOs reduce to CTOs at times $t_1,t_2$ | when the Hamiltonian of the catalyst is trivial, (i.e. if $\hat H_\cat\propto \id_\cat$), and when the catalyst is periodic with $t_1, t_2$ integer multiples of its period $T_0$ (i.e. if $\rho_\cat^0( t_1)=\rho_\cat^0( t_2)= \rho_\cat^0(T_0)$\,).\Mspace

From the resource theoretic perspective, the characterisation of t-CTOs is the same as CTOs as the following proposition shows.
\begin{prop}[t-CTO \& CTO operational equivalence]\label{prop:equiv of t-CNO and CNO} 
	A t-CTO from $\rho_\Sy^0(t_1)$ to $\rho_\Sy^1(t_2)$ using a catalyst $\rho_\cat^0(0)$, exists iff a CTO from $\rho_\Sy^0$ to $\rho_\Sy^1$ exists using catalyst $\rho_\cat^0(0)$. In other words 
	\begin{equation}\label{eq:t-CNO equiv CNO 1}
	\rho_\Sy^0(0)\otimes \rho_\cat^0(0) \TO \bar\sigma_\Sy(t) \otimes \rho_\cat^0(t), 
	\end{equation}
	where $\bar\sigma_\Sy(t)$ is defined in Eq. \eqref{eq:sigma bar def}, if and only if
	\begin{equation}\label{eq:t-CNO equiv CNO 2}
	\rho_\Sy^0\otimes \rho_\cat^0(0) \TO \rho_\Sy^1 \otimes \rho_\cat^0(0).
	\end{equation}
\end{prop}
\begin{proof} 
	It is simple. For $t\in[0, t_1]$, Eq. \eqref{eq:t-CNO equiv CNO 1} always holds since the l.h.s. and r.h.s. only differ by an energy preserving unitary on the catalyst, which is a valid TO. Therefore the only non-trivial instance of  Eq. \eqref{eq:t-CNO equiv CNO 1} is for $t\in [t_2,t_3]$. Let us now compare Eqs. \eqref{eq:t-CNO equiv CNO 1} and \eqref{eq:t-CNO equiv CNO 2} for $t\in[t_2,t_3]$\,: the only difference is an energy preserving unitary transformation on the catalyst state on the r.h.s. However, all energy preserving unitary translations are TOs. Therefore one can always go from the r.h.s. of Eq. \eqref{eq:t-CNO equiv CNO 2} to the r.h.s. of Eq. \eqref{eq:t-CNO equiv CNO 1} via a TO. This proves the ``if'' part of the proposition. Conversely, since the inverse of an energy preserving unitary is another energy preserving unitary, one can always go from the r.h.s. of Eq. \eqref{eq:t-CNO equiv CNO 1} to the r.h.s. of Eq. \eqref{eq:t-CNO equiv CNO 2} via a TO.
\end{proof}

While the generalisation to t-CTOs is admittedly quite trivial in nature, it is nevertheless important when considering the autonomous implementation of CTOs. So far, the t-CTOs have only allowed us to include the external control mechanism explicitly into the CTOs paradigm in such a way that they constitute a free resource. In the next section, we will see how this free resource unfortunately corresponds to unphysical time evolution governed by an idealised clock. It will however set the benchmark for what we should be aiming to achieve, if only approximately, with a more realistic control device.

\subsection{Idealised Control, Clocks and Embezzling Catalysts}\label{Idealised control}

When a dynamical catalyst in a t-CTO is responsible for autonomously implementing the transition, it must have its own internal notion of time in order to implement the unitary between times $t_1$ and $t_2$. While in practice, the clock part may only form a small part of the full dynamical catalyst system, for convenience of expression, we refer to such dynamical catalysts as a clock and denote the state of the clock with the subscript $\cl$. Specifically, we would require the clock to induce dynamics on a system $\A$ which corresponds to a t-CTO on $\A$. In other words, evolution of the form 
$\rho^F_{\A\cl\G}(t)= \me^{-\mi t \hat H_{\A\cl\G}} \left(\rho_\A^0\otimes\rho_\cl^0\otimes\tau_\G\right) \me^{\mi t \hat H_{\A\cl\G}}$ where $\rho^F_{\A\cl\G}(t)$ satisfies~\footnote{For any bipartite state $\rho_{\A\B}$, we use the notation of reduced states $\rho_{\A}:=\tr_\textup{B}(\rho_{\A\B})$, $\rho_{\B}:=\tr_\textup{A}(\rho_{\B\A})$.} 
\begin{align}\label{eq:idealised control}
\rho^F_{\A\cl}(t) =\rho_\A^F(t)\otimes\rho_\cl^0(t),\quad \rho_\A^F(t)
= \begin{cases}
\rho_\A^0(t) \mbox{ if } t \in[0, t_1]\\
\rho_\A^1(t) \mbox{ if } t \in [t_2,t_3]
\end{cases}
\end{align}
Here $\rho_\cl^0(t)$ denotes the free evolution of the clock,
\begin{align}
\rho_\cl^0(t)= \me^{-\mi t\hat H_\cl} \rho_\cl^0\, \me^{\mi t\hat H_\cl}.\label{eq:def free clock ev}
\end{align}
In the case in which the clock aims to implement autonomously a TO, we would have that the r.h.s. of Eq. \eqref{eq:idealised control} satisfies $\rho_\A^0(t)= \rho_\Sy^0(t)$ and $\rho_\A^1(t)= \rho_\Sy^1(t)$, while in the case of a CTO, $\rho_\A^0(t)= \rho_\Sy^0(t)\otimes \rho_\cat^0(t)$ and $\rho_\A^1(t)= \rho_\Sy^1(t)\otimes \rho_\cat^0(t)$. 
In this latter case, we see that we have two catalysts. The 1st one, $\rho_\cat^0$ simply allows for a transition on $\Sy$ which would otherwise be forbidden under TOs, while the second one, $\rho_\cl^0$ is the clock which implements autonomously the transition. Furthermore, note that while the r.h.s. of Eq. \eqref{eq:idealised control} is evolving according to the free Hamiltonian $\hat H_\A+\hat H_\cl$, the Hamiltonian $\hat H_{\A\cl\G}$ can, in principle, be of any form such that Eq. \eqref{eq:idealised control} holds.\Mspace

The following rules out the possibility of dynamics of the form Eq. \eqref{eq:idealised control} for a wide class of clock Hamiltonians even when Eq. \eqref{eq:idealised control} is relaxed to include correlations between system $\A$ and the clock.

\begin{prop}[Idealised Control No-Go]\label{prop:idealised control}
	Consider a time independent Hamiltonian $\hat H_{\A\cl\G}$ on $\mathcal{H}_{\A\G}\otimes \mathcal{H}_\cl$ where $\mathcal{H}_{\A\G}$ is finite dimensional, and $\mathcal{H}_\cl$ arbitrary; which w.l.o.g. we expand in the form $\hat H_{\A\cl\G}=\hat H_{\A\G}\otimes\id_\cl+\sum_{l,m=1}^{d_\A d_\G} \ketbra{E_l}{E_m}_{\A\G}\otimes \hat H_\cl^{(l,m)}$, where $\{\,\ket{E_l}_{\A\G}\}_{l=1}^{d_\A d_\G}$ are the energy eigenstates of $\hat H_{\A\G}=\hat H_\A+\hat H_\G$; the free Hamiltonian on $\mathcal{H}_\A$ and the bath. Both of the following two assertions cannot simultaneously hold:
\begin{itemize}
\item [\bf{1)}]
	For all $k,l=1,2,\ldots,d_\A d_\G$; $k\neq l$, the power series expansion in $t$
\begin{align}
&\tr\left[\me^{-\mi t \hat H_\cl^{(k,k)}} \!\rho_\cl^0\, \me^{\mi t \hat H_\cl^{(l,l)}} \right]\label{line 1 of prop 3}\\
&=\sum_{n,m=0}^\infty \tr\left[\frac{\big(-\mi  \hat H_\cl^{(k,k)}\big)^n}{n!} \rho_\cl^0\, \frac {\big(\mi\hat H_\cl^{(l,l)}\big)^m}{m!} \right] t^{n+m} \!\!\!\!\!\!
\label{eq:convergent series} 
\end{align}
has a radius of convergence $r>t_2$. 
\item [\bf{2)}] For some $0<t_1<t_2<t_3$ there exists a TO from $\rho_\A^0(t)$ to 
\begin{align}
\;\;\;\;\;\rho_\A^F(t)=\begin{cases}
\rho_\A^0(t) &\mbox { for }  t\in[0,t_1]\\
\tr_\G[U_{\A\G}(\rho_{\A}^0(t)\otimes\tau_\G )U_{\A\G}^\dag] &\mbox { for }  t\in[t_2,t_3],\!\!\!\!\!\!\!\label{eq:cases Prop idealised control}
\end{cases}
\end{align}
which is implementable via unitary dynamics of the form
\begin{align}
\;\,\rho^F_{\A}(t)= \tr_{\G\cl}\left[\me^{-\mi t \hat H_{\A\cl\G}} \left(\rho_\A^0\otimes\rho_\cl^0\otimes\tau_\G\right) \me^{\mi t \hat H_{\A\cl\G}}\right]\!,\!\!\!\!\label{eq:rho F A t}
\end{align}
where $U_{\A\G}$ in Eq. \ref{eq:cases Prop idealised control} has non-degenerate spectrum, and is an energy preserving unitary, namely $[U_{\A\G}\otimes\id_\cl, \hat H_{\A\G\cl}]=0$.
\end{itemize}	

\end{prop}

See Sec. \ref{sec:prof of lemma Idealisec control} for a proof by contradiction. The requirement of non degenerate spectrum in 2) for $U_{\A\G}$ allows for exclusion of the trivial cases $U_{\A\G}\propto \id_{\A\G}$ for which 1) and 2) can simultaneously hold\footnote{It is likely that the no-go theorem holds for all non-trivial $U_{\A\G}$, i.e. all cases for which there exists $t\in[0,t_1]\cup[t_2,t_3]$ such that $\rho_\A^F(t)\neq \rho_\A^0(t)$. However, the point of the no-go theorem is simply to show that the problem is non-trivial for most cases of interest.}. Furthermore, the no-go proposition also covers the more relaxed setting in which the clock (or any catalyst included in $\A$) is allowed to become correlated with the system. The correlated scenario is also important and studied within the context of idealised control in \cite{Muller2017,PhysRevA.97.062114,PhysRevLett.122.210402}.\Mspace

Physical intuition suggests that if the Hamiltonian $\hat H_{\A\cl\G}$ is infinite dimensional, the dynamics it induces can be arbitrarily well approximated by replacing it by a projection onto an arbitrarily large finite dimensional subspace. However, such a projection would imply that the terms $\hat H_\cl^{(l,m)}$ found in the Hamiltonian $\hat H_{\A\cl\G}$ are replaced with finite dimensional matrices and the series in line \ref{line 1 of prop 3} would converge. Therefore, according to the above proposition if 1) holds, the Hamiltonian $\hat H_{\A\cl\G}$ cannot be approximated as one would expect.

On the other hand, 2) includes the desirable scenario of idealised control discussed at the beginning of Sec. \ref{Idealised control}. Therefore, the no-go proposition tells us that if idealised control is possible, it requires infinite dimensional Hamiltonians which cannot be approximated in the way one might expect.\Mspace

It can also be seen that the contradicting statements 1) and 2) in Prop. \ref{prop:idealised control} are not due to a necessity to implement 2) with ``abruptly changing'' dynamics, since the unitary $U_{\A\G}$ facilitating the TO from $\rho_\A^0$ to $\rho_\A^1$ can be implemented via a smooth function of $t$; namely $U_{\A\G}(t)=\exp\big[-\mi \hat H_u\int_{t_1}^t\bar\delta(x)dx \big]$, with $\bar \delta(t)$ a normalised bump function with support on some interval $\subseteq[t_1,t_2]$ and $\hat H_u$ an appropriately chosen time independent Hamiltonian.\Mspace

The no-go proposition thus rules out physical implementation of idealised control for a number of cases. We now give some examples in which 1) or 2) hold. Prop. \ref{prop:idealised control}, case 1) holds when $\rho_\cl^0$ is an analytic vector \cite{ReedSimon}. The simplest examples of this is when $\rho_\cl^0$ has bounded support on the spectral measures of the Hamiltonians $\{\hat H_\cl^{(k,k)}\}_{k=1}^{d_\A d_\cl}$, such as in the finite dimensional clock case. One can however find examples for Prop. \ref{prop:idealised control} in which 2) is fulfilled while 1) is not. This corresponds to the case of the idealised momentum clock used for control in \cite{Malabarba}. In this case the Hamiltonian $\hat H_{\A\G\cl}$ from Prop. \ref{prop:idealised control} can be written in the form $\hat H_{\A\G\cl}= \hat H_{\A\G}\otimes\id_\cl+ \sum_{n=1}^{d_\A d_\G}\Omega_n \proj{E_n}_{\A\G}\otimes g(\hat x_\cl)+\id_{\A\G}\otimes \hat p_\cl$, with $\hat x_\cl$, $\hat p_\cl$ canonical position and momentum operators of a particle on a line. When $g$ and the initial clock state have bounded support in position, 2) in Prop. \ref{prop:idealised control} is satisfied, but 1) is not. Unfortunately such a clock state is so spread out in momentum, the power series expansion $\Exp[-\mi t\hat p_\cl]=\sum_{n=0}^\infty (-\mi t \hat p_\cl)^n/{n!}$ diverges in norm when evaluated on it. This is closely related to another unphysical property of such clock states, namely that the Hamiltonian has no ground state, as 1st pointed out by Pauli \cite{pauli1958allgemeinen}. We will also see how this idealised control allows one to violate the 3rd law or thermodynamics in Sec. \ref{sec: main text CTO case} | something which should not be possible with control coming from a physical system. We will thus refer to dynamics for which $\rho_{\A\cl}^F(t)$ satisfies Eq. \eqref{eq:idealised control} as \emph{idealised dynamics}.\Mspace

\noindent\setlength\fboxsep{0.3cm}\fbox{%
	\parbox{0.92\linewidth}{%
		\emph{Take home message from Proposition \ref{prop:idealised control}}: Control devices which do not suffer any back-reaction when implementing a thermodynamic transition, arguably necessitate unphysical Hamiltonians.
	}%
}
\vspace{0.2cm}

At first sight, these observations may appear to be of little practical relevance, since indeed, one does not care about implementing the transition from $\rho_\Sy^0$ to $\rho_\Sy^1$ exactly, but only to a good approximation. Furthermore, for a sufficiently large clock, one might reasonably envisage being able to implement all transformations whose final states $\rho_\Sy^F(t)$ are in an epsilon ball of those reachable under t-CNO (and not a larger set) to arbitrary small epsilon as long as the final clock state becomes arbitrarily close in trace distance to the idealised case, namely if $\| \rho^F_{\cl}(t) - \rho_\cl^0(t)\|_1$ tends to zero as the dimension of the clock becomes large and approaches an idealised clock of infinite energy. Unfortunately, this intuitive reasoning may be false due to a phenomenon known as embezzlement. Indeed, when Eq. \eqref{eq:idealised control} is not satisfied the clock is disturbed by the act of implementing the unitary. As such, it is no longer a catalyst, but only an inexact one. Inexact catalysis has been studied in the literature with some counter intuitive findings. In \cite{Ng_2015} an inexact catalysis pair $\rho_\cat^0$, $\rho_\cat^1$ of dimension $d_\cat$ were found such that for any $d_\Sy$ dimensional system, their trace distance vanished in the large $d_\cat$ limit:
\begin{align}\label{eq:cat distance Nelly example}
\| \rho_\cat^0 - \rho_\cat^1 \|_1 = \frac{d_\Sy}{1+(d_\Sy-1) \log_{d_\Sy} d_\cat}.
\end{align} 
Yet the noisy operation $\rho_\Sy^0 \otimes \rho_\cat^0 \NO \rho_\Sy^1 \otimes \rho_\cat^1$ becomes valid for \emph{all} states $\rho_\Sy^0$, $\rho_\Sy^1$ in the large $d_\cat$ limit. In other words, they showed that the actual transition laws for the achievable state $\rho_\Sy^1$ given an initial state $\rho_\Sy^0$ cannot be approximated by those of CNOs | they are completely trivial, since all transformations are allowed. This paradoxical phenomena is known as work embezzlement\footnote{This is because in order for all states in the system Hilbert space to be reachable by an initial state under CNOs, the initial state needs to be supplemented with a work bit which is depleted in the process.} and stems from the concept of entanglement embezzlement \cite{Patrick_Embezzelment}.\Mspace

\begin{figure}[h!]
	\includegraphics[width=0.9\linewidth]{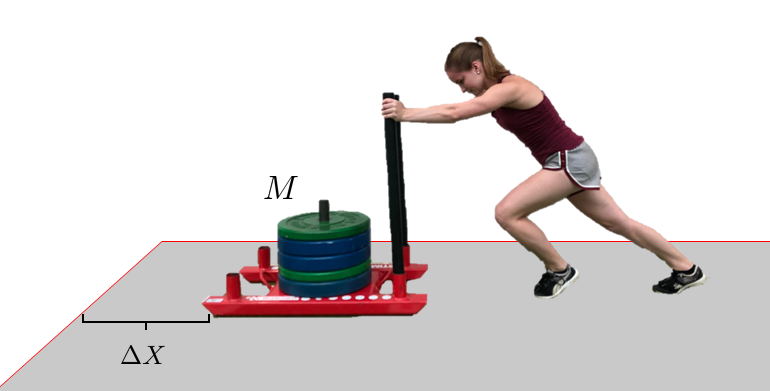}
	\caption{\textbf{The counter intuitive phenomenon of embezzlement.} Consider a thought experiment in which an athlete who has to push a mass $M$ a distance $\Delta X$ against a resistive force $F=M  g$ due to gravity pushing down on the weight. Suppose the distance the athlete has to push the weight is given by $\Delta X= f(M)$, where $f(M)\rightarrow 0$ as $M \rightarrow \infty$. The work done by the athlete pushing the weight is $W= \mu_0 F \Delta X=\mu_0 g M f(M)$, for some coefficient of resistance $\mu_0$.  One might be inclined to reason that the amount of work the athlete has to do in the limit of infinite mass $M$ is zero, since the distance $\Delta X$ the weight has to be pushed is zero in this limit. 
		However, a closer analysis would reveal that this is only correct if $f(M)$ decays sufficiently quickly | quicker than an inverse power. An analogous phenomenon is at play in our control setting. There, in the case of the idealised clock, Eq. \eqref{eq:idealised control} holds, yet this is unachievable since it requires infinite energy. However, all finite clocks, suffer a minimal back-reaction and even though this back-reaction can vanish in the large dimension/energy limit (c.f. Eq. \ref{eq:cat distance Nelly example}), this is not sufficient to conclude that the set of implementable transformations are close to those implementable via the idealised clock. Moreover, the rate at which the error needs to vanish, and whether this is physically achievable; were (prior to this work) completely unknown.}\label{Fig:ClssicalEmbezzExaple}
\end{figure}

By virtue of Prop \ref{prop:equiv of t-CNO and CNO}, the above example shows that simply finding a clock satisfying  $\| \rho^F_{\cl}(t) - \rho_\cl^0(t)\|_1\rightarrow 0$ as $d_\cl \rightarrow \infty$ is \emph{not} sufficient to conclude that the set of allowed transformations generated by t-CNOs (and thus CTOs) corresponds to the set of transformations which can actually be implemented with physical control systems. A thought experiment illustrating such phenomena can be found at the classical level in Fig. \ref{Fig:ClssicalEmbezzExaple}.

\section{Results}\label{sec:Results}

We will start with the easier case of CNOs in Sec. \ref{sec:maint text: Auto CNOs} before moving on to the more demanding setting of CTOs in Sec. \ref{sec: main text CTO case}.

\subsection{Autonomous control for Catalytic Noisy Operations}\label{sec:maint text: Auto CNOs}

In this section we will provide two theorems which together show that there exist clocks which are sufficiently accurate to allow the full realisation of t-CNOs to arbitrarily high precision. Our first result will give a sufficient condition on the clock so as to be guaranteed that the achieved dynamics of the system are close to a transition permitted under t-CNOs. It can be viewed as a converse theorem to the result in \cite{Ng_2015} discussed at the end of Sec. \ref{Idealised control}.

In the following theorem, let $V_{\Sy\cat\cl\G}(t)=\me^{-\mi t \hat H_{\Sy\cat\cl\G}}$ be an arbitrary unitary implemented via a time \emph{independent} Hamiltonian $\hat H_{\Sy\cat\cl\G}$, over $\rho_\Sy^0\otimes \rho_\cat^0\otimes\rho_\cl^0\otimes\tilde\tau_\G$ and suppose that the final state	at time $t\geq 0$,
\begin{equation}
\rho^F_{\Sy\cat\cl\G}(t)=V_{\Sy\cat\cl\G}(t) \left(\rho_\Sy^0\otimes \rho_\cat^0\otimes\rho_\cl^0\otimes\tilde\tau_\G\right)\!V_{\Sy\cat\cl\G}^\dag(t)
\end{equation}
deviates from the idealised dynamics by an amount
\begin{equation}
\| \rho^F_{\Sy\cat\cl}(t)- \rho^F_{\Sy}(t)\otimes \rho_\cat^0(t)\otimes\rho_\cl^0(t) \|_1 \leq \epemb(t; d_\Sy,  d_\cat d_\cl),\label{eq:ep emb def}
\end{equation} 
where recall $\rho_\cl^0(t)$ is the free evolution of the clock according to its free, time independent, Hamiltonian $\hat H_\cl$ (Eq. \ref{eq:def free clock ev}) and likewise for $\rho_\cat^0(t)$ with arbitrary Hamiltonian $\hat H_\cat$.

\begin{theorem}[Sufficient conditions for t-CNOs]\label{thm:noemb physical} For all states $\rho_\Sy^0$ not of full rank, and for all catalysts $\rho_\cat^0$, clocks $\rho_\cl^0$ and maximally mixed states $\tilde\tau_\G$, 
there exists a state $\sigma_\Sy(t)$ which is $\epres$  close to $\rho_\Sy^F(t)$,
	\begin{equation}
	\|\sigma_\Sy(t)-\rho_\Sy^F(t)\|_1 \leq \epres\left(d_\Sy,  d_\cat d_\cl,\epemb(t; d_\Sy, d_\cat d_\cl)\right),
	\end{equation}
	such that for all times $t\geq 0$, a transition from 
	\begin{equation}\rho_\Sy^0\otimes\rho_\cat^0\otimes\rho_\cl^0 \;\;\;\text{ to }\;\;\; \sigma_\Sy(t)\otimes\rho_\cat^0(t)\otimes\rho_\cl^0(t)\label{eq:sigma is valid NO}
	\end{equation}
	is possible via a NO (i.e. $\rho_\Sy^0$ to $\sigma_\Sy(t)$ via t-CNO). 
Specifically, for fixed $d_\Sy$ and in the limit that $d_\cat d_\cl$ and  $1/\epemb$ tend to infinity\textup{:}
	\ba\label{eq:epres func def}
\begin{split}
	\epres&(d_\Sy, d_\cat d_\cl,\epemb)=\\
	&15\sqrt{ \frac{d_\Sy  \ln(d_\cat d_\cl)}{\ln\left(1/\epemb\right) } \left(1 + d_\cat d_\cl \epemb^{1/7} \right)}.
\end{split}
	\ea
	Explicitly, one possible choice for $\sigma_\Sy(t)$ is 
	\ba
	&\sigma_\Sy(t) =\nonumber \\
	&\,\,\left\{
	\begin{array}{ll}
		\id_\Sy/d_\Sy &\text{\rm if } \|\rho_\Sy^{F}(t) -\id_\Sy/d_\Sy\|_1 < \epres \\
		(1\!-\epres) \rho_\Sy^{F}\! (t)+ \epres \id_\Sy/d_\Sy &\text{\rm if } \|\rho_\Sy^{F}(t) -\id_\Sy/d_\Sy\|_1 \geq \epres \\
	\end{array}
	\right.\nonumber
	\ea
\end{theorem}
See \App~\ref{Main proof section} for a proof and an expression for $\epres$ which holds when $d_\cat d_\cl$ and $\epemb$ are finite. Note that this theorem also holds more generally if one replaces ${\hat H_{\Sy\cat\cl\G}}$ with any time dependent Hamiltonian. However, the time independent Hamiltonian case is better physically motivated. \Mspace

Before we move on, let us understand the physical meaning of the terms $\epemb,\epres$. By comparing the definition of $\epemb$ in Eq. \eqref{eq:ep emb def} with that of Eq. \eqref{eq:idealised control}, we see that it is the difference in trace distance between the dynamics achieved with the idealised clock, and the actual dynamics achieved by the clock. Thus the quantity $\epemb$ upper bounds how much one can embezzle from the resulting unavoidable inexact catalysis of the \ACDlong. Then $\epres$ (which is a function of $\epemb$) characterises the resolution, i.e. how far from a t-CNO transition one can achieve due to embezzlement from the inexact catalysis.
For example, consider a hypothetical clock for which $\epemb$ decays as an inverse power with $d_\cl$. Then, $\epres$ would diverge with increasing $d_\cl$ and Theorem \ref{thm:noemb physical} would not tell us anything useful. On the other hand, if we had a more precise clock with, for example, $\epemb$ exponentially small in $d_\cl$ then Theorem \ref{thm:noemb physical} would tell us that $\epres$ converges to zero as $d_\cl$ increases.\Mspace

\noindent\fbox{%
	\parbox{0.92\linewidth}{%
		\emph{Take home message from Theorem \ref{thm:noemb physical}}: There is a threshold on the amount of back-reaction the control system can incur, above which the laws of thermodynamics have to be modified to include the thermodynamics of the control system. Theorem \ref{thm:noemb physical} provides a bound on this threshold when the bath transfers entropy but not heat.
	}%
}
\vspace{0.2cm}
\begin{figure}[]
	\includegraphics[width=1.\linewidth]{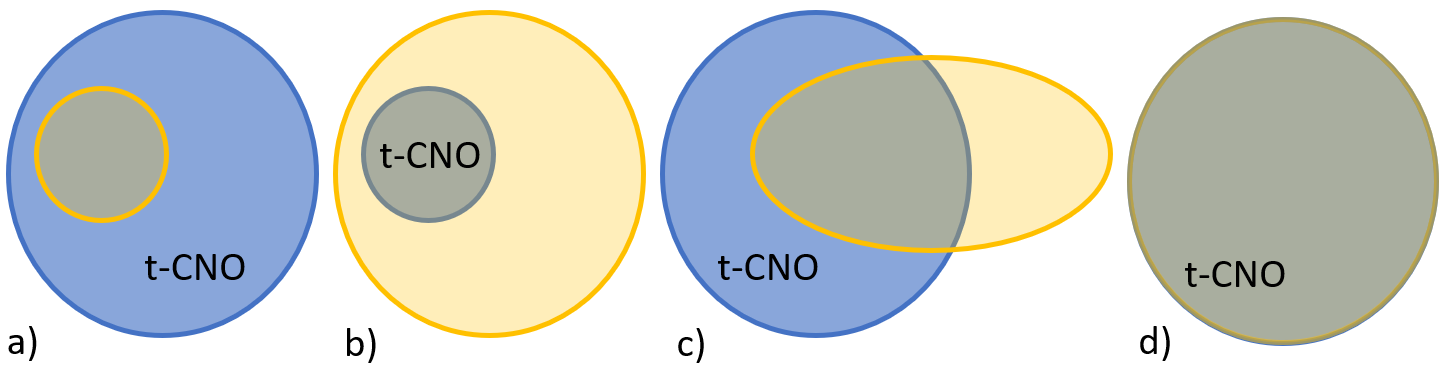}
	\caption{\textbf{Possible scenarios resulting from the physical implementation of t-CNOs.} Given a state $\rho_\Sy^0$, the above blue Venn diagrams represent the set of states $\rho_\Sy^1$ which can be reached under t-CNOs. The orange Venn diagrams represent possible scenarios of reachable states when attempting to implement a t-CNO, while grey represents the intersection of the two sets. Due to the apparent impossibility of perfect control and that embezzlement can occur (see Eq. \eqref{eq:cat distance Nelly example}), all options a) to d) are in principle open. Theorem \ref{thm:noemb physical} gives sufficient conditions on the control (clock) so that either a) or d) occur. Theorem \ref{Thm:Implementation with Quasi-Idela clock} shows that transitions implemented via the Quasi-Ideal clock can achieve d) under reasonable circumstances.}\label{Fig:VennDiag}
\end{figure}


Whether $\epemb$ and $\epres$ can both be simultaneously small, depends on both the quality of the clock used and the transition one wishes to implement. Two examples at opposite extremes are as follows. Both $\epemb$ and $\epres$ are trivially arbitrarily small (zero in fact), and conditions in Theorem \ref{thm:noemb physical} are satisfied, when the t-CTO transition is the identity transition (i.e. $\rho_\Sy^0$ to $\rho_\Sy^0$). At the opposite extreme, both $\epemb$ and $\epres$ cannot be small or vanishing when one attempts a non-trivial t-CNO transition which occurs instantaneously, i.e. one for which $\rho_\Sy^F(t)=\rho_\Sy^0$ for $t\in [0, t_1]$ and $\rho_\Sy^F(t)=\rho_\Sy^1$ for $t\in ( t_1, t_3]$.\Mspace

Our next theorem shows how one can implement to arbitrary approximation all t-CTO transitions, over any fixed time interval $(t_1,t_2)$, yet without allowing for a larger class | as the examples in Eq. \eqref{eq:cat distance Nelly example} and Fig. \ref{Fig:VennDiag} b)  do. To achieve this, one must choose the time independent Hamiltonian $ \hat H_{\Sy\cat\cl\G}$ and initial clock state $\rho_\cl^0$ appropriately. The theorem will use the \textit{Quasi-Ideal clock} \cite{WSO} discussed in detail in Sec. \ref{Overview of the quasi-ideal clock} for the clock system on $\mathcal{H}_\cl$. The Quasi-Ideal clock has been proven to be optimal for some tasks related to reference frames \cite{woods_continuous_2019,Sephir_2019,PhysRevResearch4023107} and clocks \cite{PRXQuantum.3.010319,YuxiangUltimateBound}, and is also believed to be optimal for others \cite{Sami_Joe}. In the following, $T_0$ denotes the period of the Quasi-Ideal clock (when evolving under its free evolution), i.e. $\rho_\cl^0(T_0)=\rho_\cl^0(0)$.


\begin{theorem}[Achieving t-CNOs]\label{Thm:Implementation with Quasi-Idela clock}
Consider the Quasi-Ideal clock \cite{WSO} detailed in Sec.~\ref{Overview of the quasi-ideal clock} with a time independent Hamiltonian of the form $ \hat H_{\Sy\cat\cl\G}=   \hat H_{\Sy}+  \hat H_{\cat}+  \hat H_{\G}+  \hat I_{\Sy\cat\cl\G}+\hat H_{\cl}$, giving rise to unitary dynamics
\begin{align}
&\rho^F_{\Sy\cat\cl\G}(t)=\nonumber\\
&\qquad V_{\Sy\cat\cl\G}(t) \left(\rho_\Sy^0\otimes \rho_\cat^0\otimes\rho_\cl^0\otimes\tilde\tau_\G\right)V_{\Sy\cat\cl\G}^\dag(t).\nonumber
\end{align}
 For every pair $\rho_\Sy^0$, $\rho_\Sy^1$ for which there exists a t-CNO from $\rho_\Sy^0$ to $\rho_\Sy^1$ using a catalyst $\rho_\cat^0$, there exists an interaction term $\hat I_{\Sy\cat\cl\G}$ such that the following hold.\\
\begin{itemize}
	\item [\textup{1)}] $\sigma_\Sy(t)$ satisfies Eq. \eqref{eq:sigma is valid NO} and is of the form:
	\begin{align}
	&\quad\sigma_\Sy(t)=\nonumber\\
	&\quad\begin{cases}
	\rho_\Sy^0(t) \mbox{ for times $t\in[0,t_1]$ (i.e.\! ``before'' the transition)}\\
	\rho_\Sy^1(t) \mbox{ for times $t\in[t_2,T_0]$ (i.e.\! ``after'' the transition)}
	\end{cases}\nonumber
	\end{align}
	\item [\textup{2)}]  $\epemb$ \textup{(}satisfying Eq. \eqref{eq:ep emb def}\textup{)} is given by 
	\begin{equation}
	\label{eq:thm2-ep-emb}
	\epemb =  \left(2+3\sqrt{d_\Sy d_\cat}\right) \sqrt{\varepsilon_\cl(d_\cl)},
	\end{equation}
	for all $t\in[0,t_1]\cup[t_2,T_0]$, where $\varepsilon_\cl(\cdot)$ is independent of $d_\Sy,d_\cat,d_\G$ and is of order
	\begin{equation}
	\varepsilon_\cl(d_\cl)=\bo\bigg( 
 	poly(d_\cl)\, \exp\left[-c d_\cl^{1/4}\right]
\bigg),\!\!\!\!\!\!\!\!\!\label{eq:scaling of quasi ideal clock}
	\end{equation}
	as $d_\cl  \rightarrow \infty$, with $c=c(t_1,t_2,T_0)>0$ for all $0<t_1<t_2<T_0$\textup{;} and is independent of $d_\cl$.

\end{itemize}

\end{theorem}
See Sec. \ref{Sec: proof of thm 2: Quasi-Ideal clock bound} 
 for a proof.\Mspace

As a direct consequence of Theorem \ref{thm:noemb physical}, in the scenario described in Theorem \ref{Thm:Implementation with Quasi-Idela clock}, $\epres$ is of power law decay in $d_\cl$ as $d_\cl\rightarrow\infty$ and thus both $\epemb$ and $\epres$ are simultaneously small. Therefore the Quasi-Ideal clock allows \emph{all} t-CNOs to be implemented without additional costs not captured by the resource theory.\Mspace

\noindent\fbox{%
	\parbox{0.92\linewidth}{%
		\emph{Take home message from Theorem \ref{Thm:Implementation with Quasi-Idela clock}}:
There exist control systems whose incurred back-reaction is small enough that one is below the threshold mentioned in the previous box. Hence, in conjunction with Theorem \ref{thm:noemb physical}, it implies that the laws of thermodynamics (for baths that transfer only entropy and not heat) do not need to be modified by taking into account the control device.	
	}%
}
\vspace{0.2cm}

The property that $\tilde\tau_\G$ is a maximally mixed state for CNOs is at the heart of two important aspects involved in proving Theorems \ref{thm:noemb physical} and \ref{Thm:Implementation with Quasi-Idela clock}. On the one hand, all CNOs (and hence all t-CNOs by virtue of Prop. \ref{prop:equiv of t-CNO and CNO}), which are implemented via an arbitrary finite dimensional catalyst $\rho_\cat$ can be done so with maximally mixed states $\tilde\tau_\G$ of \emph{finite} dimension\footnote{More precisely, the dimension $d_\G$ is uniformly upper bounded for all t-CNOs on a fixed $d_\Sy$-dimensional system Hilbert space $\mathcal{H}_\Sy$.} \cite{Qhornslemma18}. The other relevant aspect is that they are the only states which are not ``disturbed'' by the action of a unitary, namely $U_\G \tilde\tau_\G U_\G^\dag= \tilde\tau_\G$ for all unitaries $U_\G$. 
Together these mean that the clock only needed to control a system of finite size, and thus the back-reaction it experiences is limited and independent of the dimension $d_\G$.\footnote{Indeed, observe how the dimension $d_\G$ does not enter in either of the bounds in Theorems \ref{thm:noemb physical} or \ref{Thm:Implementation with Quasi-Idela clock}.}\Mspace

One would like to prove analogous theorems to Theorems \ref{thm:noemb physical} and \ref{Thm:Implementation with Quasi-Idela clock} for t-CTOs. Unfortunately, their Gibbs states satisfy neither of these two aforementioned properties. Indeed, there exists CTOs on finite dimensional systems $\mathcal{H}_\Sy$ which require infinite dimensional Gibbs states of infinite mean energy to implement them \cite{Qhornslemma18,3rdLawHenrik,3rdLawLluis}. This observation, combined with the fact that Gibbs states are also generally disturbed by the CTO in the sense that $U_\G \tau_\G U_\G^\dag \neq \tau_\G$ for some $U_\G$, suggests that a theorem like Theorem \ref{Thm:Implementation with Quasi-Idela clock} for which $\epres$ from Theorem \ref{thm:noemb physical} vanishes, is not possible; since the back-reaction on any finite energy or dimensional clock would be infinite in some cases. Furthermore, there is a technical problem which prevents such theorems. The proof of Theorem \ref{thm:noemb physical} uses the known, necessary and sufficient transformation laws for noisy operations (the non increase of the so-called R\'enyi $\alpha$-entropies). However, only necessary (but not sufficient) 2nd laws are known for CTOs (the most well-know of which are the non increase of the so-called R\'enyi $\alpha$-divergences \cite{secondlaw}).

\subsection{Autonomous control for Catalytic Thermal Operations}\label{sec: main text CTO case}
In order to circumvent the dilemma explained at the end of the previous section, we now examine how well the energy preserving unitary of t-CTOs can be implemented when one restricts to attempting to implement t-CTOs which can be implemented with finite baths. We will also allow for some uncertainty in our knowledge | or ability to prepare | the time independent Hamiltonian which implements the transition. 
Specifically we consider
\begin{align}\label{eq:Hm robust emb}
\hat H_{\Sy\cat\cl\G} =\hat H_{\Sy}+\hat H_{\cat}+ \hat H_{\G} +\hat H_{\Sy\cat\G}^\textup{int}\otimes\hat H_{\cl}^\textup{int}+\hat H_\cl,
\end{align}
where 
\begin{align}\label{eq:commutation of int term in generic Hm}
\left[\hat H_{\Sy}+\hat H_{\cat}+ \hat H_{\G} , \,\hat H_{\Sy\cat\G}^\textup{int}\right]=0
\end{align}
and normalisation chosen such that the interaction term has eigenvalues bounded by pi: $\|\hat H^\textup{int}_{\Sy\cat\G}\|_\infty\leq\pi$.
With the interaction term $\hat H^{int}_{\Sy \cat \G}$ in the Hamiltonian Eq. \eqref{eq:Hm robust emb}, and the aid of the thermal bath and clock, we are targeting to implement the joint system-catalyst state
\begin{equation}\label{eq:2nd condition of robust embezzelment Ham}
\sigma_{\Sy\cat}^1:=\tr_\G\left[ \me^{- \mi\hat H_{\Sy\cat\G}^\textup{int} } (\rho_\Sy^0\otimes\rho_\cat^0\otimes\tauGibb)\, \me^{\mi\hat H_{\Sy\cat\G}^\textup{int} } \right].
\end{equation}
From Eq. \eqref{eq:2nd condition of robust embezzelment Ham}, we can observe that the interaction term $\hat H_{\Sy\cat\G}^\textup{int}$ already allows for potential Hamiltonian engineering imperfections, since ideally, the interaction term should leave the final state $\sigma_{\Sy\cat}^1$ in Eq. \eqref{eq:2nd condition of robust embezzelment Ham} in a product state of the form $\rho_{\Sy}^1\otimes\rho_{\cat}^0$. To capture these imperfections in $\hat H_{\Sy\cat\G}^\textup{int}$, we introduce $\hat I_{\Sy\cat\G}^\textup{int}$ 
which, for the initial state 
$\rho_\Sy^0\ot\rho_\cat^0
\ot\tauGibb$, implements an uncorrelated system-catalyst state:
\begin{equation}\label{eq:2nd condition of robust embezzelment Ham 2}
\rho_{\Sy}^1\otimes\rho_{\cat}^0=\tr_\G\left[ \me^{- \mi\hat I_{\Sy\cat\G}^\textup{int} } (\rho_\Sy^0\otimes\rho_\cat^0\otimes\tauGibb)\, \me^{\mi\hat I_{\Sy\cat\G}^\textup{int} } \right].
\end{equation}
Here $\rho_\Sy^1$ is an arbitrary state that can be produced by 
such a transformation, i.e. it is 
an arbitrary state that can be obtained from $\rho^0_\Sy$ via a CTO. 
Note that the evolution
according the total Hamiltonian in Eq. \eqref{eq:Hm robust emb} would not produce such a transformation through time evolution even if we had the term
$\hat I^\textup{int}_{\Sy \cat \G}$ instead of $\hat H^\textup{int}_{\Sy \cat \G}$, since 
the clock is not ideal.\Mspace

If we denote the difference between the states in Eqs. \eqref{eq:2nd condition of robust embezzelment Ham} and \eqref{eq:2nd condition of robust embezzelment Ham 2}, by
\begin{equation}
\label{eq:epsig def}
\epsig:=\|\sigma_{\Sy\cat}^1- \rho_\Sy^1\otimes\rho_\cat^0\|_1,
\end{equation}

then Proposition \ref{prop: H Int to epsig} states that $\epsig$ is upper bounded by
\begin{align}
\epsig \leq 2 \| \delta\hat I_{\Sy\cat\G}^\textup{int} \|_\infty+  \| \delta\hat I_{\Sy\cat\G}^\textup{int} \|_\infty^2,
\end{align}
where $\|\delta\hat I_{\Sy\cat\G}^\textup{int}\|_\infty$ denotes the largest eigenvalue in magnitude of the imperfection in the Hamiltonian preparation: $\delta\hat I_{\Sy\cat\G}^\textup{int}:=\hat H_{\Sy\cat\G}^\textup{int}-\hat I^\textup{int}_{\Sy\cat\G}$. Note that there is also some freedom in the definition of $\hat I^\textup{int}_{\Sy\cat\G}$ in Eq. \eqref{eq:2nd condition of robust embezzelment Ham 2} since the final state of the bath is traced-out and hence irrelevant. One can minimise $\| \delta\hat I_{\Sy\cat\G}^\textup{int} \|_\infty$ over this degree of freedom, reducing the control requirements over the bath degrees of freedom and improving the bounds on $\epsig$.\Mspace

We now introduce a state $\rho_{\Sy\cat\G}^\textup{target}(t)$, which we call the \emph{target state}. It is the state which we would be able to implement with the Hamiltonian in Eq. \eqref{eq:Hm robust emb} if we had access to an idealized clock. Hence any deviations from this will be due to using physical clocks in the control. It is given by
\begin{align}\label{eq:target state def}
\rho_{\Sy\cat\G}^\textup{target}(t):= U_{\Sy\cat\G}^\textup{target}(t) \left(\rho^0_{\Sy}(t)\otimes\rho^0_\cat(t)\otimes\tauGibb\right) U_{\Sy\cat\G}^\textup{target\,\dag}(t)
\end{align}
where $U^\textup{target}_{\Sy\cat\G} (t)= \me^{- \mi \hatf(t) \hat H_{\Sy\cat\G}^\textup{int} }$  with
\begin{align}\label{eq:def target S Cat G}
\hatf(t)=
\begin{cases}
0 &\mbox{ for }  t \in[0, t_1]\\
1 &\mbox{ for }  t\in[t_2,t_3].
\end{cases}
\end{align}
(Recall that the physical meaning of $t_1,t_2$ and $t_3$ can be found  in Def. \ref{def:t-CTO}).
Therefore, tracing out the bath we have for $t \in[0, t_1]$
\begin{align}
\label{eq:ep-H-rho-target1}
 \rho_{\Sy\cat}^\textup{target}(t)= \rho^0_{\Sy}(t)\otimes\rho^0_\cat(t)
 \end{align} 
 while for $t \in[ t_2,t_3]$
 \begin{align}
 \label{eq:ep-H-rho-target2}
\rho_{\Sy\cat}^\textup{target}(t)=\me^{-\mi t (\hat H_\Sy+\hat H_\cat)} \sigma^1_{\Sy\cat}\me^{\mi t (\hat H_\Sy+\hat H_\cat)}.
 \end{align}
We now define a quantity $\Delta(t;x,y)$ which \textit{only} depends on properties of the clock system:
\begin{align}
\Delta(t;x,y)&:= \braket{\rho_\cl^0|\hat\Gamma_\cl^\dag(x,t)\hat \Gamma_\cl(y,t)|\rho_{\cl}^0},\\
\hat\Gamma_\cl(x,t)&:= \me^{-\mi t\hat H_\cl+ \mi x \left(\hatf(t) \id_\cl- t\hat H_\cl^\textup{int}\right)}, \quad x,t\in\rr.\label{eq:clock plus Ham def}
\end{align} 
The following theorem states that if $\Delta(t;x,y)$ is small for all $x,y\in[-\pi,\pi]$, and the dimension of the bath $d_\G$ is not too large, then the clock can implement a unitary over the system, catalyst and clock which is close to a t-CTO using the time independent Hamiltonian in Eq. \eqref{eq:Hm robust emb}. Furthermore, the clock itself is not disturbed much during the process.

\begin{theorem}[Sufficient conditions for t-CTOs]\label{thm:noemb physical t CTO} For all states $\rho_\Sy^0$ and $\rho_\cat^0$, consider unitary dynamics $V_{\Sy\cat\cl\G}(t)=\me^{-\mi t \hat H_{\Sy\cat\cl\G}}$ implemented via any Hamiltonian of the form Eq. \eqref{eq:Hm robust emb}, with an initial pure clock state $\rho_\cl^0=\ketbra{\rho_\cl^0}{\rho_\cl^0}$. Namely,  $\rho^F_{\Sy\cat\cl\G}(t)= V_{\Sy\cat\cl\G}(t)(\rho_\Sy^0\otimes \rho_\cat^0\otimes\rho_\cl^0\otimes\tauGibb)V_{\Sy\cat\cl\G}^\dag(t)$. Then the following hold:

\begin{itemize}
	\item [\textup{1)}] The deviation from the idealised dynamics is bounded by
	\begin{align}\label{eq:F - F thm 3 main text}
		&\| \rho^F_{\Sy\cat\cl}(t) - \rho^F_{\Sy}(t)\otimes \rho_\cat^0(t)\otimes  \rho_\cl^0(t)  \|_1 \leq\,\, 2\epsig\,\hatf(t)\,+\nonumber\\
		&  6\sqrt{d_\Sy\, d_\cat\, d_\G \tr\!\left[\tau_\G^2\right]}\, \sqrt{ \max_{x, y\in[-\pi,\pi]} \left|1-   \Delta^2(t; x, y) \right|}.
	\end{align}
	
	\item [\textup{2)}] The final state $\rho^F_\Sy(t)$ is 
		\begin{align}
			&\| \rho^F_{\Sy}(t) - \rho^\textup{target}_{\Sy}(t)  \|_1 \leq \epsig\,\hatf(t)\,+ \nonumber\\
			&\sqrt{d_\Sy \, 
				d_\cat \, d_\G \tr\left[\tau_\G^2\right]} \max_{x,y\in[-\pi,\pi]} \left|1 - \Delta^2(t;x,y)\right|\label{eq:final state and target state}
		\end{align}
	close to one which can be reached via t-CTO:  For all $t\in[0,t_1]\cup[t_2,t_3]$ the transition
	\begin{equation}\rho_\Sy^0\otimes\rho_\cat^0\otimes\rho_\cl^0 \;\;\;\text{ to }\;\;\; \rho_\Sy^\textup{target}(t)\otimes\rho_\cat^0(t)\otimes\rho_\cl^0(t)\label{eq:sigma is valid TO}
	\end{equation}
is possible via a TO i.e. $\rho_\Sy^0$ to $\rho_\Sy^\textup{target}$ via a t-CTO. 
\end{itemize}
\end{theorem}
A proof can be found in \App~\ref{Sec:proof of thm 3, t-CTOs generic bound}.

Since the definition of the target state in Eq. \eqref{eq:target state def} allows one to reach all t-CTOs which are implementable with a $d_\G$ dimensional bath,\footnote{This is by construction, c.f. Eq. \eqref{eq:target state def} and definitions of CTO and t-CTO in Sec. \ref{sec:CTO def}.} Theorem \ref{thm:noemb physical t CTO} provides sufficient conditions for the implementation of \emph{all} t-CTOs which are implementable via such baths. As long as the set of CTOs with finite bath size is a dense subset of the set of all CTOs, Theorem \ref{thm:noemb physical t CTO} provides sufficient conditions for implementing a dense subset of CTOs. While the TO in Eq. \eqref{eq:sigma is valid TO} for $t\in[0,t_1]$ is ``trivial'' in the sense that it does not involve interactions between the subsystems nor requires the thermal bath, it is nevertheless important since it captures the notion of ``turning on'' the unitary | an essential step in the implementation of any unitary operation.\Mspace

Intuitively, in order for $\Delta(t;x,y)\approx 1$ for all $x,y\in[-\pi,\pi]$, we see from Eq. \eqref{eq:clock plus Ham def} that we want the initial clock state $\ket{\rho_\cl^0}$ to be orthogonal to the interaction term $\hat H_\cl^\textup{int}$ initially, and subsequently the dynamics of the clock according to its free Hamiltonian $\hat H_\cl$ to ``rotate'' the initial clock state $\ket{\rho_\cl^0}$ to a state which is no longer orthogonal to $\hat H_\cl^\textup{int}$ after a time $t_1$ when the interaction starts to happen. Similarly, the evolution induced by $\hat H_\cl$ should make the state $\ket{\rho_\cl^0}$ orthogonal to $\hat H_\cl^\textup{int}$ after time $t_2$. Meanwhile, the interaction term $\hat H^\textup{int}_\cl$ should have imprinted a phase of approximately $\me^{-\mi x }$ onto the state $\ket{\rho_\cl^0}$ during the time interval $(t_1,t_2)$ to cancel out the phase factor $\me^{\mi x \hatf(t)}$ in Eq. \eqref{eq:clock plus Ham def}. So we can think of the quantity $\Delta(t,x,y)$ as a formal mathematical expression which quantifies the intuitive physical picture of ``turning on and off an interaction''.\Mspace

The Quasi-Ideal clock, which recall is of dimension $d_\cl$ and period $T_0$ (when evolving under its free evolution), can realise the above intuition to a very good approximation. Indeed, the following theorem bounds the quantities on the r.h.s. of Eqs. \ref{eq:F - F thm 3 main text}, \ref{eq:final state and target state}, up to engineering errors $\ep_H$, by setting $t_3=T_0$ in Theorem \ref{thm:noemb physical t CTO}:
\begin{theorem}[Achieving t-CTOs]
	\label{thm:4}
	For the Quasi-Ideal clock, we have:
	\begin{align}
	\max_{x,y\in[-\pi,\pi]} &\left|1 - \Delta^2(t;x,y)\right| \leq 
	\\ \nonumber
 	& \bo\bigg( poly(d_\cl)\, 
	\exp\left[-c d_\cl^{1/4}\right] \bigg)
	\end{align}
	as $d_\cl\to\infty$ for all $t\in[0,t_1]\cup[t_2,T_0]$, where $\Delta^2(t;x,y)$  is 
	defined
	in \eqref{eq:clock plus Ham def} and where $c=c(t_1,t_2,T_0)>0$ for all $0<t_1<t_2<T_0$\textup{;} and is independent of $d_\cl$.
\end{theorem}

See \App~\ref{subsec:deltaxy} for proof.
On the other hand, it turns out that the idealised momentum clock discussed in Sec. \ref{Idealised control}\,, satisfies $\Delta(t;x,y)= 1$ for all $x,y\in[-\pi,\pi]$ for an appropriate parameter choice in which case 1) in Prop. \ref{prop:idealised control} fails (see Sec. \ref{sec:proof Delat=1 for idelaised clock} in \app). Thus the r.h.s. of Eqs. \eqref{eq:F - F thm 3 main text}, \eqref{eq:final state and target state} are \emph{exactly} zero for all $t_1<t_2$ in this case. This observation highlights another point of failure for this clock: it allows for the violation of the 3rd law of thermodynamics. The 3rd law states that any system cannot be cooled to absolute zero (its ground state) in finite time. In \cite{Qhornslemma18,3rdLawHenrik}, it was shown that under CTOs, both the mean energy and dimension $d_\G$ of the bath need to diverge in order to cool a $d_\Sy$ dimensional system to the ground state. The inability to do this in finite time by any realistic control system on $\mathcal{H}_\cl$ manifests itself in that $\max_{x,y\in[-\pi,\pi]} \left|1 - \Delta^2(t;x,y)\right|$ cannot be exactly zero in this case, so that the r.h.s. of Eq. \eqref{eq:final state and target state} becomes large due to the factor $d_\G\tr[\tauGibb^2]$ diverging.\footnote{Note that the only case in which $d_\G\tr[\tauGibb^2]$ does not diverge in the large $d_\G$ limit, is when the purity of the Gibbs state $\tauGibb$ converges (in purity) to the maximally mixed state, since in that case $\tr[\tauGibb^2]=1/d_\G$. This is not the case for the baths needed to cool to absolute zero in which $\tr[\tauGibb^2]$ converges to a positive constant \cite{Qhornslemma18}.}  However, for the idealised momentum clock, the r.h.s. of Eqs. \eqref{eq:F - F thm 3 main text}, \eqref{eq:final state and target state} are exactly zero even in the limit $d_\G\tr[\tauGibb^2] \rightarrow \infty$, thus allowing one to cool the system on $\mathcal{H}_\Sy$ to absolute zero in any finite time interval $[t_1,t_2]$. Finally, it is also worth noting that the change in Von Neumann entropy of the clock between before and after the unitary is implemented is vanishingly small for the Quasi-Ideal clock as its dimension increases. This follows from applying the Fannes inequality to the results of Theorems~\ref{thm:noemb physical t CTO} and \ref{thm:4}. This is because the Fannes inequality implies that the change in Von Neumann entropy between two states approaches zero when the trace distance between said states decreases faster than $1/\log(d)$, where $d$ is the dimension of the system in question.

\vspace{0.2cm}
\noindent\fbox{%
	\parbox{0.92\linewidth}{%
		\emph{Take home message from Theorem \ref{thm:noemb physical t CTO}}:
		The result provides bounds which characterise the back-reaction incurred on any control device implementing an arbitrary thermodynamic transition i.e. with baths which transfer both entropy and heat. It includes and quantifies engineering imperfections,
		 and has important physical consequences
		 for nonequilibrium physics and the 3rd law.  
	}%
}
\vspace{0.2cm}

\noindent\fbox{%
	\parbox{0.92\linewidth}{%
		\emph{Take home message from Theorem \ref{thm:4}}:
		There exists a control device, such that the bounds in Theorem \ref{thm:noemb physical t CTO} for the incurred back-reaction are, up-to engineering inaccuracies, exponentially small in the device's dimension. 
		Thus Theorems \ref{thm:noemb physical t CTO} and \ref{thm:4} together imply the existence of control devices such that the laws of thermodynamics are not modified for baths that can transmit both entropy and heat. 
		
	}%
}
\vspace{0.2cm}


\section{Discussion}\label{sec:Discussion}
Other than the fact that Theorem \ref{thm:noemb physical} provides necessary conditions for implementation of t-CNOs while Theorem \ref{thm:noemb physical t CTO} for implementation of t-CTOs, there are two main differences between them. The first is that Theorem \ref{thm:noemb physical} applies to any time independent Hamiltonian while Theorem \ref{thm:noemb physical t CTO} to Hamiltonians of a particular form. The other main difference, is that Theorem \ref{thm:noemb physical} provides bounds in terms of how close the catalyst and clock are in \emph{trace distance} to their desired states, while Theorem \ref{thm:noemb physical t CTO} provides bounds in terms of how close $\Delta(t;x,y)$ is to unity. While the latter condition implies small trace distance between the clock and its free evolution, the converse is not necessarily true. Fortunately, while $\Delta(t;x,y)\approx 1$ is a stronger constraint, we have shown that it can be satisfied by the Quasi-Ideal clock (This is Theorem \ref{thm:4}). However, from a practical point of view, its fulfilment is likely harder to verify experimentally, since quantum measurements can be used to evaluate trace distances, while the ability to experimentally determine $\max_{x,y\in[-\pi,\pi]} \Delta(t;x,y)$ is less clear.

Observe how the bounds in Theorems \ref{thm:noemb physical} and \ref{thm:noemb physical t CTO} increase with $d_\cat$, the dimension of the catalyst. This aspect of the bound is also relevant in some important cases. Most exemplary is the setting of the important results of \cite{Muller2017} which show that if one allows the catalysts to become correlated, then | up to an arbitrarily small error $\epsilon$ | there exists a catalyst and energy preserving unitary which achieves any TO between states block diagonal in the energy basis if and only if the second law (non increase of von Neumann free energy) is satisfied. Here, the dimension of the catalyst diverges as $\epsilon$ converges to zero. The setting considered was that of idealised control, and thus the divergence of the catalyst did not affect the implementation of transitions. However, if one were to consider realistic control such as in our paradigm, the rate at which the catalyst diverges would be an important factor in determining how much back-reaction the clock would receive and consequently how large it would have to be to counteract this effect, and achieve small errors in the implementation of the control.\Mspace

There are various results regarding the costs of implementing unitary operations \cite{PhysRevLett89057902,PhysRevA74060301,PhysRevLettJohan,Skrzypczyk2013,Tajima2019,Takagi2019,Chiribella2019,Clivaz2017,Clivaz2019}. These all have in common the assumption of implicit external control, while only restricting the set of allowed unitaries which is implemented by the external control. The allowed set of unitaries is motivated physically by demanding that they obey conservation laws (such as energy conservation), or by comparing unitaries which allow for coherent vs. incoherent operations. So while these works consider interesting paradigms, the questions they can address are of a very different nature to those posed and answered in this~\doc. In particular, the assumption of perfect control on the allowed set of unitaries means that effects such as back-reaction or degradation of the control device are neglected.\Mspace

While other bounds do impose limitations arising from dynamics, these bounds are not of the right form to address the problem at hand in the~\doc. Perhaps one of the most well-known results in this direction is the so-called quantum speed limit which characterises the minimum time required for a quantum state to become orthogonal to itself or more generally, to within a certain trace-distance of itself. Indeed, such results have been applied to thermodynamics, metrology and the study of the rate at which information can be transmitted from a quantum system to an observer \cite{QSL_review,Bekenstein90}. In our context, the promise is of a different form, namely rather than the final state being a certain distance away from the initial state, we need it to be a state which is close to one permissible via the transformation laws of the resource theory (t-CNOs or t-CTOs). Similar difficulties arise when aiming to apply other results from the literature. Perhaps most markedly is \cite{121110403}. Here necessary conditions in terms of bounds on the fidelity to which a unitary can be performed on a system, via a control device, is derived. Unfortunately, this result is unsuitable for our purposes for two reasons. Firstly, their bounds become trivial in the case that the unitary over the system to be implemented commutes with the Hamiltonian of the system (as is the case in this~\doc). Secondly, since catalysis is involved in our setting, bounds in trace distance for the  clock precision of how well the unitary was implemented, are not meaningful, due to the embezzling problem discussed in Sec. \ref{Idealised control}. The latter problem is also why one cannot arrive at the conclusions of this \doc{} from \cite{WSO} alone.\Mspace

This work opens up interesting new questions for future research. In macroscopic thermodynamics, the 2nd law applies to transitions between states which are in thermodynamic equilibrium. Such a notion is not present in the CTOs, since the 2nd laws governing transitions apply always, regardless of the nature of the state. One intriguing possibility which comes in to view with the results in this \doc{} is that the CTOs actually only hold in equilibrium, and the apparent absence of this property had been hidden in the unrealistic assumption of idealised control. To see why, observe that we have only proven that the transition laws for t-CTOs hold for times $t\in[0,t_1]\cup[t_2,t_3]$ where the unitary implementing the transition occurs within the time interval $(t_1,t_2)$. It would appear that CTOs are not satisfied for the state during the transition period $(t_1,t_2)$. If this can be confirmed and proven to hold in general, then this would suggest that the CTOs actually only hold in equilibrium. A potential physical mechanism explaining this could be that at times around $t_1$ the clock sucks up entropy from the system it is controlling | allowing it to become more pure | after finally releasing entropy back around the $t_2$ time | so that the system can then become mixed enough to satisfy the 2nd laws.\Mspace

Another aspect which the introduction of the paradigm of physical control into the paradigm of CTOs has given rise naturally to, is the variant of the 3rd law of thermodynamics stating that one cannot cool to absolute zero in finite time. It is noticeably absent from the CTO formalism. Future work could now investigate this property in more depth. Previous characterisations of the 3rd law \cite{3rdLawLluis} had to assume that the spatial area which the unitaries in the idealised control could act upon, satisfied a light-cone bound. While this is indeed a realistic assumption, it did not arise from the mathematics. Here it arises naturally even without the need for a light-cone bound assumption.\Mspace

Introducing similar non idealised control  for other resource theories \cite{Liu2019,Sparaciari2018} could allow us to understand the requirements of these paradigms.

\section{Conclusions}\label{sec:Conclusions}

The resource theory approach to quantum thermodynamics has been immensely popular over the last few years. However, to date the conditions under which its underlying assumptions of idealised external control can be fulfilled, have not been justified. While it is generally appreciated that they cannot be achieved perfectly, to what extent and under what circumstances they can be approximately achieved, remained elusive. Our \doc~addresses this issue, providing sufficient conditions which we prove are satisfiable. In doing so, our work has united two very popular yet starkly different paradigms: fully autonomous thermal machines and resource theoretic non-autonomous ones. Our approach and methods set the groundwork for future unifications of generic quantum processing machines | of which resource theoretic thermal machines can be seen as a particular example | with generic autonomous quantum processes.\Mspace

Not only could these results be instrumental for future experimental realisations of the 2nd laws of quantum thermodynamics, but they can also open up new avenues of research into the 3rd law of thermodynamics and the role of non-equilibrium physics.\Mspace

In particular, we have introduced a paradigm in which the cost of control in the resource theory approach of quantum thermodynamics using CNOs and CTOs can be characterised. This was achieved via the introduction of t-CNOs and t-CTOs in which control devices fit naturally into this thermodynamic setting as dynamic catalysts.\Mspace

We have then derived sufficient conditions on how much the global dynamics including the control device can deviate from the idealised case, in order for the achieved state transition to be close to one permissible via CNOs. This is followed by examples of a control device which achieves this level of precision.\Mspace


Finally, we introduced Hamiltonians which led us to a criteria for all CTOs with a finite dimensional bath. The bound captures the requirement of better quality control, as the bath size needed to implement the CTO gets larger.\Mspace

\acknowledgments
We thank Jonathan Oppenheim for stimulating discussions and pointing out the control problem in the resource theory approach to quantum thermodynamics. We thank Gian Michele Graf for discussions regarding functional analysis for Prop. \ref{prop:idealised control} and Elisa B\"aumer for physically demonstrating the embezzling phenomenon in Fig. \ref{Fig:ClssicalEmbezzExaple}. M.W. acknowledges support from the Swiss National Science Foundation (SNSF) via an AMBIZIONE Fellowship (PZ00P2\_179914) in addition to the National Centre of Competence in Research QSIT. M.H. acknowledges support from the National Science Centre, Poland,
through grant OPUS 9. 2015/17/B/ST2/01945 and the Foundation for Polish Science through IRAP project co-financed by EU within the Smart Growth Operational Programme (contract no. 2018/MAB/5).

\onecolumngrid
\vspace{1cm}\begin{center}
\fbox{\parbox{0.96\textwidth}{Unless stated otherwise, the below commonly used notation has the indicated meaning:
		\begin{itemize}
			\item \underline{Abbreviations for transformations}: NO = \emph{Noisy Operation}, CNO = \emph{Catalytic Noisy Operation}, TO = \emph{Thermal operation}, CTO = \emph{Catalytic Thermal Operation}. The prefix ``t-'' can be added to any of these abbreviations, and stands for \emph{time}. See Sec. \ref{sec:CTO def} for their definitions.
			\item \underline{Subscripts}: the following subscripts are added to states to indicate the subsystem they belong to. Subscript ${}_\Sy$ is the \emph{system}, ${}_\cat$ is the \emph{catalyst}, ${}_\cl$ is the \emph{clock}, ${}_\G$ is the \emph{bath}. A result with a subscript ${}_\textup{A}$ means the result holds for both cases $\textup{A}=\Sy$ and $\textup{A}=\Sy\cat$.
			\item \underline{Partial trace}: we use the quantum information notation for partial trace: for a generic bipartite quantum state $\rho_{\textup{X}_1\textup{X}_2}$, we denote the state on subsystem $\textup{x}_1$ after tracing out subsystem $\textup{x}_2$, by $\rho_{\textup{X}_1}$.
			\item \underline{Time dependency}: 
			$\rho_\textup{X}^0$ or $\sigma_\textup{X}^0$ is the initial state on subsystem $\textup{X}$. $\rho_\textup{X}^1$ or $\sigma_\textup{X}^1$ is the state on subsystem $\textup{X}$, after the application of a fixed transformation to the initial state $\rho_\textup{X}^1$ or $\sigma_\textup{X}^1$ respectively. $\rho_\textup{X}^n(t)$ or $\sigma_\textup{X}^n(t)$ for $n=0,1$ is the dynamically evolved state $\rho_\textup{X}^n$ or $\sigma_\textup{X}^n$ according to its local Hamiltonian: $\rho_\textup{X}^n(t)=\me^{\mi t\hat H_\textup{X}} \rho_\textup{X}^n \me^{-\mi t\hat H_\textup{X}}$ or  $\sigma_\textup{X}^n(t)=\me^{\mi t\hat H_\textup{X}} \sigma_\textup{X}^n \me^{-\mi t\hat H_\textup{X}}$. These definitions are introduced in Sec. \ref{sec:CTO def}. The notation $\rho^F_{\textup{X}_1\ldots \textup{X}_l}(t)$ refers to a state on subsystems ${\textup{X}_1\ldots \textup{X}_l}$ at time $t$ whose time evolution is not given (in general) by the sum of the local Hamiltonians $\hat H_{\textup{X}_1}+\ldots+ H_{\textup{X}_l}$. Its exact definition is context dependent and given locally in the text. 
			\item \underline{Dimensions}: $d_\textup{X}$ is the Hilbert space dimension of subsystem $\textup{X}$.
			\item \underline{Thermal states}: $\tau_\textup{X}$ is the Gibbs state of subsystem $\textup{X}$, i.e. $\tau_\textup{X}= \me^{-\hat H_\textup{X}/T} /Z$, where $Z$ is the partition function, and $T$ is temperature in appropriate units. The \emph{maximally mixed state} denoted $\tilde\tau_\textup{X}$, is a special Gibbs state corresponding to when $\hat H_\textup{X}$ is proportional to the identity $\id_\textup{X}$. It takes on the form $\tilde\tau_\textup{X}= \id_\textup{X}/d_\textup{X}$.
\end{itemize}}}
\end{center}
\vspace{1cm}\newpage
\twocolumngrid
 \appendix
\section{Proof Overviews}\label{sec:appendix}

In this \app{} we will provide the proofs of the results in the maim text. Owing to the complexity of some of these proofs, with the exception of Proposition \ref{prop:idealised control}, the others will have a high level overview of the proof here, with details relegated to the \Supp.

\subsection{Proof of Proposition \ref{prop:idealised control}}\label{sec:prof of lemma Idealisec control}

We will here prove Proposition \ref{prop:idealised control}. We will assume the assertions under both bullet points in the proposition, and culminate in a contradiction hence showing that the assertions cannot simultaneously hold. 
To start with, we denote the unitary transformation implementing the TO from $\rho_{\A\G}^0(t)$ to $\rho_{\A\G}^1(t)$ by $U_{\A\G}(t)=\me^{-\mi \delta(t) \hat H_u}$ where
\begin{align}
\delta(t)=\begin{cases}
0 &\mbox{ if } t\in[0,t_1]\\
1 &\mbox{ if } t\in[t_2,t_3].
\end{cases}
\end{align}
By definition of TOs, $U_{\A\G}(t)$ is an energy preserving unitary which must commute with $\hat H_{\A\G}=\hat H_{\A}\otimes\id_\G+\id_\A\otimes\hat H_{\G}=\sum_{n=1}^{d_\A d_\G}  E_n\proj{E_n}_{\A\G}$ and can therefore be chosen to be of the form $\hat H_u=\sum_{n=0}^{d_\A d_\G}\Omega_n \proj{E_n}_{\A\G}$ with $\Omega_n\in[-\pi,\pi)$. In order to avoid trivial unitaries, we have also assumed that the phases are non-degenerate, $\Omega_n\neq \Omega_p$ for $n\neq p$. It then follows from $[U_{\A\G}\otimes\id_\cl, \hat H_{\A\G\cl}]=0$ that 
\begin{align}
\hat H_\cl^{(k,l)}=0,
\end{align}
for $k\neq l$. Using the expansion of $\hat H_{\A\G\cl}$ from the preposition, it then follows that
\begin{align}
\hat H_{\A\G\cl}=  \hat H_{\A\G}\otimes \id_\cl + \sum_{n=1}^{d_\A d_\G} \proj{E_n}_{\A\G}\otimes\hat H_\cl^{(n,n)}.
\end{align}
Expanding the state $\rho_{\A\G}$ in the energy basis, $\rho_{\A\G}= \sum_{l,m=1}^{d_\A d_\G} A_{l,m} \ketbra{E_l}{E_m}_{\A\G}$, we find from the definition of $\rho_{\A\G\cl}^F(t)$
\begin{align}
\braket{E_l|\rho_{\A\G}^F(t)|E_m}= A_{l,m}(t) \,\tr\left[\me^{-\mi t\hat H_\cl^{(l,l)}}\! \rho_\cl\, \me^{\mi t\hat H_\cl^{(m,m)}}\right],\label{eq:final state with clock appendix}
\end{align}
where the time dependency of the coefficients $A_{l,m}(t)$ is defined via $\rho_{\A\G}(t)=\me^{-\mi t\hat H_{\A\G}}\,\rho_{\A\G}\,\me^{\mi t\hat H_{\A\G}}= \sum_{l,m=1}^{d_\A d_\G} A_{l,m}(t) \ketbra{E_l}{E_m}_{\A\G}$. On the other hand, 
\begin{align}
\braket{E_l|U_{\A\G}(t)\rho_{\A\G}(t)U_{\A\G}^\dag(t)|E_m}= A_{l,m}(t)\, \me^{-\mi t\left(\Omega_m-\Omega_l\right)\delta(t)}.\label{eq:idelaised control appendix}
\end{align}
We will now proceed to show the contradicting statement. Let us assume we can equate Eqs. \eqref{eq:cases Prop idealised control}, \eqref{eq:rho F A t} and furthermore assume that the power series in Eq. \eqref{eq:convergent series} is convergent in the neighbourhood of either $t_1$ or $t_2$.  
Since Eq. \eqref{eq:idelaised control appendix} holds in the case of Eq. \eqref{eq:cases Prop idealised control}, and Eq. \eqref{eq:final state with clock appendix} holds in the case of Eq. \eqref{eq:rho F A t}, we find by equating these equations for all $m\neq l,$\,\, $m,l=1,2,\ldots,d_\A d_\G$ :
\begin{align}
\me^{-\mi t\left(\Omega_m-\Omega_l\right)\delta(t)}=\tr\left[\me^{-\mi t\hat H_\cl^{(l,l)}}\! \rho_\cl^0\, \me^{\mi t\hat H_\cl^{(m,m)}}\right].  \label{eq:idelaised control appendix equating}
\end{align}
Hence if the power series expansion Eq. \ref{eq:convergent series} holds,  
\begin{align}
&\me^{-\mi t\left(\Omega_m-\Omega_l\right)\delta(t)}=\\
&\sum_{n,p=0}^\infty \tr\left[\frac{\big(-\mi  \hat H_\cl^{(l,l)}\big)^n}{n!} \rho_\cl^0\, \frac {\big(\mi\hat H_\cl^{(m,m)}\big)^p}{p!} \right] t^{n+p}.  \label{eq:idelaised control appendix with power series}
\end{align}
However, for $t\in[0,t_1]$ we have that $\delta(t)=0$, thus since $0<t_1<r$, with $r$ the radius of convergence of the power series, for any $\tilde t\in(0,t_1)$, we find \footnote{Note that in order to make this conclusion, we interchange derivatives with the infinite sum, which is well known to hold for power series.}
\begin{align}
\frac{d^q}{dt^q}\, \me^{-\mi t\left(\Omega_m-\Omega_l\right)\delta(t)}\Bigg{|}_{t=\tilde t}=0 \quad\text{ for }\quad q\in \nn^+.
\end{align}
If we take derivatives of the r.h.s. of Eq. \eqref{eq:idelaised control appendix with power series}, evaluate at $t=\tilde t$ and set to zero, we find
\begin{align}
\tr\left[\frac{\big(-\mi  \hat H_\cl^{(l,l)}\big)^n}{n!} \rho_\cl^0\, \frac {\big(\mi\hat H_\cl^{(m,m)}\big)^p}{p!} \right]= \delta_{0,n} \delta_{0,p},
\end{align}
where $\delta_{n,p}$ denotes the Kronecker-Delta function. Yet if we plug this solution into the r.h.s. of \eqref{eq:idelaised control appendix with power series}, we find a contradiction for $t\in[t_2,r)\neq \emptyset$. 

\subsection{Proof of Theorem \ref{thm:noemb physical}}\label{Main proof section}

In this section we prove Theorem \ref{thm:noemb physical}. We will also need the results from sections \ref{sec:entropies divergences def and properties} to \ref{Renyi Entropy Continuity Theorem} to aid the proof.

The below theorem, is a slightly more general version of Theorem \ref{thm:noemb physical} in three ways:\\
1) In the below theorem no time dependency is assumed, since while it is physically reasonable to do so, it is not necessary from a mathematical perspective to prove our theorem.\\
2) We denote by $\rho_\cat^0$ a generic catalyst of dimension $D_\cat$. To achieve the version Theorem \ref{thm:noemb physical} in the main text, one makes the identification $\rho_\cat^0$ in the below theorem with $\rho_\cat^0\otimes\rho_\cl^0$ in Theorem \ref{thm:noemb physical}, and letting $D_\cat=d_\cat d_\cl$.  The motivation for this relabelling, is that for the purposes of this proof, there is no point in distinguishing between the clock catalyst (which controls the interaction) and the other catalyst, which allows for thermodynamic transitions, which would otherwise not be permitted under TOs. In other words, it is only in later theorems that we care about actual dynamics where the distinction between the two types of catalysts is important.\\
3) The bound on $\epres(\epemb,d_\Sy, D_\cat)$ in Eq. \ref{eq:long ep emb form} is a more general version than that stated in Theorem \ref{thm:noemb physical} in the main text. A proof that Eq. \ref{eq:long ep emb form} implies the version stated in the main text can be found in Corollary \ref{rem:tigher bound on ep emb}.

\begin{theorem}[Sufficient conditions for implementing CNOs]
	\label{thm:noemb}
	Consider arbitrary initial state $\rho_\Sy^0$ of not full rank and arbitrary catalyst $\rho^0_\cat$. Consider arbitrary unitary $V_{\Sy\cat\G}$ over $\rho_\Sy^0\otimes \rho_\cat^0\otimes\tilde\tau_\G$,
	and suppose that the final state, $\rho_{\Sy\cat\G}^F=V_{\Sy\cat\G} (\rho_\Sy^0\otimes\rho_\cat^0\otimes\tilde\tau_\G) V_{\Sy\cat\G}^\dag$ satisfies
	\be
	\label{eq:final close}
	\| \rho^F_{\Sy \cat} - \rho^F_\Sy\ot \rho^0_\cat\|_1 \leq \epemb.
	\ee
	Then there exists a state $\sigma_\Sy$ which is close to $\rho^F_\Sy$
	\be
	\|\sigma_\Sy - \rho^F_\Sy\|_1  \leq \epres
	\ee
	such that 
	\be
	\rho_\Sy^0 \ot \tilde \rho_\cat \succ  \sigma_\Sy \ot \tilde \rho_\cat,
	\ee
	for some finite dimesioanl catalyst $\tilde \rho_\cat$.
	Here $\epres=\epres(\epemb,d_\Sy, D_\cat)$ where $d_\Sy$, $D_\cat$ are the dimensions of system $\rho_\Sy^0$ and catalyst $\rho_\cat^0$ respectively. Specifically,
	\begin{widetext}
	\begin{align}
	\epres(\epemb,d_\Sy, D_\cat)= 5\sqrt{\frac{d_\Sy^{5/3} + 4(\ln d_\Sy D_\cat) \ln d_\Sy}{\ln (1/\epemb)}+ d_\Sy D_\cat \epemb^{1/6} + 5\left( (d_\Sy D_\cat)^2 \sqrt{\frac{\epemb}{d_\Sy D_\cat}} \ln\sqrt{\frac{d_\Sy D_\cat}{\epemb}}  \right)^\frac23}.\label{eq:long ep emb form}
	\end{align}
	\end{widetext}
	Explicitly one possible choice for $\sigma_\Sy$ is
	\be
	\label{eq:sigma-S}
	\sigma_\Sy\!= \!
	\left\{
	\begin{array}{ll}
	\!	\id_\Sy/d_\Sy &\!\text{\rm if } \|\rho_\Sy^{\mpwno{F}} -\id_\Sy/d_\Sy\|_1 < \epres \\
	\!	(1-\epres) \rho_\Sy^{\mpwno{F}} + \epres \id_\Sy/d_\Sy &\!\text{\rm if } \|\rho_\Sy^{\mpwno{F}} -\id_\Sy/d_\Sy\|_1 \geq \epres \\
	\end{array}
	\right.
	\ee
\end{theorem}
\noindent
{\it Overview of the proof}.
We shall show  that catalytic majorization holds, by using Klimesh conditions given in Theorem \ref{thm:Klimesh}. Since we assume that the initial state is not of full rank, and the final state $\sigma_S$ by definition is of full rank, 
 it is enough to show that for   $\alpha>0$, $g_\alpha(\rho_\Sy^0)$ 
 is strictly larger than 
 $g_\alpha(\sigma_\Sy)$.
 In terms of simpler functions $f_\alpha$ given by \eqref{eq:f1}, 
 we need to show that 
 \begin{align}
     &f_\alpha(\rho_\Sy^0) >
     f_\alpha(\sigma_\Sy)  \quad \text{for} \quad \alpha >1 \nonumber \\
     &f_\alpha(\rho_\Sy^0) <
     f_\alpha(\sigma_\Sy)  \quad \text{for} \quad \alpha \in(0,1].
 \end{align}
 In particular $f_1$ is the Shannon entropy, so the condition for $\alpha=1 $ can also be written as 
 \begin{align}
 S_1(\rho_\Sy^0) <  S_1(\sigma_\Sy).
 \end{align}
 There are other equivalent ways of writing the conditions using the Tsallis-Aczel-Daroczy entropy (in short Tsallis entropy)  
 defined in Eq. \eqref{eq:Tsalis def}, or Renyi entropy of Def. \ref{def:renyi entrpies}
 \begin{align}
     &T_\alpha(\rho_\Sy^0) < T_\alpha(\sigma_\Sy) \quad \text{for} \quad \alpha>0, 
     \label{eq:T-klimesh}\\
     & S_\alpha(\rho_\Sy^0) <
     S_\alpha(\sigma_\Sy) \quad \text{for} \quad \alpha>0
     \label{eq:S-klimesh}.
 \end{align}
 It is in one way more convenient than the condition in terms of $f_\alpha$. 
 Namely, the case $\alpha=1$ is not given by a separate formula. Indeed, 
 $T_1=\lim_{\alpha\to1} T_\alpha$
 (same for $S_\alpha$).
 Note here that for each single $\alpha$, the inequality with Renyi entropy $S_\alpha$ is equivalent to
 inequality for $T_\alpha$. 
 Thus for some $\alpha$'s we may show the inequality for $T_\alpha$ 
 while for others for $S_\alpha$.
Now, let us sketch 
how we will approach this problem. 

\noindent{\it 1) Showing the inequalities \eqref{eq:T-klimesh} for states $\rho^0_\Sy$ and $\rho^F_\Sy$  up to term $\eta_\alpha$.}\vspace{0.2cm}\\
By assumption the initial state $\rho^0_\Sy\ot\rho^0_\cat\ot\rho^0_\G$ is unitarily transformed into $\rho^F_{\Sy\cat\G}$. 
This transformation does not change functions like $f_\alpha$, $T_\alpha$, $S_\alpha$. Then,
going back and forth between $f_\alpha$'s and $T_\alpha$'s, and using the continuity of $T_\alpha$ 
from Theorem \ref{thm:Tsalis continuity}
we obtain 
\be
	\label{eq:nonstrict T-overview}
	T_\alpha(\rho_\Sy^0)\leq T_\alpha(\rho^F_\Sy)+ \eta_\alpha D_\cat^{\alpha} \quad \text{\rm for } \alpha>0
	\ee
with $\eta_\alpha$ satisfying 
 \ben
	\label{eq:eta1-overview}
	&&\eta_\alpha \geq 6 D 	\left( \frac{\mpwno{\epemb}}{D}\right)^\alpha \quad  \text{for } \alpha \in (0,1/2] \\
	\label{eq:eta2}
	&&\eta_\alpha \geq  -32  D \sqrt{\frac{\mpwno{\epemb}}{D}} \ln \sqrt{\frac{\mpwno{\epemb}}{D}} \,\, \text{for}\, \mpwno{\epemb}\leq \frac{1}{\mpwno{32}D^2},\, \alpha \in \!\left(\frac12, 2\!
	\right)\nonumber \\
	\label{eq:eta3}
	&&\eta_\alpha \geq  6 \sqrt{D \mpwno{\epemb}} \quad 	\text{for } \alpha \in [2,\infty),
	\een
	where $D=d_\Sy D_\cat$.
 Anticipating that there may be problems with $\alpha$ around $\infty$ (this will become clear later) we also obtain a similar inequality for the Renyi $S_\infty$ entropy:
 \be
	\label{eq:nonstrict infty-overview}
	S_\infty(\rho_\Sy^0)\leq  S_\infty(\rho_\Sy^F) \mpwno{+}\eta_\infty,
	\ee
	with 
	\begin{equation}\label{eq:et infinity-overview}
	\eta_\infty= D_\cat \epemb.
	\end{equation}
 {\it 2) Removing term $\eta_\alpha$ by replacing 
 $\rho^F_\Sy$ with its approximated version $\sigma_\Sy$.} \vspace{0.2cm}\\
 The inequalities \eqref{eq:nonstrict T-overview} are not yet satisfactory, since we need strict inequalities, while the above ones are not only not strict, but also 
 there are terms $\eta_\alpha$.
 Fortunately, we want to show 
 the strict inequality 
 not for the state
  $\rho^F_\Sy$ itself, but for its 
 approximated version    $\sigma_\Sy$. 
 The state $\sigma_\Sy$ is just $\rho^F_\Sy$ with admixture of the maximally mixed state when it is far from it, and it is just the maximally mixed state, when it is close to it. 
 
 The idea now is to show that due to this admixture, $\sigma_\Sy$ will have larger values of entropies 
 than $\rho_\Sy^F$ by such an amount,
 that it will allow one to bypass the $\eta$'s, and obtain the needed strict inequalities. A crucial step is done in Prop. \ref{prop:strict}, where for $\epemb\leq \frac{1}{32 D^2}$  we obtain the following 
 inequalities
 \begin{align}
     T_\alpha(\rho_\Sy^0) < T_\alpha(\sigma^F_\Sy(\tilde\ep_T(\alpha))) \quad \text{\rm for } \alpha>0
 \end{align}
 where 
  \begin{widetext}
 \begin{align}
 \label{eq:epres-for-T}
     \tilde\ep_T(\alpha)\leq
	\left\{\bea{ll}
	\left( 96 D\, \frac{\epemb^\alpha}{\alpha} \right)^{\frac 13} =: \bar \ep_{Tmin}(\alpha) & \text{\rm for } \alpha\in(0,1/2]\vspace{0.1cm}  \\ 
	\Big(-1024\, D^2 \sqrt{\frac{\epemb}{D}} \ln\sqrt{\frac{\epemb}{D}}  \Big)^\frac13 =: \bar \ep_{Tmid} & \text{\rm for all } \alpha\in(1/2,2] \vspace{0.1cm}\\
	\big(96   \sqrt{D\epemb}  D^\alpha\big)^\frac13 =: \bar \ep_{Tmax}(\alpha) & \text{\rm for } \alpha \in(2,\infty). \vspace{0.1cm}
	\eea
	\right.
 \end{align}
 \end{widetext}
 It may appear like the end of the story is near.  We would need just to choose some $\epres$  
 which is larger than all of three values above | the Tsallis entropies on the right hand side will then just grow (as the entropies grow when we increase the admixture with identity, or if we replace with identity; see Lemma \ref{lem:up bound on S and T epsilons}), hence the inequalities will be still satisfied. Thus for so chosen $\epres$ we will obtain 
 \eqref{eq:T-klimesh} (where recall 
 that $\sigma_F$ depends on $\epres$
 as in \eqref{eq:sigma-S}.
 
However, there is a problem with $\alpha$ around $0$ and around $\infty$.  For those $\alpha$'s
the above bounds for $\tilde\ep_T $ become large, while we want them to tend to zero for $\epemb$ going to zero.
In other words, we do not have a 
uniform bound for $\tilde\ep_T $ for all $\alpha$'s at the moment. 

For $\alpha$ lying in those regions, we shall turn to Renyi entropies, and will prove 
inequality  \eqref{eq:S-klimesh}
rather than \eqref{eq:T-klimesh}. To deal with large values of $\alpha$ 
we shall use Eq. \eqref{eq:nonstrict infty-overview}   
in conjunction with Prop. \ref{prop:strict} to show that for $\alpha>1$
\begin{align} 
S_\alpha(\rho_\Sy^0) < S_\alpha(\sigma^F_\Sy(\ep_\infty(\alpha))) 
\end{align}
with
\begin{align}
    \ep_\infty(\alpha) \leq 4 \sqrt{\frac{\ln d_\Sy}{\alpha}+D \epemb} =: \bar \ep_{\infty}(\alpha)   \label{eq:ep infinity up bound-overview}.
\end{align}
To deal with values of $\alpha$ around 
zero, we prove in Prop. \ref{prop:strict} 
that 
for $ \alpha\in(0,1)$
\begin{align}
    S_\alpha(\rho_\Sy^0) < S_\alpha(\sigma^F_\Sy(\epzero(\alpha)))    \quad 
\end{align}
for
\begin{align}
 \epzero(\alpha)  \leq \left( \frac{d_\Sy-1}{d_\Sy}\right)^{\frac{1}{2\alpha}} =: \bar \ep_{0}(\alpha).  \label{eq:ep zero up bound-overview}   
\end{align}
In this latter case we had to use our assumption that the rank of the initial state $\rho_\Sy^0$ was not full.  
Note that the above $\ep$'s 
behave reasonably for large (or small) values of $\alpha$. 
The 
 $\bar\epzero (\alpha)$ goes to zero as $\alpha$ goes to zero,
 and  $\bar\ep_\infty$ 
 tends to $4\sqrt{D \epemb}$. 
 We shall then choose some $\alpha_{\min}$ and $\alpha_{\max}$,
 and below $\alpha_{\min}$ as well as above $\alpha_{\max}$ we shall 
 use the inequalities  for Renyi entropies \eqref{eq:S-klimesh} while 
 between $\alphazero$ and $\alphamax$, we shall use inequalities for Tallis entropy \eqref{eq:T-klimesh}. 
 The rest of the proof is to choose
 $\alpha_{\min}$ and $\alpha_{\max}$
 in such a way that the resulting common bound $\epres$ for all five types of $\ep$'s (i.e. three coming from \eqref{eq:epres-for-T},
 the other two from Eqs. \eqref{eq:ep infinity up bound-overview} and \eqref{eq:ep zero up bound-overview})
 is the smallest possible. 
 Finally, one may ask why we have not used the Renyi entropy everywhere. We did not use it, because it was easier to deal with Tsallis entropies for this region of $\alpha$ in Prop. \ref{prop:strict}. 
 
 Now we are ready to present the full proof of the theorem, with most of the technical lemmas relegated to the \Supp.

\begin{proof}
	Since $\rho_\Sy^0$ is not of full  rank, and the final state $\sigma_\Sy$ is by definition of full rank, 
	we need only to consider Klimesh conditions from Theorem \ref{thm:Klimesh} for $\alpha>0$.
	Consider first $\alpha>0$, $\alpha\not=1$.  If for some unitary $U$ we have 
	\be
	U \rho_S^0\ot\rho_{\cat}^0\ot\tauT_\G U^\dagger = \rho_{S \cat \G}^F
	\ee
	then
	\be
	f_\alpha(\rho^{\mpwno{0}}_{\Sy}\ot\rho_\cat^0\ot \tauT_\G)=f_\alpha(\rho_{\Sy\cat\G}^F ),
	\ee
	where $f_\alpha$ is defined in Sec. \ref{sec:entropies divergences def and properties}.
	Due to convexity/concavity of $f_\alpha$  and their  multiplicativity,  by  lemma \ref{lem:poor sub}, 
	putting  $A=\Sy \cat$ and $B=\G$ we obtain 
	\ben
	\label{eq:fa F Cat1}
	&&f_\alpha(\rho_{\Sy}^{\mpwno{0}}\ot \rho^0_\cat)\geq f_\alpha(\rho_{\Sy\cat}^F ), \quad \text{\rm for } \alpha>1 \\
	&&f_\alpha(\rho_{\Sy}^{\mpwno{0}}\ot \rho^0_\cat)\leq f_\alpha(\rho_{\Sy\cat}^F ), \quad \text{\rm for } \alpha\in(0,1).  \nonumber
	\een
	This implies, by definition of Tsallis entropy $T_\alpha$ [Eq. \eqref{eq:Tsalis def}], that for all $\alpha>0$, $\alpha\not =1$ we have
	\be
	\label{}
	T_\alpha(\rho_{\Sy}^{\mpwno{0}}\ot \rho^0_\cat)\leq T_\alpha(\rho_{\Sy\cat}^F).
	\ee
	We now use $\|\rho^F_{\Sy \cat} - \rho^F_\Sy\ot \rho^0_\cat\|_1 \leq \epemb$ 
	and the continuity lemma \ref{lem:cont} to find for $\alpha>0$
	\be
	T_\alpha(\rho_{\Sy\cat}^F ) \leq T_\alpha(\rho_{\Sy}^F\ot\rho_{\cat}^0) \mpwno{+}\eta_\alpha,
	\ee
	for all $\eta_\alpha$ satisfying
	\begin{align}
	&&\eta_\alpha \geq 6 D 	\left( \frac{\mpwno{\epemb}}{D}\right)^\alpha \quad  \text{for } \alpha \in (0,1/2] 
	\label{eq:eta1-2}
	\\
	&&\eta_\alpha \geq  -32  D \sqrt{\frac{\mpwno{\epemb}}{D}} \ln \sqrt{\frac{\mpwno{\epemb}}{D}} \, 
	\nonumber \\ &&\text{for } \mpwno{\epemb}\leq \frac{1}{\mpwno{32}D^2},\, \alpha \in \!\left(\!\frac12,2
	\!\right) 
	\label{eq:eta2-2}
	\\
	&&\eta_\alpha \geq  6 \sqrt{D \mpwno{\epemb}} \quad 	\text{for } \alpha \in [2,\infty),
	\label{eq:eta3-2}
	\end{align}
	where $D=d_\Sy D_\cat$.
	We rewrite the above equation back in terms of functions $f_\alpha$, which 
	gives
	\ben
	\label{eq:fa F cat2}
	&&f_\alpha(\rho_{\Sy\cat}^F ) \geq f_\alpha(\rho_{\Sy}^F\ot\rho_{\cat}^0) -(\alpha-1)\eta_\alpha \quad \text{\rm for } \alpha>1\\
	&&f_\alpha(\rho_{\Sy\cat}^F ) \leq f_\alpha(\rho_{\Sy}^F\ot\rho_{\cat}^0) -(\alpha-1)\eta_\alpha \quad \text{\rm for }  \alpha\in(0,1).\nonumber
	\een
	Then by using Eq. \eqref{eq:fa F Cat1} followed by the multiplicativity of the $f_\alpha$'s, we obtain from the above equations
	\ben
	\label{eq:fa F cat3}
	&&f_\alpha(\rho^0_{\Sy} ) \geq f_\alpha(\rho_{\Sy}^F) -\frac{(\alpha-1)}{f_\alpha(\rho_{\cat}^0)}\eta_\alpha \quad \text{\rm for } \alpha>1\\
	&&f_\alpha(\rho_{\Sy}^0 ) \leq f_\alpha(\rho_{\Sy}^F) -\frac{(\alpha-1)}{f_\alpha(\rho_{\cat}^0)}\eta_\alpha \quad \text{\rm for }  \alpha\in(0,1). 
	\een 
	Finally using $f_\alpha(p)\geq d^{1-\alpha}$ for $\alpha>1$ and 
	$f_\alpha(p)\geq 1$ for $\alpha\in(0,1)$ (These inequalities follow from setting $r=1, p=\alpha$ and $r=\alpha, p=1$ respectively in Eq. \eqref{eq:p r inequalities}, in Lemma \eqref{norm lem trangle inequality}) rewriting back in terms of $T_\alpha$'s we obtain
	\ben
	\label{eq:nonstrict T_old}
	&&T_\alpha(\rho_\Sy^0)\leq T_\alpha(\rho^F_\Sy)+ \eta_\alpha D_\cat^{\alpha-1} \quad \text{\rm for } \alpha\geq 1\\
	&&T_\alpha(\rho_\Sy^0)\leq T_\alpha(\rho^F_\Sy)+ \eta_\alpha \quad \text{\rm for } \alpha\in(0,1).	
	\een
	Here we have included the case $\alpha=1$, which is obtained by taking the limit $\alpha\to1$.\footnote{This extension of the domain of $\alpha$ for which the inequality holds, follows trivially using proof by contradiction and noting that the functions in Eq. \eqref{eq:fa F cat3} are continuous for $\alpha\in[1,\infty)$.} We can somewhat crudely unify this equation into 
	\be
	\label{eq:nonstrict T}
	T_\alpha(\rho_\Sy^0)\leq T_\alpha(\rho^F_\Sy)+ \eta_\alpha D_\cat^{\alpha} \quad \text{\rm for } \alpha>0.
	\ee
	Furthermore, \eqref{eq:fa F Cat1} implies that  for $\alpha>1$ 
	\be
	S_\alpha(\rho_\Sy^0\ot \rho_\cat^0) \leq  S_\alpha(\rho_{\Sy\cat}^F) 
	\ee
	and by taking limit $\alpha\to \infty$ we get 
	\be
	S_\infty(\rho_\Sy^0\ot \rho_\cat^0) \leq  S_\infty(\rho_{\Sy\cat}^F)
	\ee
	which by  Lemma \ref{lem:cont} and additivity of $S_\infty$ 
	gives
	\be
	\label{eq:nonstrict infty}
	S_\infty(\rho_\Sy^0)\leq  S_\infty(\rho_\Sy^F) \mpwno{+}\eta_\infty,
	\ee
	where 
	\begin{equation}\label{eq:et infinity}
	\eta_\infty= D_\cat \epemb.
	\end{equation}
	Let us now define  as in Proposition \ref{prop:strict}
	\be
	\sigma_\Sy^F(\ep)=
	\left\{
	\begin{array}{ll}
	\mpwno{ \id_\Sy/d_\Sy} &\text{when } \|\rho_\Sy^F -\id_\Sy/d\mpwno{_\Sy}\|_1  \mpwno{<} \ep \\
	(1-\ep) \rho_\Sy^F + \ep \id_\Sy/d\mpwno{_\Sy} & \text{when } \|\rho_\Sy^F -\id_\Sy/d\mpwno{_\Sy}\|_1 \geq \ep. \\
	\end{array}
	\right.
	\ee
	Eqs. \eqref{eq:nonstrict T} and \eqref{eq:nonstrict infty} by using Proposition  \ref{prop:strict} lead to the following conclusion:
	for 
	\ben
	&&
	\tilde\ep_T(\alpha)=
	\left\{\bea{ll}
	(16\eta_\alpha   D_\cat^\alpha d_\Sy^{\alpha-1})^\frac13 & \text{\rm for } \alpha \mpwno{ \geq } 1  \label{eq:tilde epsilon T}\\
	(16\eta_\alpha  D_\cat^\alpha  d_\Sy^{\alpha-1}\alpha^{-1})^\frac13 & \text{\rm for } \alpha\in(0,1) \label{eq:tilde epsilon T def 0}\\
	\eea
	\right. \\
	&&\ep_\infty(\alpha) = 4 \sqrt{\frac{\ln d_\Sy}{\alpha}+\eta_\infty} \quad \text{\rm for } \alpha > 1 \\
	&&\epzero(\alpha) =  \left( 1 - \frac1{d_\Sy}\right)^{\frac{1}{\mpwno{2}\alpha}}\quad \text{\rm for } \alpha\in(0,1),\label{eq:tilde epsilon 0}
	\een
	we have 
	\begin{align}
	 T_\alpha(\rho_\Sy^0)&\mpwno{\leq} T_\alpha(\sigma^F_\Sy(\tilde\ep_T(\alpha)))\\
	& \mpwno{- \min\left\{D_\cat^\alpha  \eta_\alpha,\,\,  T_\alpha(\id/d_\Sy) - T_\alpha(\rho^0_\Sy) \right\} } \quad \text{\rm for } \alpha>0 \nonumber\\
	S_\alpha(\rho_\Sy^0)& \mpwno{\leq} S_\alpha(\sigma^F_\Sy(\ep_\infty(\alpha)))\\
	&   \mpwno{- \min\left\{D_\cat^\alpha  \eta_\alpha,\,\,  \ln d_\Sy - S_1(\rho^0_\Sy) \right\} } \quad \text{\rm for }  \alpha>1 \nonumber\\
	 S_\alpha(\rho_\Sy^0)& \mpwno{\leq} S_\alpha(\sigma^F_\Sy(\epzero(\alpha))) \\
	 & \mpwno{- \frac{1}{2} \ln\left(\frac{d_\Sy}{d_\Sy-1}\right) }   \quad \text{\rm for }  \alpha\in(0,1),
	\end{align}
	from which we achieve 
	\ben
	&& T_\alpha(\rho_\Sy^0) < T_\alpha(\sigma^F_\Sy(\tilde\ep_T(\alpha))) \quad \text{\rm for } \alpha>0\\
	&& S_\alpha(\rho_\Sy^0) < S_\alpha(\sigma^F_\Sy(\ep_\infty(\alpha)))    \quad \text{\rm for }  \alpha>1 \\
	&& S_\alpha(\rho_\Sy^0) < S_\alpha(\sigma^F_\Sy(\epzero(\alpha)))    \quad \text{\rm for }  \alpha\in(0,1).
	\een

	Let us now insert explicitly the $\eta$'s from Eqs. \eqref{eq:eta1-2}-\eqref{eq:eta3-2}  and Eq. \eqref{eq:et infinity} into Eqs. \eqref{eq:tilde epsilon T}-\eqref{eq:tilde epsilon 0}. For
	\begin{equation}
	\epemb\leq \frac{1}{32 D^2},
	\end{equation}
	we achieve the upper bounds
\begin{widetext}
	\begin{align} \label{eq:ep tilde up bound}
\tilde\ep_T(\alpha)\leq
	\left( 96 D\, \frac{\epemb^\alpha}{\alpha} \right)^{\frac 13} &=: \bar \ep_{Tmin}(\alpha)\quad  \text{\rm  for } \alpha\in(0,1/2]\vspace{0.1cm}  \\ 
	\Big(-1024\, D^2 \sqrt{\frac{\epemb}{D}} \ln\sqrt{\frac{\epemb}{D}}  \Big)^\frac13 &=: \bar \ep_{Tmid} \quad  \text{\rm  for } \alpha\in(1/2,2] \vspace{0.1cm}\\
	\big(96   \sqrt{D\epemb}  D^\alpha\big)^\frac13 &=: \bar \ep_{Tmax}(\alpha) \quad  \text{\rm  for } \alpha \in(2,\infty) \\
	\ep_\infty(\alpha) \leq 4 \sqrt{\frac{\ln d_\Sy}{\alpha}+D \epemb} &=: \bar \ep_{\infty}(\alpha)  \quad \text{\rm  for } \alpha\in [1, \infty) \label{eq:ep infinity up bound}\\
	\epzero(\alpha)  \leq \left( \frac{d_\Sy-1}{d_\Sy}\right)^{\frac{1}{2\alpha}} &=: \bar \ep_{0}(\alpha) \quad \text{\rm  for } \alpha\in(0, 1), \label{eq:ep zero up bound}
	\end{align}
\end{widetext}
	where $D=d_\Sy D_\cat$. 
	We now divide the set $(0,\infty)$ into five subintervals (some of which may be empty). For $\alphazero\in(0,1)$, $\alphamax\in[2,\infty)$, we have
	$(0,\infty)=(0,\alphazero]\cup (\alphazero,1/2] \cup (1/2, 2] \cup [2, \alphamax] \cup (\alphamax,\infty).$
	For each of these intervals, we compute upper bounds on our epsilons. Specifically, from Eqs. \eqref{eq:ep tilde up bound}, \eqref{eq:ep infinity up bound}, \eqref{eq:ep zero up bound}, we observe that:
	\begin{align}
	&\epzero(\alpha) \leq \bar\ep_0(\alphazero)  \quad \forall \, \alpha\in(0,\alphazero), \forall\, \alphazero\in(0,1)\\
	&\tilde\ep_T(\alpha)\leq \nonumber\\
	&\left\{\bea{ll}
	\bar\ep_{Tmin}(\alphazero) & \forall \, \alpha\in(\alphazero, 1/2], \forall\, \alphazero\in(0,1/2]\\
	\bar\ep_{Tmid} & \forall\, \alpha \in (1/2, 2]  \\
	\bar\ep_{Tmax}(\alphamax) & \forall\, \alpha\in[2,\alphamax], \forall \,\alphamax\in[2,\infty) \\
	\eea
	\right.\\
	&\ep_\infty(\alpha)\leq \bar\ep_\infty(\alphamax) \quad\forall\, \alpha\in(\alphamax,\infty), \forall\, \alphamax\in[1,\infty).
	\end{align}
	Now we define $\epres$ as any value satisfying
	\begin{equation}\label{eq:ep res}
	\epres(\alphazero, \alphamax) \geq \max \left\{ \varepsilon_\textup{min}(\alphazero), \varepsilon_\textup{max}(\alphamax), \bar\ep_{Tmid}  \right\}
	\end{equation}
	where
	\ba 
	\varepsilon_\textup{min}(\alphazero)& = \max \left\{\bar\ep_{Tmin}(\alphazero), \bar\ep_0(\alphazero)\right\}\\
	\varepsilon_\textup{max}(\alphamax)&= \max \left\{  \bar\ep_{Tmax}(\alphamax), \bar\ep_\infty(\alphamax)  \right\}.
	\ea 
	Thus using Lemma \ref{lem:up bound on S and T epsilons}, we have  
	\begin{widetext}
	\begin{alignat}{3}
	S_\alpha(\rho_\Sy^0) &< S_\alpha\big(\sigma^F_\Sy(\bar\ep_0(\alphazero))\big) &&< S_\alpha\big(\sigma^F_\Sy(\epres(\alphazero,\alphamax))\big)\quad &&\forall\, \alpha\in(0,\alphazero) \\
	T_\alpha(\rho_\Sy^0) &< T_\alpha\big(\sigma^F_\Sy(\bar\ep_{Tmin}(\alphazero))\big) &&< T_\alpha\big(\sigma^F_\Sy(\epres(\alphazero,\alphamax))\big) &&\forall\, \alpha\in(\alphazero,1/2) \label{line:first T in T<T eq}\\
	T_\alpha(\rho_\Sy^0) &< T_\alpha\big(\sigma^F_\Sy(\bar\ep_{Tmid})\big) &&< T_\alpha\big(\sigma^F_\Sy(\epres(\alphazero,\alphamax))\big) && \forall\, \alpha\in[1/2,2] \\
	T_\alpha(\rho_\Sy^0) &< T_\alpha\big(\sigma^F_\Sy(\bar\ep_{Tmax}(\alphamax))\big) &&< T_\alpha\big(\sigma^F_\Sy(\epres(\alphazero,\alphamax))\big) &&\forall\, \alpha\in[2,\alphamax) \\
	S_\alpha(\rho_\Sy^0) &< S_\alpha\big(\sigma^F_\Sy(\bar\ep_{\infty}(\alphamax))\big) &&< S_\alpha\big(\sigma^F_\Sy(\epres(\alphazero,\alphamax))\big) &&\forall\, \alpha\in[\alphamax,\infty)
	\end{alignat}
	\end{widetext}
	holds for all $\alphazero\in(0,1)$, $\alphamax\in(2,\infty)$.\footnote{Note that the range of $\alpha$ for which Eq. \eqref{line:first T in T<T eq} holds is the empty set if $\alphazero$ is large enough. Under such circumstances, this equation contains no information.}
	
	Thus for any particular choice of $\alphazero\in(0,1)$ and  $\alphamax\in(2,\infty)$, $\epres(\alphazero, \alphamax) $ is such that the Klimesh conditions are satisfied, so that for any $\rho_\Sy^0$ there exists catalyst $\tilde \rho_\cat$ 
	such that 
	\be
	\rho_\Sy^0 \ot \tilde \rho_\cat \succ  \sigma_\Sy^F(\epres)  \ot  \tilde \rho_\cat
	\ee
	
	Our next aim is to find an explicit expression for $\epres(\alphazero,\alphamax)$ with the aim of choosing the parameters $\alphazero,\alphamax$, so that $\epres(\alphazero,\alphamax)$ is not too large.  In lemma \ref{lem:alpha-min-max}
	we show that the $\epres$ given in the statement of the theorem upper bounds 
	$\epres(\alphazero,\alphamax)$ 
	for some $\alphazero$ and $\alphamax$. 
	This finalises the proof. 
	\end{proof}

To see how to write Theorem \ref{thm:noemb} in the form of Theorem \ref{thm:noemb physical}, see Corollary \ref{rem:tigher bound on ep emb} in the \Supp.

\subsection{Introduction to the Quasi-Ideal Clock and Proof of Theorem \ref{Thm:Implementation with Quasi-Idela clock}}

In the following subsection (\ref{Overview of the quasi-ideal clock}), we start with a brief overview of the properties of the Quasi-Ideal Clock. These are necessary for the proof of Theorem \ref{Thm:Implementation with Quasi-Idela clock}, which is in subsection \ref{Sec: proof of thm 2: Quasi-Ideal clock bound}.

\subsubsection{Brief overview of the Quasi-Ideal Clock}\label{Overview of the quasi-ideal clock}
\label{subsec:clock}

In this section we will recall the clock 
construction from \cite{WSO} which will 
be subsequently used to prove Proposition \ref{prop:clock-error}, which in turn will lead to the proof of Theorem \ref{Thm:Implementation with Quasi-Idela clock}.
\cblack

The time independent total Hamiltonian over system $\rho_\A\otimes\rho_\cl$ is
\be 
\hat H_{\A\cl}=\hat H_\A\otimes\id_\cl+ \hat H^\textup{int}_\A\otimes\hat V_d +\id_\A\otimes\hat H_\cl, 
\ee 
where $\hat H_\A$ is the system Hamiltonian which commutes with the target unitary $U^\text{target}_\A$. The term $\hat H^\textup{int}_\A$ encodes the target unitary via the relation $U^\textup{target}_\A=\me^{-\mi \hat H^\textup{int}_\A}$, 
 with 
\begin{align}
\label{eq:target-unitary}
    H^\textup{int}_\A= \sum_{n=1}^{d_\A} \Omega_n |n\>\< n|.
\end{align}

\cblack

The clock's free Hamiltonian, $\hat H_\cl$ is a truncated harmonic oscillator Hamiltonian. Namely, $\hat H_\cl=\sum_{n=0}^{d-1}\omega n \ketbra{n}{n}$. The free evolution of any initial clock state under this Hamiltonian has a period of $T_0=2\pi/\omega$, specifically, $\me^{-\mi T_0 \hat H_\cl}\rho_\cl\me^{\mi T_0\hat H_\cl} = \rho_\cl$ for all $\rho_\cl$. The clock interaction term $\hat V_d$, takes the form,
\be
\label{eq:potential-op}
\hat{V}_d = \frac{d}{T_0} \sum_{k=0}^{d-1} {V}_d(k) \ketbra{\theta_k}{\theta_k},
\ee 
where the basis $\{\ket{\theta_k}\}_{k=0}^{d-1}$ is the Fourier transform of the energy eigenbasis $\{\ket{n}\}_{n=0}^{d-1}$. The function ${V}_d:\rr\mapsto\rr$ 
will be called potential and \cblack is defined by 
\be 
\label{eq:potential-Vd}
V_d(y)= \frac{2\pi}{d} V_0\left( \frac{2\pi}{d}(y -y_0) \right),
\ee 
where $V_0$ is an infinitely differentiable periodic function of period $2\pi$ centered around $0$ (so that $V_d$ has period $d$ and is centered around $y_0$).
 A lot of results hold for this general form of potential. To obtain all the results 
we shall need a specialized form of potential given by  
\begin{align}
\label{eq:potential}
    V_0(x) = A_c \cos^{2n}\left(\frac{x}{2}\right),
    \quad \text{with} \quad A_c=\frac{2^{2n}}{2\pi{2n \choose n}},
\end{align}
and where $n$ will be later be taken to be a suitable function of clock dimension (specifically, later we shall take $n\sim d^{1/4}$). Here $A_c$ is normalization constant so that $\int_{-\pi}^\pi V_0(x) \dt x=1$.
It is important that $V_0$ has exponentially decaying tails 
\begin{align}
\label{eq:potential-tail-1}
    \tilde\epsilon_V:=\int_{-2\pi(1-\potdelta)}^{-2 \pi \potdelta}V_0(x)\dt x \leq  \frac{1}{\delta_V} e^{-\delta_V^2 n}, \quad \text{for } \potdelta\in(0,\pi).
\end{align}
The bound in \cite{WSO} is tighter and does not diverge as $\delta_V\to 0^+$, 
but the present one is just enough | as we anyway care just about scaling for the proof (see lemma \ref{lem:tail-potential}).

Recall that for the Quasi-Ideal clock, the initial state is pure $\rho_\cl=\ketbra{\Psi_\textup{nor}(k_0)}{\Psi_\textup{nor}(k_0)},$ where 
\begin{align}\label{gaussianclock}
\ket{\Psi_\textup{nor}(k_0)} &= \sum_{\mathclap{\substack{k\in \mathcal{S}_d(k_0)}}}\psi(k_0;k) \ket{\theta_k},\\
\psi(k_0;x) &= A e^{-\frac{\pi}{\sigma^2}(x-k_0)^2} e^{i2\pi n_0(x-k_0)/d}, \quad x\in\rr.
\end{align}
with $\sigma \in (0,d)$, $n_0 \in (0,d-1)$, $k_0\in\rr$, $A \in \rr^+$, and $\mathcal{S}_d(k_0)$ is the set of $d$ integers closest to $k_0$, defined as
\begin{align}\label{eq: mathcal S def}
\mathcal{S}_d(k_0) = \left\{ k \; : \; k\in \mathbb{Z} \text{   and  }  -\frac{d}{2} \leq  k_0-k < \frac{d}{2} \right\}.
\end{align}
Note that for $k$ larger than $d-1$ or smaller than $0$,
we define $\theta_k$ as 
$\theta_{k \mod d}$.
The quantity $A$ is defined so that the state is normalised, namely
\be\label{eq:A normalised}
A=A(\sigma;k_0)=\frac{1}{\sqrt{\sum_{k\in \mathcal{S}_d(k_0)} \me^{-\frac{2\pi}{\sigma^2}(k-k_0)^2}}}=\bo\left(\!\! \left(\frac{2}{\sigma^2}\right)^{\!\!\frac14} \right)\!.
\ee
The parameter $n_0$ is approximately the mean energy of the clock,
and for a good clock performance, it should be not too close to $0$ nor to $d$. 
We will later set it to $(d-1)/2$ as suggested in \cite{WSO}; see Def. 1.

\subsection{Small error on the clock}
Here we prove a proposition that is crucial to prove Theorem \ref{Thm:Implementation with Quasi-Idela clock}. The proposition states that for the clock described above, the state of the clock acquires a small error. 
\begin{prop}
\label{prop:clock-error}
    Consider the Quasi-Ideal clock described above. Consider times $t_1,t_2$
    satisfying $0<t_1<t_2<T_0$. Then for the potential $V_d$ determined by Eqs. \eqref{eq:potential-Vd} and \eqref{eq:potential}
    with $y_0=(t_1-t_2) d/T_0$ and $n=\lceil d^{1/4}\rceil$,
    we have:
    \begin{align}
    \label{eq:prop-clock-error}
        \| \rho^F_{\cl}(t)-\rho^0_\cl(t)\|_1 \leq 
    \frac{1}{t_2-t_1} poly(d) e^{-c_2 d^{1/4}}
    \end{align}
    where $c_2=\min\{\frac{1}{64 \pi}, (t_2-t_1)^2/(32T_0^2) \}$. 
\end{prop}
\begin{remark}
The unbounded from above factor $1/(t_2-t_1)$ is not necessary, and in the original paper \cite{WSO} it did not appear. Here it is a price for a simpler proof of 
potential concentration properties. 
\end{remark}
\begin{proof}
Let us set arbitrary times $t_1$ and $t_2$. 
We want to show that we can choose the potential in the clock described in Sec. \ref{subsec:clock} so that, apart from times near the bondaries (i.e. those \emph{not} satisfying $0<t_1<t_2<T_0$), it will be 
close to the free evolution of the clock state.
In other words, the evolution may be different in the ``interaction zones''. 
The potential has been already determined, with two free parameters | the peak position $y_0$ and $n$, determining the concentration of the potential around the peak. 
As will be later argued,
the clock state (we shall call it pointer) will approximately travel around the circle with linear speed $d/T_0$,
so that to times $t_1$ and $t_2$ 
there correspond positions $y_1=t_1 d/T_0$ and $y_2=t_2 d/T_0$ (see Fig. \ref{fig:clock}).
\begin{figure}
\includegraphics[scale=0.4]{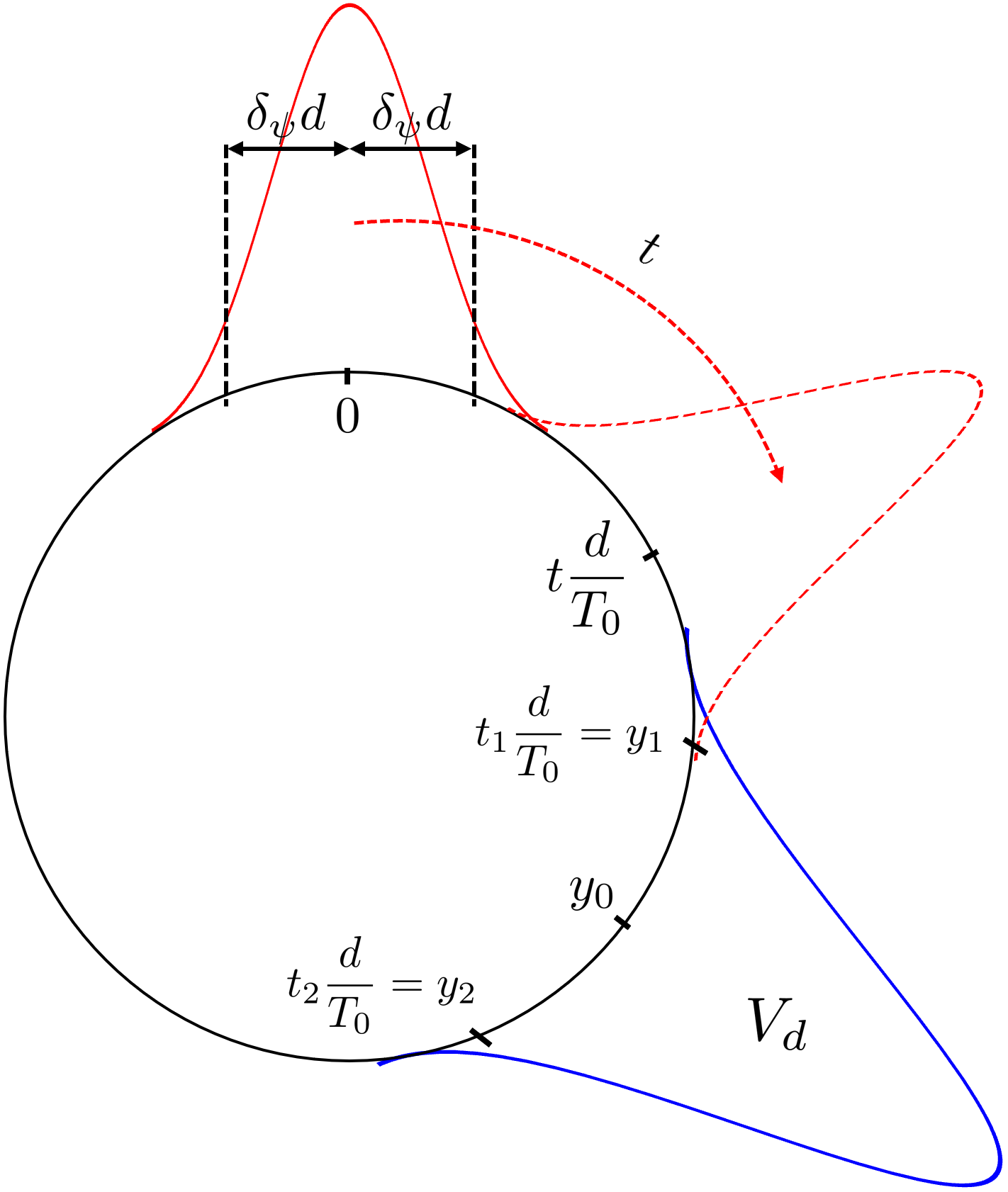}
		\centering
		
		                \caption{ \textbf{Dynamics of the clock.} The Circumference of the circle is $d$. The red profile represents the amplitudes of the clock state (called pointer) with the weight concentrated within $\pm\psidelta d$ from the center. It moves around  approximately with speed $d/T_0$. The potential $V_d$ has peak at $y_0$.	The positions $y_1$ and $y_2$ are determined by the times $t_1$ and $t_2$ and denote places which the peak of the clock state will reach at times $t_1$ and $t_2$. \label{fig:clock} }
\end{figure}
Since the  interaction can take place on the interval
 where potential is non-negligible, we shall aim to have 
potential concentrated in the area between $y_1$ and $y_2$. To this end we  choose the peak of potential to be in the middle between $y_1$ and $y_2$:
\begin{align}
    y_0=\frac{y_1+y_2}{2}=
    \frac{t_2+t_2}{2} \frac{d}{T_0}
\end{align}
The concentration parameter $n$ will be chosen to be chosen later. With such potential we want to estimate the following quantity:
\begin{equation}
    \| \rho^F_{\cl}(t)-\rho^0_\cl(t)\|_1.
\end{equation}
To this end we shall evaluate fidelity and will use 
$\|\rho-\sigma\|\leq 2 \sqrt{1- F(\rho,\sigma)^2}$. 
In \cite{WSO}, (see the proof of Lemma 10.0.3; page 192)
 after a bit of algebra, the following was obtained 
\be 
	\rho_{\cl}(t)= \sum_{n=1}^{d_\A} \rho_{n,n}(0)\ketbra{\bar\Phi_n(t)}{\bar\Phi_n(t)}_\cl, 
	\ee 
	where $\{\rho_{n,n}(0)\}_n$ are the eigenvalues of the initial system state $\rho_\A$, and thus also constitute a set of normalised probabilities. $\ket{\bar\Phi_n(t)}_\cl$ is defined by,
	\begin{align} 
	\ket{\bar\Phi_n(t)}_\cl &= \hat \Gamma_n(t) \ket{\Psi_\textup{nor}(k_0)}_\cl, \\
	 \hat \Gamma(t,\Omega_n)&=\me^{-\mi t( \Omega_n \hat V_d+\hat H_\cl)},
	\end{align}
	where $\{\Omega_n\in[-\pi,\pi]\}_{n=1}^{d_\A}$ are a set of phases which determine the target unitary one wishes to apply (see Eq. \eqref{eq:target-unitary}). Using $F(\rho,\psi)=\<\psi|\rho|\psi\>$ we get 
\begin{align}
\label{eq:fidelity-clock-new}
	&F\left( \rho_\cl^0(t), \rho_\cl(t) \right)  \\
	&=\sum_{n=1}^{d_\cl} \rho_{n,n} (0) \left| \braket{\Psi_\textup{nor}(k_0)| \me^{\mi t \hat H_\cl} \hat \Gamma(t,\Omega_n) |\Psi_\textup{nor}(k_0)} \right|^2 \nonumber \\
& \geq \min_{\Omega\in[-\pi,\pi]}  \left| \braket{\Psi_\textup{nor}(k_0)| \me^{\mi t \hat H_\cl} \hat \Gamma(t,\Omega) |\Psi_\textup{nor}(k_0)} \right|^2.
\end{align}
We thus aim to show that the following states 
\begin{equation}
    \me^{-\mi t \hat H_\cl}  \ket{\Psi_\textup{nor}(k_0)}
    \quad \text{and} \quad  \hat\Gamma(t,\Omega) \ket{\Psi_\textup{nor}(k_0)}
\end{equation}
have overlap close to 1, irrespectively of phase $\Omega$.
To this end we shall use core theorems in \cite{WSO} 
(Theorems: VIII.1, page 19 and IX.1, page 35).
They say that 
\begin{itemize}
    \item[(i)]  under evolution $\me^{\mi t \hat H_\cl}$
the state $|\Psi_\textup{nor}(k_0)\>$ up to a small correction 
evolves in a trivial way - namely its peak undergoes translation. 
\item[(ii)]  under evolution $\Gamma(\Omega,t)$
the above translation occurs too, but in addition the $k$-th amplitude  
of the state acquires phase equal to the potential integrated over the interval that $k$ travelled.
\end{itemize}
More specifically for $n_0=(d-1)/2$ (c.f. Def. 2 in \cite{WSO}) which means that $n_0$ which has interpretation of average energy, is not too close to 0 or to the maximal energy, we have  
\begin{align}
& \me^{-\mi t \hat H_\cl}  \ket{\Psi_\textup{nor}(k_0)} \\
&= \ket{\Psi_\textup{nor}(k_0+td/T_0)}+ \ket{\varepsilon_c}\\
&= \sum_{k\in\mathcal{S}_d (k_0+td/T_0)} \psi(k_0+td/T_0;k)\ket{\theta_k} + \ket{\varepsilon_c},\label{eq:free evollution of Quasi ideal}\\
&\Gamma(t,\Omega) \ket{\Psi_\textup{nor}(k_0)}\\
 &= \ket{\bar\Psi_\textup{nor}(k_0+td/T_0, t d/T_0)}+ \ket{\varepsilon_\nu} \\
&= \sum_{k\in\mathcal{S}_d (k_0+td/T_0)} \me^{-\mi  \phi_k(t)} \psi(k_0+td/T_0;\,k)\ket{\theta_k} + \ket{\varepsilon_\nu}  , \label{eq:non free evollution of Quasi ideal}
\end{align} 
where  the phase acquired by $k$'th amplitude is given by
\begin{align}
\label{eq:phase}
\phi_k(t) &=\Omega \int_{k-t d/T_0}^k dy  V_d(y).
\end{align}
Now, for $\sigma=\sqrt d$ 
and for $n=\frac12\lceil d^{1/4}\rceil$ in the potential form 
of Eq. \eqref{eq:potential} we have 
(see Lemma \ref{lem: ep-nu-estimate})
\begin{align}\label{eq:epsilon-nu-epsilon-c}
	\||\varepsilon_\nu\>||_2=:
	\varepsilon_\nu&\lesssim
    t\, poly(d) e^{- \frac{1}{16\pi} d^{1/4}}
	\nonumber \\
	\||\varepsilon_c\>||_2=: \varepsilon_c&=\bo \left( poly(d_\cl)\,  \me^{-\frac{\pi}{4} d_\cl} \right)
\end{align}
Actually, only estimate on $\varepsilon_\nu$
depends on the potential form. 
The bound for $\varepsilon_c$ 
holds for arbitrary periodic $V_0$. 

Now, since for normalized $|\psi\>,|\phi\>$ and $|x\>,|y\>$
such that $\||x\>\|, \||y\>\|\leq 1$ we have 
$|(\<\psi|+\<x|) | (|\phi\>+|y\>)|^2\geq |\<\psi|\phi\>|^2 - 3 \||x\>\|-3 \||y\>\|$, 
we obtain 
from \eqref{eq:fidelity-clock-new}
\begin{align}
    \label{eq:fidelity-clock-new2}
	&F\left( \rho_\cl^0(t), \rho_\cl(t) \right) \geq \\
	&\min_{\Omega\in[-\pi,\pi]} 
	\left|\<\Psi_\textup{nor}(k_0+td/T_0)|\bar\Psi_\textup{nor}(k_0+td/T_0, t d/T_0) \rangle \right|^2\nonumber\\
	& - 3 \varepsilon_{\nu}
	- 3 \varepsilon_c
\end{align}
We thus have the situation, that
$\Psi_\textup{nor}(k_0+td/T_0)$ 
and 
$\bar \Psi_\textup{nor}(k_0+td/T_0)$ 
have the peak moving  around with speed $d/T_0$, while  $\bar\Psi_\textup{nor}$
in addition acquires phase. 
	We now write explicitly the above inner product 
	\begin{align}
	\label{eq:Delta-total}
	    &\Delta(\Omega):=\<\Psi_\textup{nor}(k_0+td/T_0)|\bar\Psi_\textup{nor}(k_0+td/T_0, t d/T_0) \rangle = \\
	   & 
	    \sum_{k\in\mathcal{S}_d(k_0+td/T_0)} \me^{-\mi\Omega\int_{k-td/T_0}^k dy V_d(y)} \left|\psi_\textup{nor}(k_0+td/T_0;\,k)\right|^2.
	\end{align}
and the goal is to show that it is close to $1$ independently of $\Omega$
for all times before $t_1$ 
and after $t_2$. 

The idea to prove this (along the lines of \cite{WSO}) is the following way (see also Fig. \ref{fig:clock}).
Let us denote the position of the peak of the pointer by $k_0(t)=k_0+td/T_0$.
We shall set the initial pointer's peak to be at 12 O'clock, i.e. $k_0=0$.
First,  since Gaussians 
have rapidly decaying tails we have (see Lemma \ref{lem:gaussian-tail})
\begin{align}
\label{eq:epsilon-LR}
  \varepsilon_{LR}:=\sum_{k:|k-k_0|\geq \psidelta \, d} \left|\psi_\textup{nor}(k_0;\,k)\right|^2
    \leq poly(d) e^{-\psidelta^2 d},
\end{align}
for $\psidelta>0$. 
We can restrict the sum in \eqref{eq:Delta-total}
and leave only 
$k's$ lying within the interval 
$k_0(t)\pm \psidelta d$. Since the pointer's peak travels with speed $d/T_0$, i.e.  
$k_0(t)= t d/T_0$, 
 for times ``before the interaction'', i.e.  $t\leq t_1$, 
those $k$'s will be to the left of 
$y_1+\psidelta d$, and for times ``after interaction'', i.e.  $t \geq t_2$ 
they will be to the right of $y_2-\psidelta d$.

The potential is strongly peaked around $y_0$ which sits between those two positions.
Thus for times $t\leq t_1$ those $k's$ were travelling within the tail of the potential, 
while for times $t\geq t_2$, all those $k$'s 
have passed the ``body'' of the potential. 
Now, according to \eqref{eq:phase} 
acquired by $|\theta_k\>$ 
are given by $\Omega$ times the integral of the potential  over the interval the $k$ travelled.  Thus for times $t\leq t_1$, the acquired phase
for all those $k$'s will be close to zero (less than $\Omega \tilde \ep_V$, where $\ep_V$ is the area of the tail) 
while for times $t\geq t_2$, the phase will be close to $\Omega$ (larger than  $\Omega(1-\tilde\ep_V$), since the total integral of the potential is $1$). 

We shall now write it rigorously.  
First of all,  we shall cut the tails of the pointer. We split $\Delta(\Omega)$ as 
\begin{align}
\label{eq:split-delta}
    \Delta(\Omega)=\Delta_C(\Omega) + \Delta_{LR}(\Omega)
\end{align}
with 
\begin{align}
    \Delta_C(\Omega) := \qquad & \\
	   \sum_{|k-td/T_0|\leq \psidelta\, d }& \me^{-\mi\Omega\int_{k-td/T_0}^k dy V_d(y)} \left|\psi_\textup{nor}(td/T_0;\,k)\right|^2, \\
	\Delta_{LR}(\Omega):= \qquad &  \\ 
	    \sum_{\psidelta d <|k-td/T_0|\leq d/2 }& \me^{-\mi\Omega\int_{k-td/T_0}^k dy V_d(y)} \left|\psi_\textup{nor}(td/T_0;\,k)\right|^2.
\end{align}
We then have 
\begin{align}
\label{eq:tail-for-Delta}
    |\Delta_{LR}(\Omega)|\leq &\\
    \sum_{k:|k-k_0|\geq \psidelta \, d} & \left|\psi_\textup{nor}(k_0;\,k)\right|^2
    =\varepsilon_{LR}
    \leq poly(d) e^{-\psidelta^2 d},
\end{align}
where the tail bound is, for completeness, 
given in lemma \ref{lem:gaussian-tail}. 
By the tail estimate \eqref{eq:tail-for-Delta},
for  $\varepsilon_{LR} \leq 1$ we then get 
\begin{align}
\label{eq:Delta-Delta-c}
    |\Delta(\Omega)|^2 \geq |\Delta_C(\Omega)|^2 - 2 \varepsilon_{LR},
\end{align}
so that it is enough to show that $\Delta_C(\Omega)$ is close to $1$ irrespective of $\Omega$. 

We shall now bound the phase $\phi_k(t)$
for our restricted set of $k$'s.

{\it Times ``before interaction''.} 
Consider time $t\leq t_1$, and as said we are 
restricting to $k$ such that $|k-t d/T_0|\leq \psidelta d$. This implies 
\begin{align}
    &k\leq t d/T_0 + \psidelta d \leq  t_1 d/T_0 + \psidelta d =y_1 +\psidelta d, \nonumber \\
    &k - t d/T_0 \geq - \psidelta d,     
\end{align}
which gives (see Fig. \ref{fig:phase})

\begin{figure*}
	\includegraphics[scale=0.55]{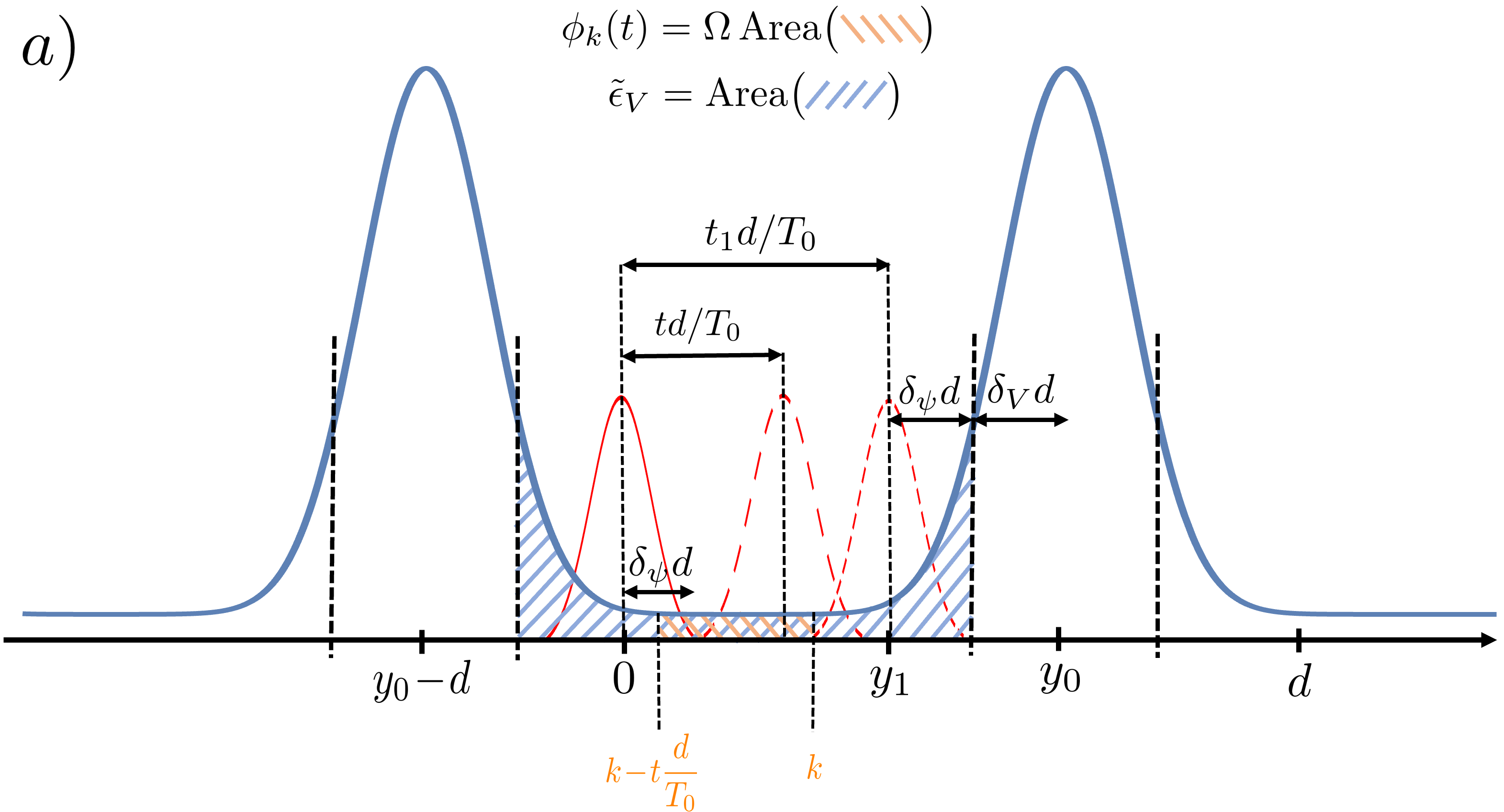}\vspace{1cm}\\
		\includegraphics[scale=0.55]{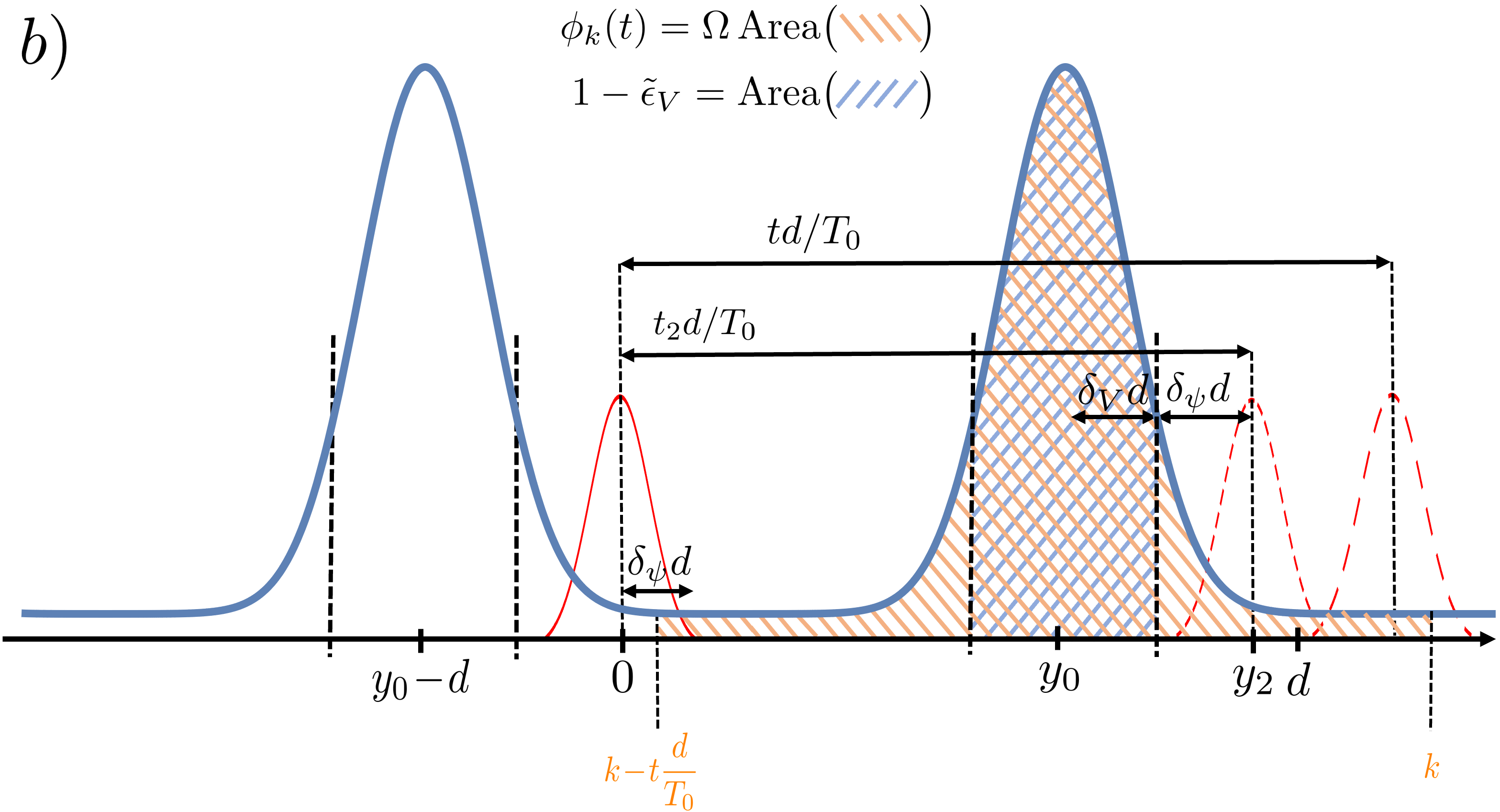}
		\centering
                \caption{ \textbf{Acquiring phase by the clock states $|\theta_k\>$.} The quantities \textup{Area}(\textbf{ \textcolor{orange}{\!\!\textbf{\textbackslash\textbackslash\textbackslash\textbackslash}}}) and  \textup{Area}(\textbf{ \textcolor{blue}{\!\!\textbf{////}}}) are the areas of the orange and blue regions respectively. 	a) Times $t\leq t_1$. For such times, the peak of the pointer	will travel with speed $d/T_0$ up to $y_1$, so that, the ``body'' of the pointer (i.e. the part within $\pm\psidelta d$ from the peak) will be always within the blue area. We consider arbitrary $k$ within the body of the pointer at time $t$ (the left dashed pointer), and its past position $k-td/T_0$. The phase $\phi_k(t)$  acquired by $|\theta_k\>$ is proportional to yellow area, which is contained in the blue one, which in turn is the tail of the potential and therefore small. b) Times $t\geq t_2$. In this case, the pointer is initially before the blue area (the ``body'' of the potential) and ends up after it. Thus any $k$ from the body of the pointer at time $t$ had to travel through the blue area from its past position $k-t d/t_0$. The acquired phase $\phi_k(t)$ is proportional to the yellow area. The latter, for any $k$ is larger than the blue one (body of the potential) and therefore the phase is close to $\Omega$. \label{fig:phase} }
\end{figure*}

\begin{align}
    &\phi_k(t)  =\Omega \int_{k-t d/T_0}^kV_d(y)\dt y
    \leq  \Omega
    \int_{-\psidelta d}^{y_1+\psidelta d} V_d(y)\dt y \\
    &     \leq \Omega 
\int_{y_0+\potdelta d-d}^{y_0-\potdelta d} V_d(y)\dt y = \Omega \int^{-2 \pi \potdelta }_{-2 \pi(1-\potdelta)} V_0(x)  \dt x = \tilde\epsilon_V
\end{align}
 where $\potdelta$ is determined by
 $\potdelta d=y_0-y_1-\psidelta d$.
Denoting $t_2-t_1=\Delta t$, we have $\potdelta + \psidelta=\Delta t/T_0$, and we may choose 
$\potdelta=\psidelta=\Delta t/(4T_0)$. 
We know from \eqref{eq:potential-tail-1}
that $\tilde\epsilon_V$ decays exponentially in the concentration parameter $n$ of the potential.  

{\it Times ``after interaction''.}
Now we consider $t\geq t_2$. Similarly as before, the condition $|k-t d/T_0|\leq \psidelta$ implies
\begin{align}
    &k\geq t d/T_0-\psidelta d \geq  t_2 d/T_0-\psidelta d  = y_2 - \psidelta d, \nonumber \\
    & k-t d/T_0\leq \psidelta d,
\end{align}
hence 
\begin{align}
    \phi_k(t) &= 
    \Omega \int_{k-t d/T_0}^k V_d(y)\dt y
    \geq  \Omega
    \int_{\psidelta d}^{y_2-\psidelta d} V_d(y)\dt y\\
    &
    \geq \Omega 
\int_{y_0-\potdelta d}^{y_0+\potdelta d} V_d(y)\dt y = \Omega \int_{-2\pi \potdelta }^{2\pi \potdelta} V_0(x)  \dt x \\
&= \Omega (1-\tilde\epsilon_V).
\end{align}
Altogether, for times $t\leq t_1$ 
and $t\geq t_2$, respectively, we have obtained
\begin{equation}
\label{eq:phik-equal}
    |\phi_k- 0|\leq  \Omega \tilde\epsilon_V,\quad |\phi_k - \Omega| \leq \Omega \tilde\epsilon_V.
\end{equation}
We can now come back to the estimation of $\Delta_C$. Denoting 
\begin{align}
    a_k= \left|\psi_\textup{nor}(td/T_0;\,k)\right|^2,
\end{align}
 we have 
 \begin{align}
     \Delta_C=\sum_{|k-td/T_0|\leq \psidelta\, d } e^{-i \phi_k} a_k.
 \end{align}
 Due to normalization and the tail estimate of \eqref{eq:tail-for-Delta}, 
 we have 
 \begin{align}
 1-\varepsilon_{LR}\leq \sum_{|k-td/T_0|\leq \psidelta\, d } a_k\leq 1.
 \end{align}
 Thus we have an expression where the $a_k$'s are almost normalized, and the
 $\phi_k$'s almost equal to each other, hence the expression must be close to one. Indeed Lemma \ref{lem:ak-phik} implies 
 that 
 \begin{align}
 \label{eq:delta-c-bound}
     |\Delta_C(\Omega)|\geq 1 - \varepsilon_{LR} - \pi\tilde\epsilon_V
 \end{align}
irrespectively of whether we are before or after the interaction, (i.e. 
whether $t\leq t_1$ or $t\geq t_2$),
because we have only used the fact that the phase is approximately equal, which happens in both cases as in Eq. \eqref{eq:phik-equal}. 

We can now come back to the fidelity. 
From estimates \eqref{eq:fidelity-clock-new2},
\eqref{eq:Delta-Delta-c} and \eqref{eq:delta-c-bound} we have 
\begin{align}
F\left( \rho_\cl^0(t), \rho_\cl(t) \right) &\geq 
	\min_{\Omega\in[-\pi,\pi]} |\Delta(\Omega)|^2 -     6 \varepsilon_\nu \\
	&\geq 
	\min_{\Omega\in[-\pi,\pi]} 
	|\Delta_C(\Omega)|^2 -2 \varepsilon_{LR}- 6 
	\varepsilon_\nu \\
	& \geq 
	1 - \varepsilon_{LR} - \pi\tilde\epsilon_V
	-2 \varepsilon_{LR}- 3 
	\varepsilon_\nu- 3 
	\varepsilon_\nu. 
\end{align}
We now use the exponential bounds on all those $\epsilon$'s. 
Recall that $\varepsilon_{LR}$  and 
$\tilde\epsilon_V$ were tails of
the pointer shape and the potential, and that
$\varepsilon_\nu$, $\varepsilon_c$ 
describe deviations of the pointer evolution from the simple picture of movement and phase acquisition. 
We write here bounds for those quantities 
given by  Eqs. \eqref{eq:epsilon-nu-epsilon-c},
\eqref{eq:tail-for-Delta}
and \eqref{eq:potential-tail-1}: 
\begin{align}
\label{eq:epsilons-bounds}
    \varepsilon_\nu&\lesssim
    t\, poly(d) e^{- \frac{1}{32\pi} d^{1/4}},
    \nonumber \\
	 \varepsilon_c&\lesssim
	 poly(d_\cl)\,  \me^{-\frac{\pi}{4} d_\cl},
	 \nonumber \\
	 \varepsilon_{LR}& \leq poly(d) e^{-\psidelta^2 d},
	 \nonumber \\
	 \tilde \epsilon_V&\leq \frac{1}{\potdelta}e^{-\potdelta^2 n },
\end{align}
where we have chosen $\psidelta=\potdelta= (t_2-t_1)/(4T_0)$, and to get the estimate for $\varepsilon_\nu$, we have chosen the 
potential concentration parameter $n$ to be $n= \lceil d^{1/4}\rceil$ 
hence we also have
\begin{align}
	\tilde \epsilon_V&\leq \frac{1}{\potdelta}e^{- \potdelta^2 d^{1/4}}.
\end{align}
We thus have 
\begin{align}
    F\left( \rho_\cl^0(t), \rho_\cl(t) \right) \lesssim 1- \frac{1}{t_2-t_1} poly(d) e^{-c_1 d^{1/4}}
\end{align}
where $c_1=\min\{\frac{1}{32 \pi}, (t_2-t_1)^2/(16 T_0^2) \}$.
We now use 
\begin{align}
    \|\rho-\sigma\|_1\leq 2 \sqrt{1-F(\rho,\sigma)^2}
\end{align}
so that  $F\geq 1-\epsilon$ implies $\|\rho-\sigma\|_1\leq 2\sqrt{2} \sqrt{\epsilon} $.
Using this we obtain  that for times $t$ satisfying  $t\leq t_1$ or $ t\geq t_2$
\begin{align}
    \|\rho_\cl^0(t), \rho_\cl(t)\| \lesssim
    \frac{1}{t_2-t_1} poly(d) e^{-c_2 d^{1/4}}
\end{align}
where $c_2=\min\{\frac{1}{64 \pi}, (t_2-t_1)^2/(32T_0^2) \}$. (We can put $t_2-t_1$ 
instead of $\sqrt{t_2-t_1}$  in front of \emph{poly}, as the differences is bounded from above by $T_0$, so for scaling, the small values of the different are relevant,
so this rough estimate is legitimate).
\end{proof}

\subsubsection{Proof of  Theorem \ref{Thm:Implementation with Quasi-Idela clock}}\label{Sec: proof of thm 2: Quasi-Ideal clock bound}

We will start with a definition and proposition whose usefulness will soon become apparent in the proof of Theorem \ref{Thm:Implementation with Quasi-Idela clock} below.

\begin{definition}[Autonomous Control device error
	]\label{def:ACD errors}
	Let $\rho^\text{target}_\A(t)$ denote the idealised/targeted control of system $\A$, namely 
	\be 
	\rho_\A^\text{target}(t)	 =
	\begin{cases}
		\rho_\A^0 & \quad \text{for } t\in[0,t_1]\\
		U{_\A^\text{target}}	\rho_\A^0 U{_\A^\text{target}}^\dag & \quad \text{for } t\in[t_2,T_0],
	\end{cases}
	\ee 
	where we associate the time interval $[t_1,t_2]$ with the time in which the CPTP map is being implemented in the ideal case. Furthmore, let $\rho^F_{\A\cl}(t)$ denote the autonomous evolution of $\A$ and the control system (the clock $\cl$),
	\be 
	\rho^F_{\A\cl}(t)=	\me^{-\mi t \hat H_{\A\cl}}\left(\rho_\A\otimes\rho_{\cl}\right)\me^{\mi t \hat H_{\A\cl}}.
	\ee 
	Let $\varepsilon_\A(t,d_\cl,d_\A)$, and $\varepsilon_\cl(t,d_\cl)$ be defined by the relations
	\ba \label{eq:varepsilon A def}
	\| \rho_{\A}^F(t)-\rho^\text{target}_\A(t)\|_1 &\leq\; \varepsilon_\A(t,d_\cl,d_\A),\\
	\| \rho^F_{\cl}(t)-\rho_\cl^0(t)\|_1 &\leq\;\label{eq:varepsilon cl def} \varepsilon_\cl(t,d_\cl),
	\ea 
	where $\rho_\cl^0(t)$ is the free evolution of the clock, 
	\be \label{eq:free clock evolution heneral eq}
	\rho_\cl^0(t):= \me^{-\mi t\hat H_\cl}\rho_\cl \,\me^{\mi t\hat H_\cl}.
	\ee 
\end{definition}


\begin{prop}\label{thm:clock int pic}
	There exists a clock state $\rho_\cl$ and time independent Hamiltonian, called the Quasi-Ideal Clock \cite{WSO} such that for all $t\in[0,t_1]\cup[t_2,T_0]$ and for all fixed $0<t_1<t_2<T_0$, the error terms $\varepsilon_\A$, $\varepsilon_\cl$ are given by
	\begin{align}
	\varepsilon_\A(t,d_\cl,d_\A)&=\sqrt{d_\A \tr[\rho_\A^2(0)]}\, \varepsilon(d_\cl),  \label{eq:epsilon A explicit main text}\\
	\varepsilon_\cl(t,d_\cl)& = \varepsilon(d_\cl),  \label{eq:epsilon cl explicit main text}
	\end{align}
	where $\varepsilon(d_\cl)$ is independent of the system $\A$ parameters, and is of order 
	\begin{align} \label{eq:ep dc for quasi ideal}
	\varepsilon(d_\cl)=\bo\bigg( poly(d_\cl)\, \exp\left[-c d_\cl^{1/4}\right] \bigg),\quad \text{as } d_\cl  \rightarrow \infty 
	\end{align}
	where the constant $c>0$ depends on $t_1,t_2$ and $poly(d_\cl)$ is a polynomial in $d_\cl$.
\end{prop}
\begin{proof}
	This proposition is a direct consequence of  Proposition \ref{prop:clock-error} and 
	and the results in \cite{WSO}. 
	Proposition  \ref{prop:clock-error} proves estimate \eqref{eq:epsilon cl explicit main text} with constant in exponent given by  
	$c_2$, 
	i.e., the constant used in estimate 
	\eqref{eq:prop-clock-error} in Prop. \ref{prop:clock-error}.
	In  \cite{WSO}
	(see Section \emph{ Examples: 2) System error faster than power-law decay}, page 47) 
	the following estimate is proven 
	\begin{align}\label{eqfract p - p example}
	\frac{\|\rho_\A(t)-\sigma_\A(t)\|_1}{\sqrt{d_\A \tr[\rho_\A^2(0)]}} &= \bo  \left( t\, poly(d_\cl)\,  \me^{-2c_0 d^{1/4}_\cl\sqrt{\ln d_\cl}}  \right),
	\end{align}
	for all $t\in[0,t_1]\cup[t_2,T_0]$ and for all fixed $0<t_1<t_2<T_0$. Here we have defined $2c_0:= \frac{\pi}{4}\alpha_0^2\chi_2^2$, where $\alpha_0$, $\chi$ are constants defined in \cite{WSO}. So this proves
	\eqref{eq:epsilon A explicit main  text} with constant $c_0$. Taking $c=\min\{2c_0,c_2\}$  finalises the proof.
\end{proof}
 This proposition is a generalisation of the results from \cite{WSO}. Specifically, these results were proven for the special case in which $t=T_0$ in Eq. \eqref{eq:epsilon cl explicit main text}. 

We are now ready to provide the proof of Theorem \ref{Thm:Implementation with Quasi-Idela clock} located in the main text.  We will precede the proof with a short overview.

{\it Overview of the proof of Theorem 
\ref{Thm:Implementation with Quasi-Idela clock}.}
The aim of the theorem is to show that 
in our autonomous setup, the final state of the system, catalyst and clock is close to product. 
Indeed, now that the catalyst and clock play the role of the total catalyst state, 
in order to prevent embezzling, we have to make sure that the total catalyst will not be polluted too much.  
Of course the final state on the system 
is to be close to the required state. 
Proposition \ref{prop:clock-error}
implies  that on the system and catalyst we get a state close to the required state,
and we have a small error on the clock.

To pass from this outcome to what we want, 
we note that the initial clock state is pure. 
Thus having a small error on the clock means also that the total state stays approximately product between the system-catalyst state and the clock. 
In the proof we express this in terms of fidelity. 

We now present the full proof. 

\begin{proof}\label{proof of Thm:Implementation with Quasi-Idela clock}
	We start by demonstrating part 1) of the theorem.  Define 
	\begin{align}
	U_{\Sy\cat\G}(t)= \begin{cases}
	\id_{\Sy\cat\G} &\mbox{ if } t\in[0,t_1]\\
	U_{\Sy\cat\G}' &\mbox{ if } t\in[t_2,T_0]
	\end{cases}
	\end{align}
	where $U_{\Sy\cat\G}'$ satisfies
	\begin{align}
		\tr_\G\left[U_{\Sy\cat\G}'\left(\rho_\Sy^0\otimes\rho_\cat^0\otimes\tauT_\G\right)U_{\Sy\cat\G}^{\prime\, \dag}  \right]=\rho_\Sy^1\otimes\rho_\cat^0.
	\end{align} 
	Define
	\begin{align}\label{eq:sigma def proof of 2nd thm}
	\sigma_{\Sy\cat\G}(t):= U_{\Sy\cat\G}'(t)\left(\rho_\Sy^0\otimes\rho_\cat^0\otimes\tauT_\G\right)U_{\Sy\cat\G}^{\prime\, \dag}(t).
	\end{align}
	It follows by the definition of t-CNO (Def. \ref{def:t-CTO} and Prop. \ref{prop:equiv of t-CNO and CNO}) that for every pair $\rho_\Sy^0$, $\rho_\Sy^1$ for which there exists a t-CNO from $\rho_\Sy^0$ to $\rho_\Sy^1$, there exists a unitary $U_{\Sy\cat\G}$ satisfying the above criteria. Since the catalyst $\rho_\cat^0$ is arbitrary, this is true iff Eq. \eqref{eq:sigma is valid NO} holds. Therefore $\sigma_{\Sy}(t)$ in Eq. \eqref{eq:sigma def proof of 2nd thm} fulfils part 1) of the theorem.
	
	We now proceed with proving part 2) of the theorem. 
	Recalling Def. \ref{def:ACD errors} and Prop. \ref{thm:clock int pic}, and using the identifications $A=\Sy\cat\G$, $U_\A^\text{target}= U_{\Sy\cat\G}'$, for every unitary $U_{\Sy\cat\G}$ above, there exists an interaction term $\hat I_{\Sy\cat\cl\G}$ such that using the Quasi-Ideal Clock we have	
	\begin{align}
	\label{eq:varepsilon A proof}
	\| \rho_{\Sy\cat\G}^F(t)-\sigma_{\Sy\cat\G}&(t)\|_1 \leq\; \varepsilon_\A \\\!\!\!\!\!\!\!\! =\sqrt{d_\Sy d_\cat d_\G }&\sqrt{\tr[(\rho^0_\Sy\otimes\rho^0_\cat\otimes\tauT_\G)^2]}\,\, \varepsilon_\cl(d_\cl)\nonumber\\
\qquad\qquad	\leq \sqrt{d_\Sy d_\cat}\,\, \varepsilon_\cl&(d_\cl),\\ \label{eq:varepsilon cl proof}
	\| \rho^F_{\cl}(t)-\rho^0_\cl(t)\|_1 \leq\;& \varepsilon_\cl(d_\cl),
	\end{align}
	where in the last line of \eqref{eq:varepsilon A proof} we have taken into account that $\tauT_\G=\id/d_\G$. Recall that an expression for $\varepsilon(d_\cl)$ is given by Eq. \eqref{eq:ep dc for quasi ideal}.
	We now apply Prop. \ref{prop:red} with the identifications
	\begin{align}
	\rho^F_{\Sy\cat\G}(t)&=:\rho_\A,\quad \rho_\cl^F(t)=:\rho_\B, \quad  \rho^F_{\Sy\cat\G\cl}(t)=:\rho_{\A\B}\\
	\sigma_{\Sy\cat\G}(t)&=:\sigma_\A,\quad \rho_\cl^0(t)=:\sigma_\B,\quad\\
	  \sigma_{\Sy\cat\G}(t)&\otimes\rho_\cl^0(t)=:\sigma_{\A\B},
	\end{align}
	hence 
	\begin{equation}
	\ep_1 =\varepsilon_\A,\quad \ep_2=\vep_\cl, \quad \ep_3=0,
	\end{equation}
	($\vep_3$ vanishes because  
	$\rho_\cl^0(t)$ is a pure state)
	to achieve 
	\begin{align} 
	\| \rho^F_{\Sy\cat\G\cl}(t)-\sigma_{\Sy\cat\G}(t)\otimes \rho_\cl^0(t) \|_1 \leq  
	2\sqrt{\varepsilon_\cl}+ \varepsilon_\A, 
	\end{align}
	for all $t\in[0,t_1]\cup[t_2,T_0]$, and where $\varepsilon_\A, \varepsilon_\cl$ are given in Eqs. \eqref{eq:varepsilon A proof}, \eqref{eq:varepsilon cl proof}.
	Applying the data processing inequality, we find
	\begin{equation} 
	\| \rho^F_{\Sy\cat\cl}(t)-\sigma_{\Sy\cat}(t)\otimes \rho_\cl^0(t) \|_1 \leq
	2\sqrt{\varepsilon_\cl}+ \varepsilon_\A, 
	\end{equation}
	for all $t\in[0,t_1]\cup[t_2,T_0]$. Using the triangle inequality, we have
	\begin{align} 
	\| \rho^F_{\Sy\cat\cl}&(t)-\rho^F_{\Sy}(t)\otimes\rho_\cat^0\otimes \rho_\cl^0(t) \|_1 \\ \leq& \| \rho^F_{\Sy\cat\cl}(t)-\sigma_{\Sy\cat}(t)\otimes \rho_\cl^0(t) \|_1\\
	& +\| \sigma_{\Sy\cat}(t)\otimes \rho_\cl^0(t)- \rho^F_{\Sy}(t)\otimes\rho_\cat^0\otimes \rho_\cl^0(t) \|_1\\
	\leq  &
	2\sqrt{\varepsilon_\cl(t)}+ \varepsilon_\A(t) +\| \sigma_{\Sy\cat}(t)- \rho^F_{\Sy}(t)\otimes\rho_\cat^0 \|_1.\label{q:F - F intermediate proof of thm 2}
	\end{align}
	Now we note that by definition, it follows that $\sigma_{\Sy\cat}(t)=\sigma_{\Sy}(t)\otimes\rho_\cat^0$ for all $t\in[0,t_1]\cup[t_2,T_0]$. Plugging into Eq. \eqref{q:F - F intermediate proof of thm 2} we achieve 
	\begin{align}
	\| \rho^F_{\Sy\cat\cl}(t)-&\rho^F_{\Sy}(t)\otimes\rho_\cat^0\otimes \rho_\cl^0(t) \|_1 \\ 
	&\leq 
	2\sqrt{\varepsilon_\cl}+ \varepsilon_\A +\|\sigma_{\Sy}(t)- \rho^F_{\Sy}(t) \|_1\\
	&\leq 
	2\sqrt{\varepsilon_\cl}+ 2 \varepsilon_\A ,
	\end{align}
	for all $t\in[0,t_1]\cup[t_2,T_0]$ and where in the last line, we have used Eq. \eqref{eq:varepsilon A proof} after applying the data processing inequality to it. 
	W.l.o.g. assume that $\varepsilon_\cl\leq  2$ (If this does not hold, then the following bound holds anyway since the trace distance between any two states is upper bounded by 2), so that $\vep_\cl \leq \sqrt{2 \vep_\cl}$ and 
	using \eqref{eq:varepsilon A proof} we achieve
	\begin{align}
	\| \rho^F_{\Sy\cat\cl}&(t)-\rho^F_{\Sy}(t)\otimes\rho_\cat^0\otimes \rho_\cl^0(t) \|_1\\
	 &\leq  2\sqrt{\varepsilon_\cl} +2 \sqrt{d_\Sy d_\cat} \vep_\cl \\
	&\leq \left(2 + 2 \sqrt{2} \sqrt{d_\Sy d_\cat}\right) \sqrt{\vep_\cl} \\
	&\leq  \left(2+3\sqrt{d_\Sy d_\cat}\right) \sqrt{\varepsilon_\cl}\\ 
	&=\epemb,
	\end{align}
	for all $t\in[0,t_1]\cup[t_2,T_0]$. Now, recalling that $\varepsilon_\cl$ is independent of $d_\Sy,d_\cat,d_\G$, and only a function of $d_\cl, t_1,t_2,T_0$, 
	we obtain estimate \eqref{eq:thm2-ep-emb} of part 2) of the theorem. Next, the formula  \eqref{eq:ep dc for quasi ideal} from Proposition \ref{thm:clock int pic} gives the estimate 
	\ref{eq:scaling of quasi ideal clock} concluding the proof of 
	theorem \ref{Thm:Implementation with Quasi-Idela clock}.
\end{proof}

\cblack

\subsection{Proof of Theorem \ref{thm:noemb physical t CTO}}\label{Sec:proof of thm 3, t-CTOs generic bound}

In this section we prove Theorem \ref{thm:noemb physical t CTO} in the main text.

\begin{proof}[Proof of Theorem \ref{thm:noemb physical t CTO}]\label{proof to thm:noemb physical t CTO}
	
	The proof will be divided into two parts labelled A and B. 
	
	Part A consists in proving that the theorem statements 1) and 2) hold under a different set of conditions than those of the theorem. Namely, that 1) and 2) hold when the following two conditions both simultaneously hold:\\
	a) the final joint system-catalyst-bath state is very close to that of the target joint system-catalyst-bath state.\\
	b) the final clock state is very close to its free state. The proof of part A uses basic relationships between quantum states (such as trace distance and fidelity), but does not take into account any dynamical properties.
	
	Part B consists in proving that the conditions in the Theorem from which 1) and 2) should follow, do indeed imply conditions a) and b) from part A. The proof uses the dynamical properties of the states.
	
	\subsubsection{Part A of proof of Theorem \ref{thm:noemb physical t CTO}}
	We start with a comment on notation and a few immediate consequences. We will denote $U^\textup{target}_{\Sy\cat\G} (t)= \me^{- \mi \hatf(t) \hat H_{\Sy\cat\G}^\textup{int} }$ from the main text by $U^{\textup{target } (\epsig)}_{\Sy\cat\G} (t)= \me^{- \mi \hatf(t) \hat H_{\Sy\cat\G}^{\textup{int}} }$ here to remind ourselves that $\hat H_{\Sy\cat\G}^\textup{int}$ induces a small error $\epsig$ onto the final catalyst and system state (see Eqs. \eqref{eq:2nd condition of robust embezzelment Ham}, \eqref{eq:epsig def}). 
	We will also denote $U^{\textup{target } (0)}_{\Sy\cat\G} (t):= \me^{- \mi \hatf(t) \hat I_{\Sy\cat\G}^{\textup{int}} }$, since $\hat I_{\Sy\cat\G}^{\textup{int}}$ corresponds to the case of no error, i.e. $\epsig=0$, (see Eq. \eqref{eq:2nd condition of robust embezzelment Ham 2}). Accordingly we shall denote 
	\begin{align}
		&\rho_{\Sy\cat\G}^{\textup{target}(0)}(t):=\\
		& U_{\Sy\cat\G}^{\textup{target}(0)}(t) \left(\rho^0_{\Sy}(t)\otimes\rho^0_\cat(t)\otimes\tauGibb\right) U_{\Sy\cat\G}^{\textup{target (0)\,\dag}}(t).\label{eq:target state def 1}\\
	&	\rho_{\Sy\cat\G}^{\textup{target}(\ep_H)}(t):=\\
	& U_{\Sy\cat\G}^{\textup{target}\,(\ep_H)}(t) \left(\rho^0_{\Sy}(t)\otimes\rho^0_\cat(t)\otimes\tauGibb\right) U_{\Sy\cat\G}^{\textup{target}\, (\ep_H) \,\dag}(t).\label{eq:target state epsilon def}
	\end{align}
	Recall that
	\begin{align}
		\hatf(t)=\begin{cases}
			0 &\mbox{ if } t\in[0,t_1]\\
			1  &\mbox{ if } t\in[t_2,t_3],
		\end{cases}
	\end{align}
	Therefore, similarly to as in Eqs. \eqref{eq:ep-H-rho-target1} and \eqref{eq:ep-H-rho-target2}
	we have for $t \in[0, t_1]$
	\begin{align}
		\label{eq:ep-H-rho-target1-zero}
		\rho_{\Sy\cat}^{\textup{target}(0)}(t)= \rho^0_{\Sy}(t)\otimes\rho^0_\cat(t),
	\end{align} 
	while for $t \in[ t_2,t_3]$
	\begin{align}
		\label{eq:ep-H-rho-target2-zero}
		\rho_{\Sy\cat}^{\textup{target}(0)}(t)&=\me^{-\mi t (\hat H_\Sy+\hat H_\cat)} \rho_{\Sy}^1\otimes\rho_{\cat}^0\me^{\mi t (\hat H_\Sy+\hat H_\cat)}\\
		&= \rho_{\Sy}^1(t)\otimes\rho_{\cat}^0(t).
	\end{align}
	Hence, Eqs. \eqref{eq:ep-H-rho-target1-zero}, \eqref{eq:ep-H-rho-target2-zero} together imply
	\begin{align}
		\label{eq:ep-H-rho-target3-zero}
		\rho_{\Sy\cat}^{\textup{target}(0)}(t)= \rho^{\textup{target}(0)}_{\Sy}(t)\otimes\rho^0_\cat(t),
	\end{align}
	for $t \in[0, t_1]\cup[t_2, t_3]$.\Mspace
	
	Part A will consist in proving that the following holds.
	Let $\epSycatG(\epsig;t)>0$ 
	and $\epcl(t)>0$ satisfy
	\begin{align}
		\label{eq:1st condition thm:noemb physical t CTO proof}
		\| \rho^F_{\Sy\cat\G}(t) - \rho^{\textup{target } (\epsig)}_{\Sy\cat\G}(t)  \|_1 &\leq \epSycatG(\epsig;t),
		\\
		\label{eq:2nd condition thm:noemb physical t CTO proof}
		\| \rho^F_{\cl}(t) - \rho_{\cl}^0(t)  \|_1 &\leq \epcl(t).
	\end{align}
	It follows that:
	\begin{itemize}
		\item [1)] The deviation from the idealised dynamics is bounded by 
		\begin{align}
			\| \rho^F_{\Sy\cat\cl}&(t) - \rho^F_{\Sy}(t)\otimes \rho_\cat^0(t)\otimes \rho_\cl^0(t)  \|_1 \\
			&
			\leq   2 \epSycatG(\epsig;t)
			 +2 \sqrt{\epcl(t)}+2\epsig\,\hatf(t).\label{inter proof thm 3 1}
		\end{align}
		
		\item [2)] The final state $\rho^F_\Sy(t)$ is 
		\begin{align}\label{part to for proof closeness}
			\| \rho^F_{\Sy}(t) - \rho^{\textup{target}\, (0)}_{\Sy}(t)  \|_1 \leq \epSycatG(\epsig;t) + \epsig\,\hatf(t)
		\end{align}
		close to one which can be reached via t-CTO:  For all $t\in[0,t_1]\cup[t_2,t_3]$ the transition
		\be \rho_\Sy^0\otimes\rho_\cat^0\otimes\rho_\cl^0 \;\;\;\text{ to }\;\;\; \rho_\Sy^{\textup{target}\, (0)}(t)\otimes\rho_\cat^0(t)\otimes\rho_\cl^0(t)\label{eq:sigma is valid TO proof}
		\ee 
		is possible via a TO i.e. $\rho_\Sy^0$ to $\rho_\Sy^{\textup{target}\, (0)}$ via a t-CTO. 
	\end{itemize}
	
	We begin with proving item 2).
	To prove the Eqs. \eqref{part to for proof closeness},\eqref{eq:sigma is valid TO proof}, we start by extending the definitions of $\rho_{\Sy\cat\G}^{\textup{target}\, (\epsig)}(t)$  and $\rho_{\Sy\cat\G}^{\textup{target}\, (0)}(t)$ in \eqref{eq:ep-H-rho-target1-zero},\eqref{eq:ep-H-rho-target2-zero} to include the clock system:
	\begin{align}
		\rho_{\Sy\cat\G\cl}^{\textup{target}\, (\epsig)}(t) &=\rho_{\Sy\cat\G}^{\textup{target}\, (\epsig)}(t)\otimes \rho_\cl^0(t),\\
		\rho_{\Sy\cat\G\cl}^{\textup{target}\, (0)}(t) &=\rho_{\Sy\cat\G}^{\textup{target}\, (0)}(t)\otimes \rho_\cl^0(t),
	\end{align}
	where $\rho_\cl^0(t)$ is the free evolution of the clock defined in Eq. \ref{eq:def free clock ev}. Therefore, from Eq. \eqref{eq:ep-H-rho-target3-zero}, it follows that the reduced state after tracing out the Gibbs state on $\G$ is
	\begin{align}\label{eq:target Sy cat cl}
		\rho_{\Sy\cat\cl}^{{\textup{target}\, (0)}}(t)= \rho_{\Sy}^{\textup{target}\, (0)}(t)\otimes\rho_\cat^0(t)\otimes \rho_\cl^0(t),
	\end{align}
	for $t\in[0,t_1]\cup[t_2,t_3]$. Thus taking into account property Eq. \eqref{eq:commutation of int term in generic Hm} it follows by definition of CTOs and t-CTOs that a transition from $\rho_\Sy^0$ to $\rho_{\Sy}^{\textup{target}\, (0)}(t)$ is possible via a t-CTO. Finally, applying the data processing inequality to Eq. \eqref{eq:1st condition thm:noemb physical t CTO proof}, we achieve 
	\begin{align}
		\| \rho^F_{\Sy}(t) - \rho^{\textup{target}\, (\epsig)}_{\Sy}(t)  \|_1 \leq \epSycatG(\epsig;t),
	\end{align} 
	while applying the date processing inequality to Eqs. \eqref{eq:target state def 1},\eqref{eq:target state epsilon def}  we find $\| \rho^{\textup{target}\, (0)}_{\Sy}(t) - \rho^{\textup{target}\, (\epsig)}_{\Sy}(t)  \|_1 = 0$ for $t\in [0,t_1]$ while from Eqs. \eqref{eq:2nd condition of robust embezzelment Ham}, \eqref{eq:2nd condition of robust embezzelment Ham 2}, \eqref{eq:epsig def}, we see $\| \rho^{\textup{target}\, (0)}_{\Sy}(t) - \rho^{\textup{target}\, (\epsig)}_{\Sy}(t)  \|_1 \leq \epsig$ for $t\in [t_2,t_3]$. Hence, combining both equations, we have $\| \rho^{\textup{target}\, (0)}_{\Sy}(t) - \rho^{\textup{target}\, (\epsig)}_{\Sy}(t)  \|_1 \leq \epsig\,\hatf(t)$ for $t\in [0,t_1]\cup [t_2,t_3]$. Then Eq. \eqref{part to for proof closeness} in 2) above follows from the triangle inequality:
	\begin{align}
	&	\| \rho^F_{\Sy}(t) - \rho^{\textup{target}\, (0)}_{\Sy}(t)  \|_1 \leq \| \rho^F_{\Sy}(t) - \rho^{\textup{target}\, (\epsig)}_{\Sy}(t)  \|_1 \\
	&+ \| \rho^{\textup{target}\, (\epsig)}_{\Sy}(t) - \rho^{\textup{target}\, (0)}_{\Sy}(t)  \|_1    \leq \epSycatG(\epsig;t)\\
	& + \epsig\,\hatf(t).
	\end{align}
	We now prove the above item 1) (i.e. the estimate \eqref{inter proof thm 3 1}).  We begin by using the triangle inequality followed  by the identity $\rho_{\Sy\cat}^\textup{target (0)}(t)= \rho_{\Sy}^\textup{target (0)}(t)\otimes\rho_\cat^0(t)$ which follows from Eq. \eqref{eq:target Sy cat cl}.
	\begin{align}
		\| &\rho^F_{\Sy\cat\cl}(t) - \rho^F_{\Sy}(t)\otimes \rho_\cat^0(t)\otimes \rho_\cl^0(t)  \|_1  \\
		=&\| \rho^F_{\Sy\cat\cl}(t) - \rho^{\textup{target}\,(\epsig)}_{\Sy\cat}(t)\otimes \rho_\cl^0(t)\\
		&+  \rho^{\textup{target}\,(\epsig)}_{\Sy\cat}(t)\otimes \rho_\cl^0(t)- \rho^F_{\Sy}(t)\otimes \rho_\cat^0(t)\otimes \rho_\cl^0(t)  \|_1 \\
		\leq &	\| \rho^F_{\Sy\cat\cl}(t) - \rho^{\textup{target}\,(\epsig)}_{\Sy\cat}(t)\otimes \rho_\cl^0(t)\|_1\\
		& + \|  \rho^{\textup{target}\,(\epsig)}_{\Sy\cat}(t)\otimes \rho_\cl^0(t)-  \rho^{\textup{target}\,(0)}_{\Sy\cat}(t)\otimes \rho_\cl^0(t)\|_1 \\
		&+ \|  \rho^{\textup{target}\,(0)}_{\Sy\cat}(t)\otimes \rho_\cl^0(t)- \rho^F_{\Sy}(t)\otimes \rho_\cat^0(t)\otimes \rho_\cl^0(t)  \|_1\\
		=& 	\| \rho^F_{\Sy\cat\cl}(t) - \rho^{\textup{target}\,(\epsig)}_{\Sy\cat}(t)\otimes \rho_\cl^0(t)\|_1 \\
		&+ \|  \rho^{\textup{target}\,(\epsig)}_{\Sy\cat}(t)-  \rho^{\textup{target}\,(0)}_{\Sy\cat}(t)\|_1 \\
		&+ \|  \rho^{\textup{target}\,(0)}_{\Sy}\otimes\rho^0_{\cat}(t)- \rho^F_{\Sy}(t)\otimes \rho_\cat^0(t)  \|_1\nonumber\\
		\leq & 	\| \rho^F_{\Sy\cat\cl}(t) - \rho^{\textup{target}\,(\epsig)}_{\Sy\cat}(t)\otimes \rho_\cl^0(t)\|_1\\
& + \epsig\,\hatf(t)+\|  \rho^{\textup{target}\,(0)}_{\Sy}- \rho^F_{\Sy}(t)  \|_1\\
		\leq& 	\| \rho^F_{\Sy\cat\cl}(t) - \rho^{\textup{target}\,(\epsig)}_{\Sy\cat}(t)\otimes \rho_\cl^0(t)\|_1\\
		& +2 \epsig\,\hatf(t)+ \epSycatG(\epsig;t).\label{eq:inter thm 3 proof bound}
	\end{align}
	where we have applied the data processing inequality to Eq. \eqref{eq:1st condition thm:noemb physical t CTO proof} and used the resultant equation in the last line. Now we make the following identifications, noting that $\rho_\cl^0(t)$ all $t\in\rr$ is pure by assumption of the theorem.
	\begin{align}
		\rho^F_{\Sy\cat\G}(t)=:\rho_\A,\quad \rho_\cl^F(t)&=:\rho_\B, \quad  \rho^F_{\Sy\cat\G\cl}(t)=:\rho_{A\B}\\
		\rho^{\text{target}\,(\epsig)}_{\Sy\cat\G}(t)=:\sigma_\A,\quad \rho_\cl^0(t)&=:\sigma_\B,\quad \\
		\rho_{\Sy\cat\G}^{\text{target}\,(\epsig)}(t)\otimes\rho_\cl^0(t)=:\sigma_{\A\B}\!\!&
	\end{align}
	and apply Prop. \ref{prop:red}  with use of 
	\eqref{eq:1st condition thm:noemb physical t CTO proof},\eqref{eq:2nd condition thm:noemb physical t CTO proof}   arriving at
	\begin{align}
		\ep_1=\epSycatG(\epsig;t),\quad \ep_2
		= \epcl(t), \quad \ep_3=0,
	\end{align}
	and thus
	\begin{align} 
		\| &\rho^F_{\Sy\cat\G\cl}(t)-\rho^{\text{target}\,(\epsig)}_{\Sy\cat\G}(t)\otimes \rho_\cl^0(t) \|_1 \\
		&\leq 2\sqrt{\epcl(t)} + \epSycatG(\epsig;t).
	\end{align}
	Applying the data processing inequality to the above equation, followed by substituting into Eq. \eqref{eq:inter thm 3 proof bound}, gives
	\begin{align}
		\|& \rho^F_{\Sy\cat\cl}(t) - \rho^F_{\Sy}(t)\otimes \rho_\cat^0(t)\otimes \rho_\cl^0(t)  \|_1\\
		 & \leq  2\sqrt{\epcl(t)} + \epSycatG(\epsig;t)  +\epSycatG(\epsig;t) +2 \epsig\,\hatf(t) \nonumber\\ 
		&  =  2 \epSycatG(\epsig;t)+2 \sqrt{\epcl(t)}+2\epsig\,\hatf(t).
	\end{align}
	
	\subsubsection{Part B of proof of Theorem \ref{thm:noemb physical t CTO}}
	
	We now set out to prove the second part, which consists in 
	deriving expressions for $\epSycatG(\epsig;t)$ and $\epcl(t)$ such that the contents of sections 1) and 2) above are consistent with the claims in 1) and 2) of the theorem. 
	
	To start with, since $\hat H_\Sy+\hat H_\cat+\hat H_\G$ and $\hat H_{\Sy\cat\G}^\textup{int}$ commute, they share a common eigenbasis which we denote $\{\ket{E_j}\}_j$. We can write the interaction term in terms of this basis as follows, $\hat H^\textup{int}_{\Sy\cat\G}=\sum_{j=1}^{d_\Sy d_\cat d_\G} \Omega_j\ketbra{E_j}{E_j}$ with eigenvalues $\Omega_j$ in the range $\Omega_j\in[-\pi,\pi]$ since $\|\hat H^\textup{int}_{\Sy\cat\G}\|_\infty\leq \pi$. We can also expanding the state $\rho_\Sy^0\otimes\rho_\cat^0\otimes\tau_\G$ in the energy eigenbasis $\{\ket{E_j}\}_j$. Doing so allows one to simplify the expression for $\rho_{\Sy\cat\G\cl}^F(t)$. We find
	\begin{widetext}
	\begin{align}
		\rho_{\Sy\cat\G\cl}^F(t)&=\me^{-\mi t \left(\hat H_\Sy+\hat H_\cat+\hat H_\G +\hat H_{\Sy\cat\G}^\textup{int}\otimes \hat H_{\cl}^\textup{int} +\hat H_\cl\right)}\rho_\Sy^0\otimes\rho_\cat^0\otimes\tau_\G\otimes\ketbra{\rho_\cl^0}{\rho_\cl^0} \, \me^{\mi t \left(\hat H_\Sy+\hat H_\cat+\hat H_\G +\hat H_{\Sy\cat\G}^\textup{int}\otimes \hat H_{\cl}^\textup{int} +\hat H_\cl\right)}\\
		&=\sum_{j,j'=1}^{d_\Sy d_\cat d_\G} \me^{-\mi t \left(\hat H_\Sy+\hat H_\cat+\hat H_\G +\Omega_j \hat H_{\cl}^\textup{int} +\hat H_\cl\right)} \rho^0_{\Sy\cat\G, j,j'} \ketbra{E_j}{E_{j'}} \otimes \ketbra{\rho_\cl^0}{\rho_\cl^0} \me^{\mi t \left(\hat H_\Sy+\hat H_\cat+\hat H_\G +\Omega_{j'} \hat H_{\cl}^\textup{int} +\hat H_\cl\right)}\\
		&= \sum_{j,j'=1}^{d_\Sy d_\cat d_\G} \rho^0_{\Sy\cat\G, j,j'}(t) \ketbra{E_j}{E_{j'}} \otimes \ketbra{\rho_{\cl,j}^0(t)}{\rho_{\cl,j'}^0(t)},\label{eq:rho final proof thm 3}
	\end{align}
\end{widetext}
	where 
	\begin{align}
		\sum_{j,j'=1}^{d_\Sy d_\cat d_\G} \rho^0_{\Sy\cat\G, j,j'}(t) \ketbra{E_j}{E_{j'}}&= \rho_\Sy(t)\otimes \rho_\cat^0(t)\otimes\tau_\G,\\
		\ket{\rho_{\cl,j}^0 (t)}&= \me^{-\mi t \left( \Omega_{j} \hat H_{\cl}^\textup{int} +\hat H_\cl\right)} \ket{\rho_{\cl}^0}.
	\end{align}
	We thus have by taking partial traces
	\begin{align}
		&\rho_{\Sy\cat\G}^F(t)\\
		&= \sum_{j,j'=1}^{d_\Sy d_\cat d_\G} \rho^0_{\Sy\cat\G, j,j'}(t) \ketbra{E_j}{E_{j'}} \braket{\rho_{\cl,j'}^0(t)|\rho_{\cl,j}^0(t)},\\
		&\rho_{\cl}^F(t) = \sum_{j=1}^{d_\Sy d_\cat d_\G} \rho^0_{\Sy\cat\G, j,j}  \ketbra{\rho_{\cl,j}^0(t)}{\rho_{\cl,j}^0(t)}. \label{eq:rho final proof thm 3 trace not clock}
	\end{align}
	Similarly
	\begin{align}
		&\rho_{\Sy\cat\G}^{\textup{target } (\epsig)}(t)\\
		&= U_{\Sy\cat\G}^{\textup{target } (\epsig)}(t) \left(\rho^0_{\Sy}(t)\otimes\rho^0_\cat(t)\otimes\tauGibb\right) U_{\Sy\cat\G}^{\textup{target } (\epsig)\,\dag}(t)\\
		&= \me^{-i t \hatf(t)\hat H_{\Sy\cat\G}^\textup{int}}\\
		&\qquad \left(\sum_{j,j'=1}^{d_\Sy d_\cat d_\G} \rho^0_{\Sy\cat\G, j,j'}(t) \ketbra{E_j}{E_{j'}} \right)\me^{i t \hatf(t) \hat H_{\Sy\cat\G}^\textup{int}} \\
		&= \sum_{j,j'=1}^{d_\Sy d_\cat d_\G} \rho^0_{\Sy\cat\G, j,j'}(t) \me^{-\mi t (\Omega_j-\Omega_{j'})\hatf(t)} \ketbra{E_j}{E_{j'}}.
	\end{align}
	Noting that the Frobenious norm $\|\cdot\|_F$ upper bounds the trace distance by the inequality $\|\cdot\|_F\geq \| \cdot \|_1 /\sqrt{d}$ for a $d$ dimensional space, we find
\begin{widetext}
	\begin{align}
		\| \rho^F_{\Sy\cat\G}(t) - \rho^{\textup{target} (\epsig)}_{\Sy\cat\G}(t)  \|_1 &\leq \sqrt{d_\Sy d_\cat d_\G} \| \rho^F_{\Sy\cat\G}(t) - \rho^{\textup{target} (\epsig)}_{\Sy\cat\G}(t)  \|_F \\
		&= \sqrt{d_\Sy d_\cat d_\G} \sqrt{\sum_{j,j'=1}^{d_\Sy d_\cat d_\G} \left| \rho^0_{\Sy\cat\G, j,j'}(t)\right|^2 \left|\me^{-\mi t (\Omega_j-\Omega_{j'})\hatf(t)} - \braket{\rho_{\cl,j'}^0(t)|\rho_{\cl,j}^0(t)} \right|^2 \,}\\
		&\leq \sqrt{d_\Sy d_\cat d_\G} \sqrt{\sum_{j,j'=1}^{d_\Sy d_\cat d_\G} \left| \rho^0_{\Sy\cat\G, j,j'}(t)\right|^2 \left( \max_{m,n}\left|\me^{-\mi t (\Omega_m-\Omega_{n})\hatf(t)} - \braket{\rho_{\cl,n}^0(t)|\rho_{\cl,m}^0(t)}\right|^2 \right) \,}\\
		&\leq \sqrt{d_\Sy d_\cat d_\G} \sqrt{ \tr\left[\rho^0_\Sy(t)^2\otimes\rho_\cat^0(t)^2\otimes\tau_\G^2\right] \left( \max_{m,n} \left|\me^{-\mi t (\Omega_m-\Omega_{n})\hatf(t)} - \braket{\rho_{\cl,n}^0(t)|\rho_{\cl,m}^0(t)}\right|^2\right) \,}\\
		&= \sqrt{d_\Sy\tr\left[{\rho^0_\Sy}^2\right] d_\cat\tr\left[{\rho^0_\cat}^{\!\!\!\!\!\!2}\,\,\right] d_\G \tr\left[{\tau_\G}^2\right]} \,\, \max_{m,n}  \left|\me^{-\mi t (\Omega_m-\Omega_{n})\hatf(t)} - \braket{\rho_{\cl,n}^0(t)|\rho_{\cl,m}^0(t)}\right| \\
		& \leq \sqrt{d_\Sy\tr\left[{\rho^0_\Sy}^2\right] d_\cat\tr\left[{\rho^0_\cat}^{\!\!\!\!\!\!2}\,\,\right] d_\G \tr\left[{\tau_\G}^2\right]} \,\, \max_{x,y\in[-\pi,\pi]} \left|1 - \braket{\rho_{\cl}^0|\hat \Gamma^\dag(x,t) \hat \Gamma(y,t) |\rho_{\cl}^0}\right| \\
		&= 
		A \,\, \max_{x,y\in[-\pi,\pi]} \left|1 - \Delta(t;x,y)\right|,
	\end{align}
	\end{widetext}
	where we have denoted 
	\begin{align}
		\label{eq:A-Delta-def}
		\Delta(t;x,y)&:= \braket{\rho_\cl^0|\Gamma^\dag(x,t) \Gamma(y,t)|\rho_{\cl}^0},\\
		A&:= \sqrt{d_\Sy\tr\left[{\rho^0_\Sy}^2\right] d_\cat\tr\left[{\rho^0_\cat}^{\!\!\!\!\!\!2}\,\,\right] d_\G \tr\left[{\tau_\G}^2\right]}.
	\end{align}

	(Note here, that since $d\, \tr [\rho^2]\geq 1$ for any $d$-dimensional state, we have  $A\geq 1.$) 
	Thus 
	$\epSycatG(\epsig;t)$, 
	from \eqref{eq:1st condition thm:noemb physical t CTO proof}, we  
	set as 
	\begin{align}
		\label{ep:eqcl}
		\epSycatG(\epsig;t)=A \,\, \max_{x,y\in[-\pi,\pi]} \left|1 - \Delta(t;x,y)\right|.
	\end{align}
	Noting that the fidelity $F$ between a pure state $\ket{\rho_\cl^0(t)}=\me^{-\mi t \hat H_\cl} \ket{\rho_\cl^0}$ and a state $\rho^F_\cl(t)$ is given by $F=\bra{\rho_\cl^0(t)} \rho^F_\cl(t) \ket{\rho_\cl^0(t)}$, using Eq. \eqref{eq:rho final proof thm 3 trace not clock} and the usual bound for the trace distance in terms of the fidelity, we find 
	\begin{align}
		&\| \rho^F_{\cl}(t)-\rho_\cl^0(t)\|_1\\
		&\leq 2 \sqrt{1-F\left(\rho^F_\cl(t),\ket{\rho_\cl^0(t)}\right)}\\
		& = 2 \sqrt{1- \bra{\rho_\cl^0(t)} \rho^F_\cl(t) \ket{\rho_\cl^0(t)} }= \\
		&  2 \sqrt{1-\sum_{j=1}^{d_\Sy d_\cat d_\G} \rho^0_{\Sy\cat\G, j,j} \braket{\rho_\cl^0(t)|\rho_{\cl,j}^0(t)}\braket{\rho_{\cl,j}^0(t)| \rho_\cl^0(t)} } \nonumber\\
		&\leq  2 \sqrt{1- \min_{j} \left| \braket{\rho_\cl^0(t) |\rho_{\cl,j}^0(t)} \right|^2} \\
		&\leq 2 \sqrt{1- \min_{x\in[-\pi,\pi]} \left| \braket{\rho_\cl^0(t) |\me^{-\mi x \hatf(t)} \Gamma(x,t)|\rho_{\cl}^0} \right|^2}\\
		&= 2 \sqrt{1- \min_{x\in[-\pi,\pi]} \left| \braket{\rho_\cl^0|\Gamma^\dag(0,t) \Gamma(x,t)|\rho_{\cl}^0} \right|^2} \\
		& \leq 2 \max_{x,y\in[-\pi,\pi]} \sqrt{1-  \left| \braket{\rho_\cl^0|\Gamma^\dag(y,t) \Gamma(x,t)|\rho_{\cl}^0} \right|^2} \\
		&	=		2 \max_{x,y\in[-\pi,\pi]} \sqrt{1-  \left|\Delta(t;x,y) \right|^2}, 
	\end{align}
	so that we can set 	$\epcl(t) $,
	from \eqref{eq:2nd condition thm:noemb physical t CTO proof}, to
	\begin{align}
		\label{eq:epSy}
		\epcl(t) =
		2 \max_{x,y\in[-\pi,\pi]} \sqrt{1-  \left|\Delta(t;x,y) \right|^2}.
	\end{align} 
	Inserting \eqref{eq:epSy} and \eqref{ep:eqcl} into 
	Eqs. \eqref{inter proof thm 3 1}, \eqref{part to for proof closeness}, we conclude 
	\begin{align}
		&\| \rho^F_{\Sy\cat\cl}(t) - \rho^F_{\Sy}(t)\otimes \rho_\cat^0(t)\otimes \rho_\cl^0(t)  \|_1 
		\\
		&\leq \,  2 \sqrt{A} \max_{x,y\in[-\pi,\pi]} \sqrt{|1 - \Delta(t;x,y)|}\\
		&  + 4 \max_{x,y\in[-\pi,\pi]} \sqrt{1- \left| \Delta(t; x,y) \right|^2} + 2\epsig\,\hatf(t), \label{eq: new scc - scc} 
	\end{align}
	and
	\begin{align}
		\| &\rho^F_{\Sy}(t) - \rho^{\textup{target} (0)}_{\Sy}(t)  \|_1 \\
		&\leq \epsig\,\hatf(t)+  
		2 \max_{x,y\in[-\pi,\pi]} \sqrt{1 - \left|\Delta(t;x,y)\right|^2}  .\label{eq: new s s}
	\end{align}
	Finally, to finish the proof we need find some simplifying upper bounds to the r.h.s. of Eqs. \eqref{eq: new scc - scc} and \eqref{eq: new s s} to conclude the bounds stated in Theorem \ref{thm:noemb physical t CTO}. 
	
	To this end we apply lemma \ref{lem:c}
	which implies, by identifying $c= \Delta(t;x,y)$
	\begin{align}
	    1-|\Delta(t;x,y)|^2 &\leq |1-
	    \Delta(t;x,y)^2|,\\
	    |1-
	    \Delta(t;x,y)| &\leq  
	    |1-
	    \Delta(t;x,y)^2|
	\end{align}
	Observe that we can make the identification $c= \Delta(t;x,y)$ since $|\Delta(t;x,y)|\leq 1$ follows from unitarity 
	and $|1-c|=|1-\Delta(t;x,y)|\leq 1$ can be assumed w.l.o.g.. Indeed, if $|1-\Delta(t;x,y)|>1$,
	then the bounds would be greater than 2 (see Eqs. \eqref{eq: new scc - scc}, \eqref{eq: new s s} and recall $A\geq 1$), hence not relevant, since trace norm is always no greater than 2.

	
	We then obtain
	\begin{align}
		\|& \rho^F_{\Sy\cat\cl}(t) - \rho^F_{\Sy}(t)\otimes \rho_\cat^0(t)\otimes \rho_\cl^0(t)  \|_1 \\
		&\leq 2\epsig\,\hatf(t)
		+6 \sqrt{ A\,  \max_{x, y\in[-\pi,\pi]} \left|1-   \Delta^2(t; x, y) \right| }
		\label{eq: new s s 1}
	\end{align}
	where we used $A\geq1$, 
	and 
	\begin{align}
		\| \rho^F_{\Sy}&(t) - \rho^{\textup{target} (0)}_{\Sy}(t)  \|_1 \leq \epsig\,\hatf(t)\\
		&+ 
		\, \sqrt{A}
		\max_{x,y\in[-\pi,\pi]} \left|1 - \Delta^2(t;x,y)\right| .\label{eq: new s s 2}
	\end{align}
\end{proof}
where $A$ and $\Delta(t;x,y)$ are given by \eqref{eq:A-Delta-def}.
Finally, since $\tr [\rho^2]\leq 1 $ for any normalised density matrix $\rho$, 
we have 
\begin{equation}
	A\leq d_\Sy d_\cat d_\G \tr\left[{\tau_\G}^2\right].
\end{equation}
Inserting this into  Eqs. \eqref{eq: new s s 1}, \eqref{eq: new s s 2} we get the thesis of Theorem \ref{thm:noemb physical t CTO}.
\cblack

\subsection{Calculating $\Delta(t;x,y)$ for the Idealised Momentum Clock}\label{sec:proof Delat=1 for idelaised clock}
In the case of the idealised momentum clock, we have $\hat H_\cl=\hat p$, $\hat H_\cl^\textup{int}=g(\hat x)$, where $\hat x$ and $\hat p$ are the position and momentum operators of a particle in one dimension satisfying the Weyl form of the canonical commutation relations, $[\hat x,\hat p]=\mi$, while $g$ is an integrable function from the reals to the reals, normalised such that $\int_\rr g(x) dx=1$.\footnote{One can also come to the same conclusions for the idealised momentum clock on a circle, rather than a line. In this case, $[\hat x,\hat p]$ still satisfy the Heisenberg form of the canonical commutation relations, but not the Weyl form. See \cite{garrison1970canonically} for details.}

Therefore, we find for the idealised momentum clock, for $z,y\in\rr$
\begin{align}
	&\Delta(t;z,y)\\
	 =&\bra{\rho_\cl^0} \hat \Gamma_\cl^\dag (z,t) \hat \Gamma_\cl(y,t) \ket{\rho_\cl^0}\\
	=& \me^{-\mi (z-y)\hatf(t)}  \bra{\rho_\cl^0} \me^{\mi t \hat p + \mi z t g(\hat x)} \me^{-\mi t \hat p - \mi y t g(\hat x)} \ket{\rho_\cl^0}\\
	=& \,\me^{-\mi (z-y)\hatf(t)} \!\! \int_\rr\!\!\!\! dx \bra{\rho_\cl^0} \me^{\mi t \hat p + \mi z t g(\hat x)} \ketbra{x}{x} \me^{-\mi t \hat p - \mi y t g(\hat x)} \ket{\rho_\cl^0}.
\end{align}

We can now use the relation $\hat p=-\mi \frac{\partial}{\partial x}$ and solve the 1st order 2 variable differential equation resulting from the Schr\"odinger eq. for the Hamiltonian $ \hat p + y g(\hat x)$ and initial wave-function $\braket{x|\rho_\cl^0}$. Plugging the solution into the above, we arrive at
\begin{align}
	\Delta(t;z,y) = 
	\me^{-\mi (z-y)\hatf(t)} \!\! \int_\rr\!\!\! dx \left|\braket{x|\rho_\cl^0}\right|^2 \me^{-\mi (y-z) \int_x^{x+t} g(x') dx'}\!. \label{eq: idelaise dmomentum inter 1}
\end{align}

We now choose the support of the initial wave-function $\braket{x|\rho_\cl^0}$ to be $x\in[x_{\psi l}, x_{\psi r}]$ and the support of $g(x)$ to be $x\in[x_{g l}, x_{g r}]$. Noting that
\begin{align}
	\int_x^{x+t} g(x') dx'= \begin{cases}
		0 &\mbox{ if } x+t\leq x_{gl}\\
		1 &\mbox{ if } x\leq x_{gl} \text{ and } x+t\geq x_{gr},
	\end{cases}
\end{align}
and taking into account the support interval of $\braket{x|\rho_\cl^0}$,  we conclude
\begin{align}
	&\int_\rr dx \left|\braket{x|\rho_\cl^0}\right|^2 \me^{-\mi (y-z) \int_x^{x+t} g(x') dx'}\\
	&=\begin{cases}
		1 &\mbox{ if } t \leq x_{gl}-x_{\psi r}\\
		\me^{\mi (z-y)} &\mbox{ if } t \geq x_{gr}-x_{\psi l}
	\end{cases}. \label{eq: idelaise dmomentum inter 2}
\end{align}
Therefore, choosing $t_1=x_{gr}-x_{\psi l}$ and $t_2=x_{gr}-x_{\psi l}$, from Eq. \eqref{eq: idelaise dmomentum inter 1} we arrive at
\begin{align}
	\Delta(t;x,y)= 1\quad \forall x, y \in[-\pi,\pi],
\end{align}
as claimed in Sec. \ref{sec: main text CTO case} of the main text. Furthermore, note that the derivation holds for all $t_1<t_2$ by appropriately choosing the parameters $x_{gr},x_{\psi l},x_{gr},x_{\psi l}$.


\twocolumngrid


%


\onecolumngrid
\newpage

\tableofcontents
\newpage
\appendix

\begin{center}
	{\huge  Supplementary Material}
\end{center}

\renewcommand\appendixname{Supplementary}
\renewcommand\appendixpagename{Supplementary}

\section{Additional Information for Proof of Theorem \ref{thm:noemb physical}}\label{proof of thm:noemb}
In this section we provide the additional detailed steps of the proof of Theorem \ref{thm:noemb physical} which supplement section \ref{Main proof section}.

\subsection{Preliminaries for proof of Theorem \ref{thm:noemb physical}}

\subsubsection{Entropies, divergences:  definitions and properties}\label{sec:entropies divergences def and properties}
In this section (and throughout this~\app\ unless stated otherwise), $\cP_d$ will denote the set of normalised probability vectors in dimension $d$. Vectors $p,q\in\cP_d$ will have entries denoted $p_k,q_k$ respectively. We will also let $\idv_d\in\cP_d$ be the uniform probability vector, namely ${[\idv_d]}_k=1$, for $k=1,\ldots,d$. 

\begin{definition}[R\'enyi $\alpha$-entropies]\label{def:renyi entrpies}
	The R\'enyi $\alpha$-entropies for $\alpha\in\rr$ are defined to be 
	\begin{eqnarray} \label{eq:deef renyi entrpies}
	S_\alpha(p)&=\frac{\textup{sgn}(\alpha)}{1-\alpha} \ln \sum_{k=1}^d p_k^\alpha,\\
	\textup{sgn}(\alpha)&= 
	\begin{cases}
	\;\;\,1& \text{for } \alpha\geq 0 \\
	-1& \text{for } \alpha< 0 
	\end{cases}
	\end{eqnarray}
	where the singular point at $\alpha=1$ is defined by demanding that the R\'enyi $\alpha$-entropies are continuous in $\alpha\in\rr$, and we use the conventions $\frac{a}{0}=\infty,$ for $a>0$ and $0\ln 0=0$, $0^0=0$.
\end{definition}

We will use the R\'enyi entropies evaluated on quantum states $\rho$ in $d$ dimensions. In which case $S_\alpha(\rho):=S_\alpha(p_\rho)$, where $p_\rho\in\cP_d$ denotes the eigenvalues of $\rho$. This convention for extending the definition of functions evaluated on $\cP_d$ to functions evaluated on quantum states of dimension $d$, will be used throughout.

The $\alpha=1$ value is of particular interest, since it corresponds to the Shannon entropy, namely
\begin{equation}
S_1(p):=\lim_{\alpha\rightarrow 1} S_\alpha(p)=-\sum_{k=1}^d p_k\ln p_k.
\end{equation}
Note that the R\'enyi $\alpha$-entropies were originally defined in \cite{renyi} for $\alpha\geq 0$ only, but later extended to $\alpha\in\rr$ for convenience in 
\cite{secondlaw}. Note that $S_\alpha(p)$ can be infinite. The other functions defined in this section also have this property.

For $p\in\cP_d$ define
\be
\label{eq:f1}
f_\alpha(p)=\sum_{i=1}^d p_i^\alpha,
\ee
for $\alpha\not \in \{ 0,1\}$ and 
\begin{equation}
\label{eq:f2}
f_0(p)=-\sum_{i=1}^d \ln p_i,\quad f_1(p)=S_1(p)=-\sum_{i=1}^d p_i\ln p_i.
\end{equation}
If some of $p_k$ is equal to zero, 
the value of $f_\alpha$ for  $\alpha\leq0$ is set to infinity.
For $\alpha\in[0,1]$ these functions are concave, and for $\alpha>1$ convex.

\begin{definition}[Tsallis Entropy]
	Tsallis-Aczel-Daroczy entropy is as follows.
	\be\label{eq:Tsalis def}
	T_\alpha(p)= \textup{sgn}(\alpha)\frac{1-\sum_i p_i^\alpha}{\alpha-1},
	\ee
	for $\alpha\not=0,1$ and $\alpha>0$.
\end{definition}
The Tsallis Entropy is convex, subadditive (but not additive), and for $\alpha=1$, through a limit, it gives Shannon entropy.

\begin{definition}[Hellinger Relative Entropy]
	Hellinger divergence for $\alpha\in R$ is as follows
	\begin{equation}\label{eq:def Hellinger}
	\hcal_\alpha(p|q)=\frac{sgn(\alpha)}{\alpha-1} \left(\sum_i p_i^\alpha q_i^{1-\alpha} -1\right),
	\end{equation}
	where the singular points are defined by continuity and, in addition to the conventions in Def. \eqref{def:renyi entrpies}, we have $\frac{0}{0}=0$. 
\end{definition}

We have 
\be
\lim_{\alpha\to 1} \hcal_\alpha(p|q)=D(p|q)
\ee
where 
\be
D(p|q) = \sum_i p_i \ln\frac{p_i}{q_i}
\ee
is \emph{Kulback-Leibler entropy}. 
Moreover $\hcal_\alpha$ is monotonically increasing 	in $\alpha$ for $\alpha\in (0,\infty)$. In particular  
For  $\alpha\geq 1$ we have 
\be
\label{eq:HD}
\hcal_\alpha(p|q)\geq D(p|q).
\ee
We have also Pinsker inequality
\be
\label{eq:pinsker}
D(p|q)\geq \frac12\|p-q\|_1 ^2
\ee
We have 
\be
\label{eq:THel}
\hcal_\alpha(p|\idv_d/d)= d^{\alpha-1}\left(T_\alpha(\idv_d/d)-T_\alpha(p)\right),
\ee
for $\alpha\not=0,1$ and $\alpha>0$.

\begin{lemma}[Poor Subadditivity]\label{lem:poor sub}
	Let $f$  be Schur convex, i.e. $x\succ y$ implies $f(x)\geq f(y)$.
	Then 
	\be
	\label{eq:poor sub}
	f(\rho_{\A\B})\geq f\left(\rho_\A\otimes \frac{\id_\B}{d_\B}\right),
	\ee
	where $\id_\B$ is the identity operator on $\B$.
\end{lemma}
\begin{proof}
	Note that the state 
	\be
	\label{eq:state twirl}
	\rho_\A\otimes \frac{\id_\B}{d_\B}
	\ee
	can be obtained from  $\rho_{\A\B}$ by a mixture of unitaries (applying Haar random or discrete 2-design family) 
	unitary on subsystem $\B$). Thus by Uhlmann, the spectrum of original state $\rho_{\A\B}$ majorizes the spectrum of 
	the state \eqref{eq:state twirl}. Thus by Schur convexity of $f$ we get \eqref{eq:poor sub}. 
\end{proof}

\subsubsection{Noisy operations, catalythic noisy operations, majorization and trumping}
So called noisy operations  \cite{NOsDef} are a subclass of thermal operations introduced earlier in \cite{janzing2000thermodynamic,streater2009statistical}.
As explained in Sec. \ref{sec:CTO def} of the main text, these are all operations that can be composed of: (i) adding the free resource with a maximally mixed state. (ii) applying an arbitrary unitary transformation. (iii) taking the partial trace.

It was shown that when the input and output state belong to a Hilbert space of the same dimension, the class of noisy operations is equivalent to mixture of unitiaries. 
Therefore the condition that $\rho$ can be transformed into 
$\sigma$ is equivalent to majorization: $\rho$ can be transformed into $\sigma$ iff 
the spectrum $p$ of $\rho$ majorizes the spectrum  $q$ of $\sigma$.
We say that $p\in\cP_d$  majorizes $q\in\cP_d$  if  for all $l=1,\ldots ,d$ 
\be
\sum_{i=1}^l p_i^\downarrow \geq \sum_{i=1}^l q_i^\downarrow, 
\ee 
where $p^\downarrow$ is a vector obtained by arranging the components of $p$ in decreasing order:
$p^\downarrow = (p^\downarrow_1, \ldots, p^\downarrow_k)$ where  $p^\downarrow_1\geq  \ldots \geq  p^\downarrow_k$.
We now explain how catalytic noisy operations can be understood in terms of so called ``trumping''. As mentioned in Sec. \ref{sec:CTO def} of the main text, these are the noisy operations for which  one is allowed to use an additional system as a catalyst | namely the additional system has to be returned to its initial state after the process. This idea,  was first introduced to quantum information theory in the context of entanglement transformations \cite{Jonathan_1999}.

Namely, we say that $p\in\cP_d$ can be trumped into $q\in\cP_d$ (or, that $p$ catalytically majorizes $q$) if there exists some $k\in\nn^+$ and $r\in\cP_k$ such that 
\begin{equation}
p\ot r \succ q\ot r.
\end{equation}

Klimesh \cite{klimesh2007inequalities} and Turgut 
\cite{Turgut-trumping} provided necessary and sufficient conditions for $p$ to be trumped into $q$. 
Here we present conditions in the form provided by Klimesh. 
\begin{theorem}[Klimesh \cite{klimesh2007inequalities}]
	\label{thm:Klimesh}
	Consider $p\in \cP_d $ and $q\in \cP_d$ which do not both contain components  equal to zero (i.e. at least one of them is full rank),
	and let $p\not =q$.  
	Then $p$ can be trumped into $q$ if and only if for all $\alpha\in (-\infty, \infty)$ we have 
	\be
	g_\alpha(p)> g_\alpha(q)
	\ee
	where the functions $g_\alpha$
	are given by 
	\begin{align}
	\label{eq:g-alpha}
	g_\alpha(p)=
	\left \{
	\begin{array}{ll}
	\ln \sum_{i=1}^d p_i^\alpha & \quad \text{\rm for}\quad \alpha>1\\
	\sum_{i=1}^d p_i \ln p_i
	& \quad \text{\rm for}\quad \alpha=1\\
	-\ln \sum_{i=1}^d p_i^\alpha & \quad \text{\rm for}\quad 0<\alpha<1\\
	-\sum_{i=1}^d  \ln p_i
	& \quad \text{\rm for}\quad \alpha=0\\
	\ln \sum_{i=1}^d p_i^\alpha & \quad \text{\rm for}\quad \alpha<0\\
	\end{array} \right.
	\end{align}
\end{theorem}

\subsection{Lemmas on norms and fidelity}
This lemma says that if a state is close to a product, then it is also close to a product of 
its reductions.
\begin{lemma}
	\label{lem:product}
	We have for arbitrary states  $\rho_{\A\B}$, $\eta_\A$, $\eta_\B$ and pure state $\psi_\B$,
	\be
	\|\rho_{\A\B}-\rho_\A \otimes \eta_\B\|_1 \leq 2  \|\rho_{\A\B}-\eta_\A \otimes \eta_\B\|_1 
	\ee
	and 
	\be
	\|\rho_{\A\B}-\rho_\A \otimes \rho_\B\|_1 \leq 3  \|\rho_{\A\B}-\eta_\A \otimes \eta_\B\|_1 
	\ee
\end{lemma}
\begin{proof}
	We have 
	\ben
	&&\|\rho_{\A\B}-\rho_\A \otimes \eta_\B\|_1 \leq   
	\|\rho_{\A\B}-\eta_\A \otimes \eta_\B\|_1 +\|\eta_\A \otimes \eta_\B-\rho_\A \otimes \eta_\B\|_1 = \nonumber \\
	&& =  \|\rho_{\A\B}-\eta_\A \otimes \eta_\B\|_1 +\|\eta_\A -\rho_\A \|_1  \leq 2 \|\rho_{\A\B}-\eta_\A \otimes \eta_\B\|_1 
	\een
	The second inequality we prove in a similar way. 
\end{proof}

Next lemma says that, if a reduced state is close to a pure state then the total 
state is close to	a product (of its reduction tensored with the pure state)
\begin{lemma}
	\label{lem:close to pure}
	We have 
	\be
	\|\rho_{\A\B} -\rho_\A \otimes \psi_\B\|_1 \leq 2 \sqrt{\|\rho_\B-\psi_\B\|_1 }
	\ee
\end{lemma}

\begin{proof}
	Consider $F(\rho_\B,\psi_\B)$. We have 
	\be
	F(\rho_\B,\psi_\B)= F(\phi_{\A\B\C},\psi_{\A\C}\otimes \psi_\B)
	\ee
	where $\phi_{\A\B\C}$ is a purification of $\rho_\B$ which we are free to choose the way we want, 
	and $\psi_{\A\C}$ is some pure state. By data processing we have 
	\be
	F(\phi_{\A\B\C},\psi_{\A\C}\otimes \psi_\B)\leq F(\rho_{\A\B},\sigma_\A\otimes \psi_\B)
	\ee
	where $\sigma_\A$ is reduction of $\psi_{\A\C}$.  
	Using it and  twice Fuchs-Graaf we thus get:
	\be
	\|\rho_{\A\B}-\sigma_\A\ot\psi_\B\|_1 \leq 2 \sqrt{1- F^2(\rho_{\A\B},\sigma_\A\otimes \psi_\B)} \leq  
	2 \sqrt{1- F^2(\rho_\B, \psi_\B)}\leq 2 \sqrt{\|\rho_\B-\psi_\B\|_1 }
	\ee
	
	Now, the proposition says that if two states have closed corresponding reductions, and one of the reductions 
	is close to a pure state, then the states are close to one another.
	
\end{proof}

\begin{prop}
	\textsl{}\label{prop:red}
	Suppose that $\|\rho_\A-\sigma_\A\|_1 \leq \epsilon_1$, $\|\rho_\B-\sigma_\B\|_1 \leq \epsilon_2$, $\|\sigma_\B - \psi_\B\|_1 \leq \epsilon_3$. Then 
	\be
		\|\rho_{\A\B}-\sigma_{\A\B}\|_1 \leq \ep_1 + 2 \sqrt{\ep_2 + \epsilon_3}  +  2 \sqrt{\epsilon_3} .
	\ee
\end{prop}

\begin{proof}
	By triangle inequality we have 
	\be
	\|\rho_\B-\psi_\B\|_1 \leq  \ep_2+\ep_3. 
	\ee
	By lemma  \ref{lem:close to pure}
	we have 
	\ben
	&&\|\sigma_{\A\B}-\sigma_\A \ot \psi_\B\|_1 \leq 2 \sqrt{\|\sigma_\B-\psi_\B\|_1 } \nonumber \\
	&&\|\rho_{\A\B}-\rho_\A \ot \psi_\B\|_1 \leq 2 \sqrt{\|\rho_\B-\psi_\B\|_1 } \nonumber \\
	\een
	Sandwiching $\|\rho_{\A\B} - \sigma_{\A\B}\|_1 $ with the above, we finish the proof. 
\end{proof}

\subsection{From approximate to strict inequalities}

\subsubsection{Main lemmas}

\begin{lemma}[From approximate to strict inequalities through smoothing]
	\label{lem:approx-strict}
	Let  $f$ be  a   concave, non negative  function of $p\in\cP_d$ such that $f(p)<f(\idv_d/d)$ for any $p\not=\idv_d/d$.
	Suppose that for some  $\eta>0$  we have
	\be
	\label{eq:f eta}
	f(p) \leq f(q) +\eta
	\ee
	Then for $\ep$ satisfying $\ep \leq 1/2$ and 
	\be
	\label{eq:ep}
	\ep\geq\min_{\delta>0} \max \left\{\delta, \frac{2\eta}{ f(\idv_d/d)-\max_{\|\rho -\idv_d/d\|_1 \geq \delta/2}f(\rho)} \right\}
	\ee
	we have 
	\be
	f(p) \leq f(\tilde q(\ep)) - \min\{\eta, f(\idv_d/d) - f(p)\}
	\ee
	where $\tilde q(\ep)$ is given by 
	
	\be
	\tilde q (\ep)=
	\left\{
	\begin{array}{ll}
		\idv_d/d &\text{\rm when } \|q -\idv_d/d\|_1 <\ep \\
		(1-\ep) q + \ep \idv_d/d &\text{\rm when } \|q -\idv_d/d\|_1 \geq \ep, \\
	\end{array}
	\right.
	\ee
	%
\end{lemma}
\begin{proof}
	If state $q$ satisfies $\|q - \idv_d/d \|_1  < \ep$  then  $\tilde q (\ep)= \idv_d/d$. Then,
	\be
	f(p) = f(\idv_d/d) + f(p) - f(\idv_d/d) = f(\tilde q(\ep))  + (f(p) - f(\idv_d/d)).
	\ee
	Thus trivially
	\be
	f(p)\leq f(\tilde q(\ep))  - \left(f(\idv_d/d) - f(p)\right)\leq f(\tilde q(\ep)) - \min\{\eta, f(\idv_d/d) - f(p)\}.
	\ee
	Now suppose that 
	\be
	\|q-\idv_d/d\|_1 \geq \ep.
	\ee
	From concavity of $f$ we have 
	\be
	f(\tilde q_\ep)\geq (1-\ep)f(q)  + \ep f(\idv_d/d)
	\ee
	hence 
	\be
	f(q) \leq \frac{f(\tilde q_\ep)}{1-\ep} - \frac{\ep}{1-\ep} f(\idv_d/d).
	\ee
	Then from  \eqref{eq:f eta}
	\be
	f(p)\leq f(q)+\eta \leq \frac{f(\tilde q_\ep)}{1-\ep} - \frac{\ep}{1-\ep} f(\idv_d/d)  +\eta = 
	f(\tilde q(\ep))+\frac{\ep}{1-\ep}\bigl(f(\tilde q(\ep)) - f(\idv_d/d)\bigr)  + \eta \leq 
	f(\tilde q(\ep))+\ep\bigl(f(\tilde q(\ep)) - f(\idv_d/d)\bigr)  + \eta. 
	\ee
	Thus it remains to show that the $\ep$ satisfying \eqref{eq:ep} and $\ep \leq 1/2$ satisfies 
	\be
	\label{eq:ep1f}
	\ep\bigl(f(\idv_d/d) - f(\tilde q(\ep))\bigr) \geq 2\eta.
	\ee
	To this end, note that \eqref{eq:ep} implies 
	\be
	\label{eq:ep2}
	\ep\geq\frac{2\eta}{f(\idv_d/d) - \max_{\|\rho -\idv_d/d\|_1 \geq \ep/2} f(\rho)}.
	\ee
	Then, note that 
	\be
	\|\tilde q(\ep)-\idv_d/d\|_1 =\|(1-\ep)q +\ep \idv_d/d - \idv_d/d\|_1 =(1-\ep)\|q-\idv_d/d\|_1 \geq(1-\ep)\ep.
	\ee
	Thus, for $\ep\leq 1/2$ we have $\|\tilde q_\ep-\idv_d/d\|_1 \geq \ep/2$,  so that 
	\be
	f(\tilde q(\ep))\leq \max_{\|\rho - \idv_d/d\|_1 \geq \ep/2}  f(\rho),
	\ee
	hence \eqref{eq:ep2} implies 
	\be
	\ep \geq \frac{2\eta}{f(\idv_d/d)-f(\tilde q(\ep))}, 
	\ee
	which is equivalent to \eqref{eq:ep1f}.
\end{proof}

\begin{lemma}
	\label{lem:f g}
	Let $g$ be convex, non-negative function with the domain $D\in R$. Let also $g$ be multiplicative,
	i.e. $g(xy)=g(x) g(y)$ and $g(1)=1$. Let us denote $f_g(p)=\sum_{i=1}^{d} g(p_i)$. Then 
	for any probability distribution 
	$p\in\cP_d$ satisfying $\|p - \idv_d/d\|_1 \geq \delta$ we have 
	\be
	f_g(p) - f_g(\idv_d/d) \geq \frac{f_g(\idv_d/d)}{f_g(\idv_2/2)} \left(f_g(p^\delta_{(2)})- f_g(\idv_2)\right),
	\ee
	where
	\be
	p_{(2)}^\delta=\left\{\frac{1+\delta/2}{2}, \frac{1-\delta/2}{2}\right\}.
	\ee
	If $g$ is concave, and otherwise satisfies all the above assumptions we have 
	\be
	f_g(\idv_d/d) - f_g(p) \geq \frac{	f_g(\idv_d/d)}{f_g(\idv_2/2)} \left(f_g(\idv_2/2)-f_g(p^\delta_{(2)}) \right).
	\ee
\end{lemma}


\begin{remark}
	\label{rem:tsallis}
	Lemma \ref{lem:f g} applies to $g(x)=x^\alpha$ for $\alpha>0$. The function $f_g$ for such $g$ we will denote by $\fa$.
	We then obtain:
	\be
	\label{eq:fa convex}
	\fa(p)-\fa(I_d/d) \geq  \frac{d^{1-\alpha}}{2^{1-\alpha}} \left(\fa(p_{(2)}^\delta) - \fa(I_2/2)\right)
	\ee
	for $\alpha>1$ (i.e. when $x^\alpha$ is convex), and 
	\be
	\label{eq:fa concave}
	\fa(I_d/d)-\fa(p) \geq  \frac{d^{1-\alpha}}{2^{1-\alpha}} \left(\fa(I_2/2) - \fa(p_{(2)}^\delta)\right)
	\ee
	for $\alpha\in(0,1)$ (i.e. when $x^\alpha$ is concave).
	From these, one gets:
	\be
	\label{eq:Tdto2}
	T_\alpha (\idv_d/d)- T_\alpha(p) \geq \frac{d^{(1-\alpha)}}{2^{(1-\alpha)}} \left(T_\alpha(\idv_2/2) - T_\alpha(p_{(2)}^\delta)\right)
	\ee
	for all  $\alpha >0$ 
	(for $\alpha=1$ it is obtained by continuity, and gives inequality for Shannon entropies). 
\end{remark}

\begin{proof}[Proof of Lemma \ref{lem:f g}]
	Let $g$ be concave. Then $f_g(p)$ is concave as a function of probability distribution $p\in\cP_d$. 
	For any distribution $p$ we consider its twirled version, that depends on just two parameters: 
	the number $k$ of $p_i$'s  greater than or equal to $1/d$ and $\delta=\|p - \idv_d/d\|_1 $. 
	\be
	\label{eq:twirled p}
	\tilde p =  \Bigg\{\underbrace{\frac1d + \frac{\delta/2}{k},\ldots,\frac1d + \frac{\delta/2}{k}}_{k},
	\underbrace{\frac1d - \frac{\delta/2}{d-k},\ldots,\frac1d + \frac{\delta/2}{d-k}}_{d-k}\Bigg\}.
	\ee
	The $\tilde p$ can be obtained from $p$ by mixture of permutations (we consider two subset of $p_i$': those larger than $1/d$ and those smaller than or equal, and randomly permute elements within each of the subsets. Hence by concavity we have 
	\be
	f_g(p)\leq f_g(\tilde p)= d (r_1 g(x_1) + r_2 g(x_2)) 
	\ee
	where we denoted $r_1=k/d, r_2=(d-k)/d$ and $x_1=1/d +\delta/2d, x_2=1/d-\delta/(2(d-k))$.  Note that 
	$r_1+r_2=1$, $r_1 x_1 + r_2 x_2=1/d$. One finds (see Fig. \ref{fig:concave})
	\begin{figure}[ht]
		\centering
		\includegraphics[width=12cm]{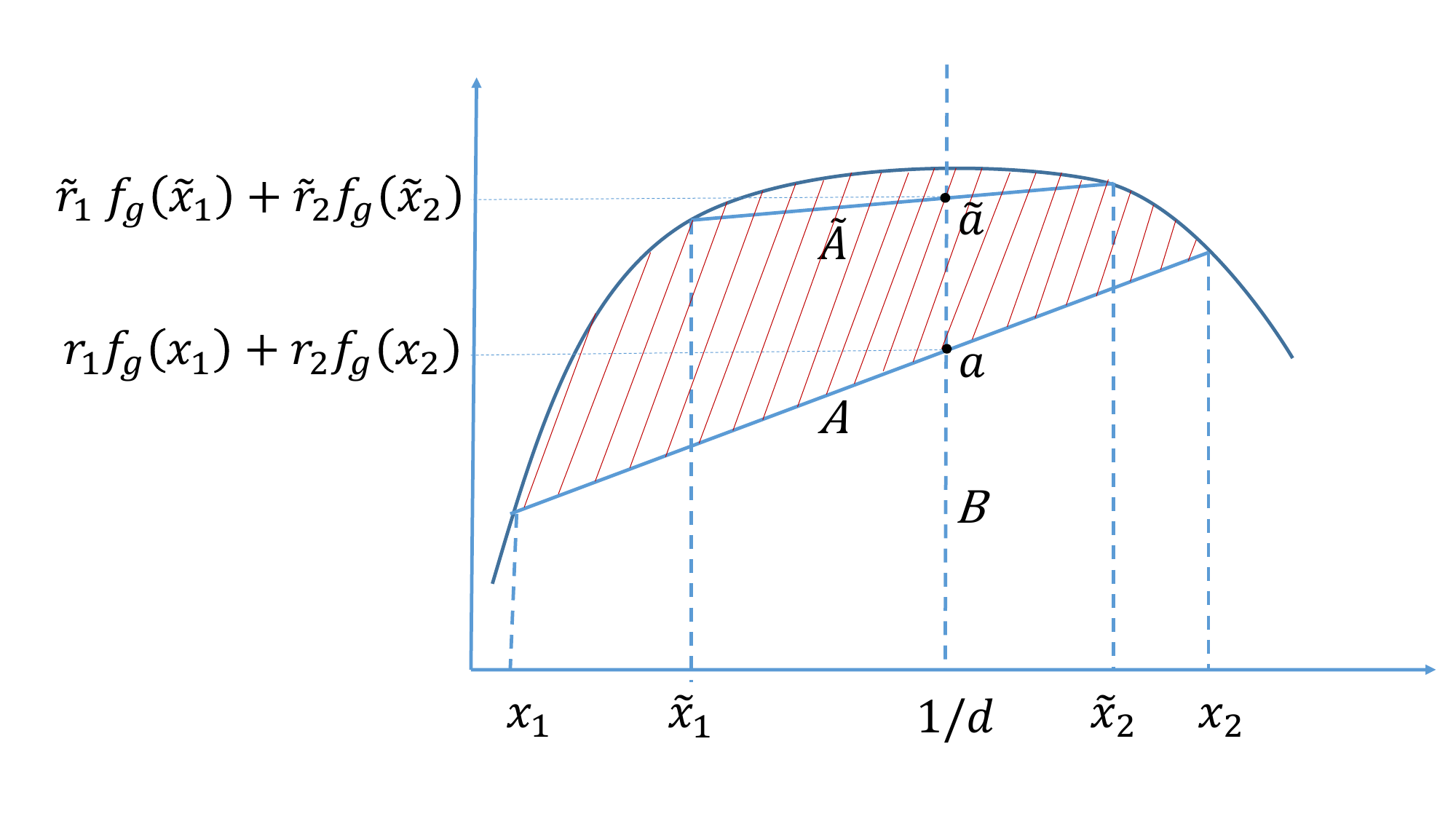}
		\caption{\label{fig:concave} Geometric proof of the inequality  \eqref{eq:fig-inequality}. Due to concavity of $f$ the interval $A$ between the points $(x_1,f(x_1))$ and $(x_2,f(x_2))$ together 
			with the part of the graph of the function laying between these two points enclose a convex body (indicated in red). Therefore  the interval $\tilde A$ between the points $(\tilde x_1,f(\tilde x_1))$ and $(\tilde x_2,f(\tilde x_2))$ must lie within the body. By assumption, the latter interval 
			has nonempty intersection $\tilde a$ with the line $x=1/d$. This intersection must be therefore above the intersection $a$ of the latter line with the interval $A$, which means that inequality \eqref{eq:fig-inequality} holds.}
	\end{figure}
	that if $x_2\leq \tilde x_2\leq 1/d\leq \tilde x_1 \leq x_1$ then, due to concavity of the function $f_g$, we have 
	\begin{equation}
	\label{eq:fig-inequality}
	r_1 f_g(x_1) + r_2 f_g(x_2)  \leq \tilde r_1 f_g(\tilde x_1) + \tilde r_2 f_g(\tilde x_2)
	\end{equation}
	provided $\tilde r_1 +\tilde r_2=1$, $\tilde r_1 \tilde x_1 + \tilde r_2 \tilde x_2=1/d$.
	In Fig. \ref{fig:concave} a graphical proof is given. The analytical argument is as follows.  One can first easily prove by concavity that the interval $\tilde A$ is above the interval $A$. Namely, it is enough to prove that the ends of the interval $\tilde A$ are above interval $A$, which in turn, follows directly from concavity. Now, the left hand side of  \eqref{eq:fig-inequality} 
	is the second argument of the 
	point from interval $A$ with
	first argument being $1/d$, 
	while the right hand side  is the second argument of the point from the interval $\tilde A$ (first argument being $1/d$), hence the inequality follows.
	
	Now let us exploit \eqref{eq:fig-inequality}. 
	Since $1\leq k\leq d$ we can choose
	\be
	\tilde x_1= \frac1d + \frac{\delta/2}{d},\quad 
	\tilde x_2= \frac1d - \frac{\delta/2}{d}
	\ee
	and $\tilde r_1=\tilde r_2 = \frac12$.  We thus obtain 
	\ben
	&&f_g(p) \leq d\left[ \frac12  g\left( \frac1d + \frac{\delta/2}{d}\right) + 
	\frac12 g\left(\frac1d - \frac{\delta/2}{d}\right) \right] =\nonumber \\
	&&= \frac{d}{2}\frac{g(2)}{g(d)} \left[g\left(\frac{1+\delta/2}{2}\right) + 
	g\left(\frac{1-\delta/2}{2}\right) \right]= \nonumber \\
	&&= \frac{f_g(\idv_d/d)}{f_g(\idv_2/2)} f_g(p_{(2)}^\delta)
	\een
	where multiplicativity of $g$ was used. From this we obtain 
	\eqref{eq:fa concave}. Eq. \eqref{eq:fa convex} is proven similarly. 
\end{proof}

\begin{lemma}
	\label{lem:pmax}
	Let $p_{\min}(p)=\min_i p_i$  and $p_{\max}(p)=\max_i p_i$. Then for arbitrary $p$ such that $\|p-\idv_d/d\|_1 =\delta$ we have 
	\ben
	p_{\min}(p)\leq \frac{p_{\min}(\idv_d/d)}{p_{\min}(\idv_2/2)} p_{\min}(p^\delta_{(2)})=\frac{2}{d} p_{\min}(p^\delta_{(2)})
	\een
	and 
	\ben
	p_{\max}(p)\geq \frac{p_{\max}(\idv_d/d)}{p_{\max}(\idv_2/2)} p_{\max}(p^\delta_{(2)})=  \frac{2}{d} p_{\max}(p^\delta_{(2)})
	\een
	where
	\be
	p_{(2)}^\delta=\left\{\frac{1+\delta/2}{2}, \frac{1-\delta/2}{2}\right\}.
	\ee
\end{lemma}
\begin{proof}
	We use again the twirled version of $p$ of Eq. \eqref{eq:twirled p}
	By convexity of $p_{\max}$, we have 
	\be
	p_{\max}(p)\geq p_{\max}(\tilde p)  \geq \frac1d +\frac{\delta/2}{d}=\frac2d p_{\max} 
	\ee
	Similarly one proves for $p_{\min}$. 
\end{proof}
\begin{prop}
	\label{prop:strict}
	Let  $p,q \in\cP_d$ and define
	\be
	\label{eq:tildeq}
	\tilde q(\ep)=
	\left\{
	\begin{array}{ll}
		\idv_d/d &\text{\rm when } \|q -\idv_d/d\|_1 < \ep \\
		(1-\ep) q + \ep \idv_d/d &\text{\rm when } \|q -\idv_d/d\|_1 \geq \ep \\
	\end{array}
	\right.
	\ee
	Denote also 
	\ben
	\label{eq:ep>}
	&&\ep_T(\alpha)=(16\,\eta_\alpha \, d^{\alpha-1})^\frac13\quad \text{\rm for } \alpha \geq 1  \\
	\label{eq:ep<}
	&&\ep_T(\alpha)=(16\,\eta_\alpha\,  d^{\alpha-1}\alpha^{-1})^\frac13 \quad \text{\rm for } \alpha\in(0,1) \\
	\label{eq:epinfty}
	&&\ep_\infty(\alpha) = 4 \sqrt{\frac{\ln d}{\alpha}+\eta_\infty} \quad \text{\rm for } \alpha > 1 \\
	\label{eq:epzero}
	&&\epzero(\alpha) =  \left(\frac{d-1}{d}\right)^{\frac{1}{2\alpha}}\quad \text{\rm for } \alpha\in(0,1).
	\label{eq:ep1}
	\een
	Now assuming, that all the above epsilons ($\ep_T,\ep_\infty,\ep_0$) are no greater than $1/2$ 
	we have 
	\begin{itemize}
		\item[(i)] for $\alpha>1$ 
		\be
		S_\infty(p) \leq S_\infty (q) + \eta_\infty,\quad \text{\rm implies} \quad S_\alpha(p) \leq  S_\alpha(\tilde q(\ep_\infty(\alpha)) - 
		\min\left\{\eta_\infty, \ln d - S_1(p) \right\}
		\ee
		\item[(ii)] For $\alpha>0 $
		\be
		T_\alpha(p) \leq T_\alpha (q) + \eta_{\alpha},\quad \text{\rm implies} \quad T_\alpha(p) \leq T_\alpha(\tilde q(\ep_T(\alpha))
		- \min\left\{\eta_\alpha,  T_\alpha(\idv_d/d) - T_\alpha(p) \right\}
		\ee
		\item[(iii)] For $\alpha\in(0,1)$, for $p$ not full rank  we have  
		\be
		S_\alpha(p) \leq  S_\alpha(\tilde q(\epzero(\alpha))) - \frac12 \ln \frac{d}{d-1}
		\ee
	\end{itemize}
\end{prop}
\begin{proof}
	
	{\it Proof of (i)}.
	Fix some $\epsilon$, and consider $\tilde q(\ep)$ given by \eqref{eq:tildeq}. 
	Consider first the case when $\|q-\idv_d/d\|_1 \leq \ep$. Let us show, that in this case, $(i)$ holds for any $\epsilon$. Indeed, we have $\tilde q(\ep)=\idv_d/d$, hence 
	\be
	S_\alpha(p)=S_\alpha(\tilde q(\ep))-\left(\ln d - S_\alpha(p)\right).
	\ee
	Since $S_\alpha$ is monotonically decreasing in $\alpha$, 
	we can
	replace on the right-hand-side $S_\alpha(p)$ with  $S_1 (p)$
	\be
	\label{eq:leq-ep}
	S_\alpha(p)\leq S_\alpha(\tilde q(\ep))  - \left(S(\idv_d/d)-S_1(p)\right)
	\ee
	which is what we want.
	%
	Now we turn to less trivial case when  $\|q -\idv_d/d\|_1 \geq \ep$. We write 
	\be
	S_\infty(\tilde q(\ep)) - S_\infty(q)=- \ln\left((1-\ep) q_{\max} + \ep/d\right) + \ln q_{\max}= 
	- \ln\left((1-\ep)  + \frac{\ep}{d q_{\max}}\right).
	\ee
	By lemma \ref{lem:pmax} and $e^{-x}\geq 1-x$  we bound it further as follows 
	\be
	S_\infty(\tilde q(\ep)) - S_\infty(q) \geq -\ln\left(1-\ep\left(1 -\frac{1}{1+\ep/2}\right)\right) \geq \ep \frac{\ep/2}{1+\ep/2}
	\geq \frac{\ep^2}{4}
	\ee
	We now use the fact that  for $\alpha>1$ 
	\be
	S_\alpha(p) \leq \frac{\alpha}{1-\alpha} S_{\infty}(p).
	\ee
	We get 
	\begin{equation}
	S_\alpha(q_\ep) \geq S_\infty(q_\ep) \geq S_\infty(q) +\frac{\ep^2}{4} \geq 
	S_\infty(p) +\frac{\ep^2}{4}  - \eta_\infty
	\geq 
	\frac{\alpha-1}{\alpha} S_\alpha(p)  +\frac{\ep^2}{4}- \eta_\infty \geq 
	S_\alpha(p) - \frac{\ln d}{\alpha}  +\frac{\ep^2}{4}- \eta_\infty,
	\end{equation}
	where the second inequality holds by assumption of $(i)$, while the last inequality comes from 
	$S_\alpha(p)\leq \ln d$ for any $p\in\cP_d$.
	Thus, provided 
	\be
	\label{eq:lndalpha}
	- \frac{\ln d}{\alpha}  +\frac{\ep^2}{4}> 2 \eta_\infty,
	\ee
	we have $S_\alpha(q_\ep)>S_\alpha(p)+\eta_\infty$. 
	If we now take $\ep=\ep_\infty(\alpha)$ we see that \eqref{eq:lndalpha} is satsifed, hence 
	we obtain that for $\|q-\idv_d/d\|_1 \geq \ep_\infty(\alpha)$
	\be
	S_\alpha(p) \leq S_\alpha(\tilde q(\ep)) -\eta_\infty. 
	\ee
	Using this and \eqref{eq:leq-ep} we get 
	that for arbitrary $q$ we have 
	\be
	S_\alpha(p) \leq S_\alpha(\tilde q(\ep)) -\min\left\{\eta_\infty, \ln d-S_1(p)\right\}.
	\ee
	This ends the proof of part (i).
	
	{\it Proof of (ii)}. 
	Since $T_\alpha$ is concave function (in probability distributions) for all  $\alpha>0$ 
	from lemma \ref{lem:approx-strict} we obtain
	that for $\ep$ given by 
	\be
	\ep=\min_\delta\max\left\{\delta, \frac{2 \eta}{T_\alpha(\idv_d/d)-\max_{p} T_\alpha(p)}\right\},
\end{equation}  
where maximum is taken over all $p\in\cP_d$ satisfying $\|p-\idv_d/d\|_1 \geq \delta/2$, we have 
\be
T_\alpha(p)\leq T_\alpha(\tilde q_\ep) - \min\left\{\eta_\alpha, T_\alpha(\idv_d/d)-T_\alpha(p)   \right\}.
\ee
By Eq.  \ref{eq:Tdto2} (a consequence of lemma \ref{lem:f g}) we have for such $p$
\be
T_\alpha(\idv_d/d)-\max_p T_\alpha(p)\geq \left(\frac{d}{2}\right)^{1-\alpha} \left( T_\alpha(\idv_2/2)- T_\alpha(p_{(2)}^\delta) \right)
\ee
for all $\alpha>0$.  The right hand side can be expressed in terms of Hellinger relative entropy  (Eq. \ref{eq:def Hellinger}) by virtue of Eq. \eqref{eq:THel}
\be
\left(\frac{d}{2}\right)^{1-\alpha} ( T_\alpha(\idv_2/2)- T_\alpha(p_{(2)}^\delta) =d^{1-\alpha}\hcal_\alpha(p_{(2)}^\delta|\idv_2/2).
\ee
We thus obtained that for all $\alpha >0$ 
\be
\ep\leq \min_{\delta}\max \left\{ \delta, \frac{2\eta}{d^{1-\alpha}\hcal_\alpha(p_{(2)}^\delta|\idv_2/2)}\right\}.
\end{equation}
Now from lemma \ref{lem:HvsD} we have
\begin{equation}
\hcal_\alpha(p_{(2)}^\delta|\idv_2/2) \geq \frac{\alpha}{8} \delta^2
\end{equation}
for $\alpha \in (0,1)$ 
and 
\begin{equation}
\hcal_\alpha(p_{(2)}^\delta|\idv_2/2) \geq \frac{1}{8} \delta^2
\end{equation}
for $\alpha>1$. 
Thus  for $\alpha>1$ 
\be
\ep\leq \min_\delta\left\{\delta,(16 \eta d^{\alpha-1})/\delta^2\right\}= (16 \eta d^{\alpha-1})^{1/3}
\ee
and for $\alpha\in(0,1)$ 
\be
\ep\leq \min_\delta\left\{\delta,(16 \eta d^{\alpha-1}\alpha^{-1})/\delta^2\right\}= (16 \eta d^{\alpha-1}\alpha^{-1})^{1/3}.
\ee
The case $\alpha=1$ we get by taking the limit $\alpha\to 1$.

{\it Proof of (iii)}. 
For $\alpha>1$ for all  full rank distributions  $p\in\cP_d$ we have 
\be
S_\alpha(p) \geq d p_{\min}^\alpha 
\ee
where $p_{\min}$ is the minimal element of $p$.  
For distribution $\tilde q(\epzero)=(1-\epzero)q +\epzero \idv_d/d$, since $p_{\min}(\tilde q(\epzero)) \geq \epzero/d$, 
irrespectively of what was $q$ we obtain 
\be
S_\alpha(\tilde q(\epzero)) \geq \frac{1}{1-\alpha} \ln \left(d \left(\frac{\epzero}{d}\right)^\alpha\right)
\ee
Now, since $p$ was assumed to be not full rank, we have 
\be
S_\alpha(p) \leq \ln (d-1)
\ee
Now,  \eqref{eq:epzero} assures that 
\be
\frac{1}{1-\alpha} \ln \left(d \biggl(\frac{\epzero}{d}\biggr)^\alpha\right) -\ln(d-1) \geq   \frac12 \ln\frac{d}{d-1}
\ee
for $\alpha\leq \alpha_0$. 
Using it, we then get
\be
S_\alpha(p) = S_\alpha(\tilde q(\epzero)- \left( S_\alpha(\tilde q(\epzero)  - S_\alpha(p)\right) \leq 
S_\alpha(\tilde q(\epzero) - \left( \frac{1}{1-\alpha} \ln \left(d \biggl(\frac{\epzero}{d}\biggr)^\alpha\right)  - \ln(d-1)\right)
\leq  S_\alpha(\tilde q(\epzero) - \frac12 \ln\frac{d}{d-1}
\ee
i.e.  $S_\alpha(p) \geq S_\alpha(\tilde q) + \frac12 \ln\frac{d}{d-1}$. 
\end{proof}


\begin{lemma}\label{lem:up bound on S and T epsilons}
For all probability distributions $q$, and $\ep\leq \ep'$, $\ep,\ep'\in(0,1)$ the Renyi and Tsalis entropies 
satisfy
\ba 
T_\alpha(\tilde q(\ep)) &\leq 	T_\alpha(\tilde q(\ep'))\label{eq:Tsalis monotonic q tilda}\\
S_\alpha(\tilde q(\ep)) &\leq  	S_\alpha(\tilde q(\ep'))\label{eq:Renyi monotonic q tilda}
\ea 
for all $\alpha\geq 0$, where $\tilde q(\ep)$ (introduced in Prop. \ref{prop:strict}), is given by
\be
\label{eq:tildeq 2}
\tilde q(\ep)=
\left\{
\begin{array}{ll}
\idv_d/d &\text{\rm if } \|q -\idv_d/d\|_1 < \ep \\
(1-\ep) q + \ep \idv_d/d &\text{\rm if } \|q -\idv_d/d\|_1 \geq \ep \\
\end{array}
\right.
\ee
\end{lemma}
\begin{proof}
We will start by proving Eq. \eqref{eq:Renyi monotonic q tilda} first. \\
The proof of Eq. \eqref{eq:Renyi monotonic q tilda} will be divided into two sub cases. We start with the easiest case.
\begin{itemize}
\item [Case 1:] $\|q -\idv_d/d\|_1 < \ep$\\
It follows that $S_\alpha(\tilde q(\ep))= S_\alpha(\idv_d/d)$ and also since $\|q -\idv_d/d\|_1 < \ep'$ we have  $S_\alpha(\tilde q(\ep'))= S_\alpha(\idv_d/d)$ and thus Eq. \eqref{eq:Renyi monotonic q tilda} holds for this case.
\item [Case 2:] $\|q -\idv_d/d\|_1 \geq  \ep$. \\
It follows that $S_\alpha(\tilde q(\ep))= S_\alpha\big((1-\ep)q+\ep\idv_d/d\big)$. We now have to further sub-divide into two possibilities.
\begin{itemize}
	\item [Case 2.1:] $\|q -\idv_d/d\|_1   <\ep'$ $\implies$  $S_\alpha(\tilde q(\ep'))= S_\alpha(\idv_d/d)$\\
	\item [Case 2.2:] $\|q -\idv_d/d\|_1   \geq \ep'$ $\implies$  $S_\alpha(\tilde q(\ep'))= S_\alpha\big((1-\ep')q+\ep'\idv_d/d\big)$.
\end{itemize}
Therefore, for Case 2 we have to prove that 
\ba 
S_\alpha\big((1-\ep)q+\ep\idv_d/d\big) &\leq S_\alpha\big((1-\ep')q+\ep'\idv_d/d\big) \label{eq:seond in proce xAW 2.1}\\
S_\alpha\big((1-\ep)q+\ep\idv_d/d\big) &\leq S_\alpha(\idv_d/d)\label{eq:seond in proce xAW 2.2}
\ea 
both hold under the quantifies stated in the Lemma.  We first observe that Eq. \eqref{eq:seond in proce xAW 2.1} implies Eq. \eqref{eq:seond in proce xAW 2.2} by setting $\ep'=1$. Thus we only need to prove Eq. \eqref{eq:seond in proce xAW 2.1}. For this, we first observe that 
\begin{equation}
(1-\ep')q+\ep'\idv_d/d = \gamma\big((1-\ep)q+\ep\idv_d/d\big)+ (1-\gamma)\idv_d/d,
\end{equation}
where $\gamma:=(1-\ep')/(1-\ep)\in[0,1]$. Hence the vector $(1-\ep')q+\ep'\idv_d/d$ is a mixture of $(1-\ep)q+\ep\idv_d/d$ with the uniform distribution $\idv_d/d$. As such $(1-\ep)q+\ep\idv_d/d$ majorises $(1-\ep')q+\ep'\idv_d/d$, and since the Renyi entropy is Schur concave for all $\alpha \geq 0$, Eq. \eqref{eq:seond in proce xAW 2.1} follows by Schur concavity. 
\end{itemize} 
We now need to prove Eq. \eqref{eq:Tsalis monotonic q tilda} to complete the proof of the lemma. From the definitions of the Renyi and Tsalis entropy, 
it follows
\begin{equation}
T_\alpha\left(S_\alpha\right)=\frac{\exp\left(\left(\frac{1-\alpha}{\alpha}\right) S_\alpha\right)-1}{1-\alpha},
\end{equation}
which is manifestly a non-decreasing function for all $\alpha>0$, thus $S_\alpha(p)\leq S_\alpha(p')$ iff $T_\alpha(p)\leq T_\alpha(p')$, and Eq. \eqref{eq:Tsalis monotonic q tilda} follows from Eq. \eqref{eq:Renyi monotonic q tilda}.
\end{proof}

\begin{lemma}
\label{lem:cont}
Let $p,p'\in\cP_d$
Denote $\ep=\|p-p'\|_1 $ 
. Then
\bei
\item For $\alpha\in (0,1/2]$ we have 
\be 
\label{eq:Thalf}
|T_\alpha(p) - T_\alpha(p')|\leq  6 d \left(\frac{\ep}{d}\right)^\alpha     
\ee 
\item for $\alpha\in [1/2,2]$ we have 
\be 
\label{eq:Thalf2}
|T_\alpha(p) - T_\alpha(p')|\leq -32 
d \sqrt{\frac{\ep}{d}} \ln \sqrt{\frac{\ep}{d}} \quad  \text{if } \ep\leq \frac1{32 d^2}
\ee 
\item for $\alpha\in [2,\infty)$ we have 
\be 
\label{eq:Tgeq2} 
|T_\alpha(p) - T_\alpha(p')|\leq 6 \sqrt{d\ep}
\ee 
\item We have 
\be 
\label{eq:Sinfty-cont}
|S_\infty(p)-S_\infty(p')|\leq d \ep
\ee 
\eei
\end{lemma}
\begin{proof} We will prove Eqs. \eqref{eq:Thalf} to \eqref{eq:Sinfty-cont} individual.\\

{\it Proof of \eqref{eq:Thalf}.} 
from \eqref{eq:Tsalis discontinous at 1 bound} we have  for $\alpha\in(0,1)$
\be 
|T_\alpha(p) - T_\alpha(p')|\leq  2\frac{ \lceil \alpha \rceil}{|\alpha-1|} d \left(\frac{ \ep}{d} \right)^{\alpha/\lceil \alpha \rceil}
\ee 
We then estimate for $\alpha\in(0,1/2]$
\be  \label{eq:cont leq 1}
2\frac{ \lceil \alpha \rceil}{|\alpha-1|} d \left(\frac{ \ep}{d} \right)^{\alpha/\lceil \alpha \rceil}\leq 
4 d \left(\frac{\ep}{d} \right)^{\alpha} 
\ee 
{\it Proof of \eqref{eq:Thalf2}.} 
From \eqref{eq:Tsalis uniform around 1 in theorem alpha 1 infty} we have for $\alpha>1$  and 
$\ep\leq \frac{1}{2 e \lceil \alpha 	\rceil d^\alpha }$
\be 
\left| T_\alpha(p)- T_\alpha(p') \right|    \leq  8 \left[ \ep\ln \left(\frac{d^{3/2}}{4}\right) - \ep \ln \ep\right]
\ee 
We then have 
\be 
8 \left[ \ep\ln \left(\frac{d^{3/2}}{4}\right) - \ep \ln \ep\right] \leq - 12 d 	 \frac{\ep}{d} \ln\frac{\ep}{d}
\ee 
so that for $\ep \leq \frac{1}{8 d^2}$ and $1\leq \alpha\leq 2 $
\begin{equation}
\label{eq:12}
\left| T_\alpha(p)- T_\alpha(p') \right| \leq  - 12 d 	 \frac{\ep}{d} \ln\frac{\ep}{d}
\end{equation}	
We now combine it with \eqref{eq:Tsalis uniform around 1 in theorem alpha 0 1} which  for $\alpha\in[1/2,1]$ implies that for $\ep \leq d \left(\frac{1}{2 \me d}\right)^{2}\leq \frac{1}{30 d}$,
\ba
\label{eq:half1}
\left| T_\alpha(p)- T_\alpha(p') \right| \leq &   4\, d \left[ \left(\frac{3}{2 \alpha}+1 \right) \left(\frac{\ep}{d}\right)^\alpha \ln d - \left(\frac{\ep}{d}\right)^\alpha \ln \ep \right] \\ 
\leq & - 32 d 	 \sqrt{\frac{\ep}{d}} \ln\sqrt{\frac{\ep}{d}}.
\ea
Taking worse of the two bounds we get  that for $\alpha\in [1/2,2]$ and $\ep\leq 1/(32 d^2)$ 
\be
\left| T_\alpha(p)- T_\alpha(p') \right| \leq  - 32 d 	 \sqrt{\frac{\ep}{d}} \ln\sqrt{\frac{\ep}{d}}.
\ee
{\it Proof of \eqref{eq:Tgeq2}.} 
From \eqref{eq:Tsalis discontinous at 1 bound} for all $\alpha\in (0,1)\cup (1,\infty)$ we have 
\be
|T_\alpha(p) - T_\alpha(p')|\leq  2\frac{ \lceil \alpha\rceil}{|\alpha-1|} d \left(\frac{ \ep}{d} \right)^{\alpha/\lceil \alpha \rceil}
\ee
Using  $\alpha\geq 2$ and $\ep/d\leq 1$  (always true, since $d\geq 2$ and $\ep \leq 2$) we have
\be
2\frac{  \lceil \alpha\rceil}{|\alpha-1|} d \left(\frac{ \ep}{d} \right)^{\alpha/\lceil \alpha \rceil}\leq 6 d 
\left(\frac{ \ep}{d} \right)^{\alpha/(\alpha+1)}\leq 6 d \sqrt{\frac{\ep}{d}} =6 \sqrt{\ep d}
\ee
{\it Proof of \eqref{eq:Sinfty-cont}.} 
\be
|\ln  p_{\max} - \ln p'_{\max}| =\left| \ln\left( \frac{|p_{\max}-p'_{\max}|}{\min\{p_{\max},p'_{\max}\}}+1\right) \right|
\leq \frac{|p_{\max}-p'_{\max}|}{\min\{p_{\max},p'_{\max}\}} \leq d \ep
\ee

\end{proof}

	\begin{lemma}
	\label{lem:alpha-min-max}
	Under the notation from the proof of Theorem \ref{thm:noemb}
	consider  the expression
	\begin{equation}\label{eq:ep res-lemma}
	 \max \left\{ \varepsilon_\textup{min}(\alphazero), \varepsilon_\textup{max}(\alphamax), \bar\ep_{Tmid}  \right\}
	\end{equation}
	where
	\ba 
	\varepsilon_\textup{min}(\alphazero)& = \max \left\{\bar\ep_{Tmin}(\alphazero), \bar\ep_0(\alphazero)\right\}\\
	\varepsilon_\textup{max}(\alphamax)&= \max \left\{  \bar\ep_{Tmax}(\alphamax), \bar\ep_\infty(\alphamax)  \right\}.
	\ea
	There exists $\alphazero\in(0,1)$ and $\alphamax\in(2,\infty)$ such that 
	\begin{align}
	    \max \left\{ \varepsilon_\textup{min}(\alphazero), \varepsilon_\textup{max}(\alphamax), \bar\ep_{Tmid}  \right\} \leq  \epres(\epemb,d_\Sy, D_\cat)
	\end{align}
	where $\epres(\epemb,d_\Sy, D_\cat)$
	is given by  Eq. \eqref{eq:long ep emb form}, i.e., 
 	\begin{align}
	    \epres(\epemb,d_\Sy, D_\cat):= 5\sqrt{\frac{d_\Sy^{5/3} + 4(\ln d_\Sy D_\cat) \ln d_\Sy}{\ln (1/\epemb)}+ d_\Sy D_\cat \epemb^{1/6} + 5\left( (d_\Sy D_\cat)^2 \sqrt{\frac{\epemb}{d_\Sy D_\cat}} \ln\sqrt{\frac{d_\Sy D_\cat}{\epemb}}  \right)^\frac23}.
	\end{align}
	\end{lemma}
	
	\begin{proof}
	We will start with what appears to be the most significant term, $\varepsilon_\textup{max}(\alphamax)$. Writing it explicitly, using Eqs. \eqref{eq:ep tilde up bound}, \eqref{eq:ep infinity up bound} we have
	\begin{equation}
	\varepsilon_\textup{max}(\alphamax)= \max \left\{  \big(96   \sqrt{D\epemb}  D^\alphamax\big)^\frac13,  4 \sqrt{\frac{\ln d_\Sy}{\alphamax}+D \epemb} \right\}.
	\end{equation}
	We now re-parametrizing $\alphamax$ in terms of a parameter $\beta_0>0$ via $\alphamax= -\beta_0 {(\ln \epemb)}/{(2\ln D)}$, to find  
	\begin{equation} \label{eq:reparametrisation}
	\max\left\{  \big(96   \sqrt{D\epemb}  D^\alphamax\big)^\frac13,  4 \sqrt{\frac{\ln d_\Sy}{\alphamax}+D \epemb} \right\}= \max\left\{  (96)^{1/3} D^{1/6} \epemb^{(1-\beta_0)/6}   ,  4 \sqrt{\frac{2(\ln D)(\ln d_\Sy)}{-\beta_0 \ln \epemb}+D \epemb} \right\}.
	\end{equation}
	With this parametrisation, we see that we need $(1-\beta_0)>0$ if the 1st term in the square brackets is to tend to zero as $\epemb$ goes to zero. Taking this and the requirement $\alphamax\geq 2$ into account we have 
	\begin{equation}\label{eq:beta 0 constraints}
	\frac{4 \ln D}{-\ln \epemb} \leq \beta_0 < 1.
	\end{equation}
	From Eqs. \eqref{eq:beta 0 constraints} and \eqref{eq:reparametrisation} we see that we need $\epemb$ to decay faster than any power of $D$. Specifically, it has to be of the form $\epemb(D)=\me^{-f(D) (\ln D)}$ where $\lim_{D\rightarrow+\infty} f(D)=+\infty$. Taking this into account, a reasonably good bound can be deduced by observing
	\begin{equation}
	(96)^{1/3} D^{1/6} \epemb^{(1-\beta_0)/6} \leq 5 \sqrt{D} \sqrt{\epemb^{(1-\beta_0)/3}}<5  \sqrt{\frac{2(\ln D)(\ln d_\Sy)}{-\beta_0 \ln \epemb} + D \epemb^{(1-\beta_0)/3} }
	\end{equation}
	and 
	\begin{equation}
	4 \sqrt{\frac{2(\ln D)(\ln d_\Sy)}{-\beta_0 \ln \epemb}+D \epemb} < 5  \sqrt{\frac{2(\ln D)(\ln d_\Sy)}{-\beta_0 \ln \epemb} + D \epemb^{(1-\beta_0)/3} }.
	\end{equation}
	Given the constraints which $\beta_0$ must satisfy, its exact choice is of little relevance. We therefore set it to $\beta_0=1/2$. Thus we have for the appropriate $\alphamax$, 
	\begin{equation}\label{eq:ep max with contraint}
	\varepsilon_\textup{max}(\alphamax) \leq 5  \sqrt{\frac{2(\ln D)(\ln d_\Sy)}{-\beta_0 \ln \epemb} + D \epemb^{(1-\beta_0)/3} }= 5  \sqrt{\frac{4(\ln D)(\ln d_\Sy)}{- \ln \epemb} + D \epemb^{1/6} } ,\quad \text{if }\;\;\; \frac{8 \ln D}{-\ln \epemb} \leq 1.
	\end{equation}
	
	We will now deal with the term $\varepsilon_\textup{min}(\alphazero)$ which plunging in Eqs. \eqref{eq:ep tilde up bound} and \eqref{eq:ep zero up bound},  reads
	\begin{equation}
	\varepsilon_\textup{min}(\alphazero)^3 =\max \left\{   96 D\, \frac{\epemb^\alphazero}{\alphazero}  ,  \left[ \left( \frac{d_\Sy-1}{d_\Sy}\right)^{3/2} \right]^{1/\alphazero}     \right\} \leq \frac{1}{\alphazero} \max \left\{   96 D\, \epemb^\alphazero  ,  \left[ \left( \frac{d_\Sy-1}{d_\Sy}\right)^{3/2} \right]^{1/\alphazero}     \right\}.
	\end{equation}
	We will now solve for $\alphazero$ the equation 
	\begin{equation}\label{eq:D epsilon alpha min}
	96 D\, \epemb^\alphazero =  \left[ \left( \frac{d_\Sy-1}{d_\Sy}\right)^{3/2} \right]^{1/\alphazero} ,
	\end{equation}
	which can be written as 
	\begin{equation}
	\alpha^2_\textup{min} \ln \left( \epemb \right) +\alphazero \ln \left(96 D \right) - \frac{3}{2}\ln\left( \frac{d_\Sy-1}{d_\Sy} \right) .
		\end{equation}
	Noting that $\alphazero\in(0,1)$ we take only the non-negative root, namely
	\begin{equation}\label{eq:alpha min}
	\alpha^{\large{*}}_\textup{min}=\frac{\ln \left( 96 D \right) +\sqrt{ \ln^2\left( 96 D \right)+ 6 \left( \ln\epemb \right) \left(\ln (1-1/d_\Sy) \right)} }{-2 \ln(\epemb)}.
	\end{equation} 
	Now note that the l.h.s. of Eq. \eqref{eq:D epsilon alpha min} is monotonically increasing with $\alphazero$ while the r.h.s. is monotonically decreasing with $\alphazero$. Therefore, since $\alphazero\in(0,1)$, we conclude 
	\begin{equation}\label{eq:ep min up bound}
	\varepsilon_\textup{min} (\alphazero)^3 \leq \frac{96 D \epemb^\gamma}{\gamma}, \quad \gamma =
	\begin{cases}
	\alpha^{\large{*}}_\textup{min} \quad &\text{if } \alpha^{\large{*}}_\textup{min}<1\\ 
	1 \quad &\text{otherwise.}
	\end{cases}
	\end{equation}
	
	Finally we will derive conditions for when $\alpha^{\large{*}}_\textup{min}<1$. To do so, we generalise the constraint in Eq. \ref{eq:ep max with contraint} to 
	\begin{equation}
	\frac{\beta_1 \ln D}{-\ln \epemb} \leq 1, \quad 8 \leq \beta_1 .
	\end{equation}
	Eq. \eqref{eq:alpha min} can now be upper bounded by
	\ba 
	\alpha^{\large{*}}_\textup{min} &= \frac{\ln(96 D)}{-2\ln \epemb} +\frac{1}{2}\sqrt{\left(\frac{\ln(96 D)}{\ln\epemb}\right)^2 + 6\; \frac{\ln(1-1/d_\Sy)}{\ln\epemb}} \leq \frac{\ln(96 D)}{2\beta_1 \ln D} +\frac{1}{2}\sqrt{\left(\frac{\ln(96 D)}{\beta_1 \ln D}\right)^2 - 6\; \frac{\ln(1-1/d_\Sy)}{\beta_1 \ln D}}\\
	&= \frac{\ln(96)}{2\beta_1 \ln (D_\cat d_\Sy)}+\frac{1}{2\beta_1} +\frac{1}{2}\sqrt{\left(\frac{\ln(96 )}{\beta_1 \ln (D_\cat d_\Sy)}+\frac{1}{\beta_1}\right)^2 - 6\; \frac{\ln(1-1/d_\Sy)}{\beta_1 \ln(D_\cat d_\Sy)}} \label{eq:alpa min lower bound}
	\ea 
	Thus recalling $D= D_\cat d_\Sy$ and noting that Eq. \eqref{eq:alpa min lower bound} is monotonically decreasing in $D$ and $d_\Sy$, we conclude that for all $D_\cat\geq D^{\large{*}}_\cat$ and $d_\Sy\geq d^{\large{*}}_\Sy$,
	\begin{equation}
	\alpha^{\large{*}}_\textup{min} \leq  \frac{\ln(96)}{2\beta_1 \ln (D^{\large{*}}_\cat d^{\large{*}}_\Sy)}+\frac{1}{2\beta_1} +\frac{1}{2}\sqrt{\left(\frac{\ln(96 )}{\beta_1 \ln (D^{\large{*}}_\cat d^{\large{*}}_\Sy)}+\frac{1}{\beta_1}\right)^2 - 6\; \frac{\ln(1-1/d^{\large{*}}_\Sy)}{\beta_1 \ln(D^{\large{*}}_\cat d^{\large{*}}_\Sy)}}.\label{eq:up bound alpah max final}
	\end{equation}
	Therefore, for $\beta_1=10$, $D^{\large{*}}_\cat=1$, $d^{\large{*}}_\Sy=2$; Eq. \eqref{eq:up bound alpah max final} gives $\alpha^{\large{*}}_\textup{min} \leq 0.921\ldots$. For larger values of $D^{\large{*}}_\cat$, $d^{\large{*}}_\Sy$, we can use $\beta_1=8$ and still achieve $\alpha^{\large{*}}_\textup{min}\leq 1$.
	We will thus assume 
	\begin{equation}\label{eq:containt value 10}
	\frac{10 \ln D}{-\ln \epemb} \leq 1
		\end{equation}
	in the rest of this proof. We can now write Eq. \eqref{eq:ep min up bound} in the form 
	\ba 
	\varepsilon_\textup{min} (\alphazero)^3 &\leq \frac{96 D \epemb^{\alpha^{\large{*}}_\textup{min}}}{\alpha^{\large{*}}_\textup{min}} = \frac{96 D (-2) \ln(\epemb)}{\ln \left( 96 D \right) +\sqrt{ \ln^2\left( 96 D \right)+ 6 \left( \ln\epemb \right) \left(\ln (1-1/d_\Sy) \right)} } \;\me^{-\alpha^{\large{*}}_\textup{min} \ln \epemb }\\
	& =\frac{96 D (-2) \ln(\epemb)}{\ln \left( 96 D \right) +\sqrt{ \ln^2\left( 96 D \right)+ 6 \left( \ln\epemb \right) \left(\ln (1-1/d_\Sy) \right)} } \sqrt{\me^{-\ln \left( 96 D \right) -\sqrt{ \ln^2\left( 96 D \right)+ 6 \left( \ln\epemb \right) \left(\ln (1-1/d_\Sy) \right)} }}.\label{eq:ep min up bound 2}
	\ea 
	We now observe that in the large $D$ limit, if $\ln^2\left( 96 D \right)+ 6 \left( \ln\epemb \right) \left(\ln (1-1/d_\Sy) \right) \approx \ln^2\left( 96 D \right)$ then the upper bound Eq. \eqref{eq:ep min up bound 2} is approximately proportional to $(-\ln \epemb)/ \ln D$. This quantity is necessarily greater than 10 due to constraint Eq. \eqref{eq:containt value 10}, and thus cannot be arbitrarily small. We will thus demand
	\begin{equation}
	\left( \ln\epemb \right) (\ln (1-1/d_\Sy) \geq \ln^2 ( D), \label{eq:2nd constaint}
	\end{equation}
	 in order to have a non-trivial bound.\footnote{This choice is so that we can use lemma \ref{lem:easy inequality}, but one could make other choices if one made a different version of the bound.} Thus using Lemma \ref{lem:easy inequality}, if follows from Eq. \eqref{eq:ep min up bound 2},
	\ba 
	\varepsilon_\textup{min} (\alphazero)^3 &\leq 
	\frac{96 D (-2) \ln(\epemb)}{\sqrt{ 6 \left( \ln\epemb \right) \left(\ln (1-1/d_\Sy) \right)} } \sqrt{\me^{-\ln \left( 96 D \right)}\, \me^{ -\sqrt{ \ln^2\left( 96 D \right)+ 6 \left( \ln\epemb \right) \left(\ln (1-1/d_\Sy) \right)} }}\\
	&\leq 
	\frac{96 D (-2) \ln(\epemb)}{\sqrt{ 6 \left( \ln\epemb \right) \left(\ln (1-1/d_\Sy) \right)} } \sqrt{\me^{-\ln \left( 96 D \right)}\, \frac{5^6 \,\me^{ -\sqrt{ \ln^2\left( 96 D \right)} }   }{\left( (\ln \epemb)(\ln(1-1/d_\Sy))  \right)^4} }\\
	&=2\frac{5^3}{\sqrt{6}} \frac{1}{(-\ln (1-1/d_\Sy))^{5/2} (-\ln\epemb)^{3/2} }.\\
	\ea
	Furthermore, we can use the standard inequality $\ln(1-x)\leq -x$ for all $x<1$ (which is sharp for small $x$.) with the identification $x=1/d_\Sy$, to achieve
	\begin{equation}
	\varepsilon_\textup{min} (\alphazero)^3 \leq  2\frac{5^3}{\sqrt{6}} \frac{d_\Sy^{5/2}}{(-\ln\epemb)^{3/2} }.
	\end{equation}
	Thus
	\ba
	\varepsilon_\textup{min} (\alphazero) &\leq 5\frac{ d_\Sy^{5/6}}{ \sqrt{-\ln\epemb} }
	.\label{eq:ep min up bound 3}
	\ea
	Therefore, by comparing Eqs.\eqref{eq:ep max with contraint} and \eqref{eq:ep min up bound 3}, we see that
	\begin{equation}
	\max\left\{ \varepsilon_\textup{min} (\alphazero), \varepsilon_\textup{max} (\alphamax) \right\} \leq 5\sqrt{\frac{f(D_\cat,d_\Sy)}{-\ln\epemb}+ D \epemb^{1/6}},
	\end{equation}
	where
	\begin{equation}
	f(D_\cat, d_\Sy):=\max\left\{ 4\ln(D_\cat d_\Sy) \ln d_\Sy, \, d_\Sy^{5/3}   \right\}=
	\begin{cases}
	4\ln(D_\cat d_\Sy) \ln d_\Sy &\text{if } D_\cat \geq \frac{1}{d_\Sy} \exp\left[\frac{d_\Sy^{5/3} }{4 \ln d_\Sy}\right]\\
	d_\Sy^{5/3} & \text{otherwise}
	\end{cases}
	\end{equation}
	as long as constraints Eqs. \eqref{eq:containt value 10}, \eqref{eq:2nd constaint} are both satisfied. Taking into account the expression for $\bar \ep_{Tmid}$ in Eq. \eqref{eq:ep tilde up bound}, have the bound,
	\ba
	\max \left\{ \varepsilon_\textup{min}(\alphazero), \varepsilon_\textup{max}(\alphamax), \bar\ep_{Tmid}  \right\} &\leq  5\sqrt{\frac{f(D_\cat,d_\Sy)}{-\ln\epemb}+ D \epemb^{1/6} + 5^{-2}\Big(-1024\, D^2 \sqrt{\frac{\epemb}{D}} \ln\sqrt{\frac{\epemb}{D}}  \Big)^\frac23}\\
	&< 5\sqrt{\frac{f(D_\cat,d_\Sy)}{-\ln\epemb}+ D \epemb^{1/6} + 5\left( D^2 \sqrt{\frac{\epemb}{D}} \ln\sqrt{\frac{D}{\epemb}}  \right)^\frac23}\\
	&< 5\sqrt{\frac{d_\Sy^{5/3} + 4(\ln D) \ln d_\Sy}{-\ln\epemb}+ D \epemb^{1/6} + 5\left( D^2 \sqrt{\frac{\epemb}{D}} \ln\sqrt{\frac{D}{\epemb}}  \right)^\frac23},\label{eq:last line}
	\ea 
	if constraints Eqs. \eqref{eq:containt value 10}, \eqref{eq:2nd constaint} are satisfied. Thus if we set $\epres$ (defined via Eq. \ref{eq:ep res}) to
	\begin{equation}\label{eq:old ep res bound}
	\epres =5\sqrt{\frac{d_\Sy^{5/3} + 4(\ln D) \ln d_\Sy}{-\ln\epemb}+ D \epemb^{1/6} + 5\left( D^2 \sqrt{\frac{\epemb}{D}} \ln\sqrt{\frac{D}{\epemb}}  \right)^\frac23},
	\end{equation}
	we conclude the proof.
\end{proof}

\subsubsection{Auxiliary lemmas}
\begin{lemma}\label{lem:easy inequality}
Let $x \geq  \ln^2 (D)$, then for all $D\geq 2$,
\begin{equation}
\me^{-\sqrt{\ln^2(96 D) +6x}} \leq 5^6\, \frac{\me^{-\sqrt{\ln^2(96 D)}}}{x^4}.
\end{equation}
\end{lemma}
\begin{proof}
We have that for all $D\geq 2$, 
\ba 
\frac{x^4}{5^6}\, \me^{-\sqrt{\ln^2(96 D)+6x} +\ln(96)+\ln(D)}  &\leq  \frac{x^4}{5^6}\, \me^{-\sqrt{\ln^2(96 \cdot 2)+6x} +\ln(96)+\sqrt{x}}  \leq \frac{96}{5^6}x^4\, \me^{-\sqrt{27+6x}+\sqrt{x}}\\
& =\frac{96}{5^6} F(x).\label{eq:F(x)}
\ea 
We now aim to find the maximum of $F(x):=x^4\, \me^{-\sqrt{27+6x}+\sqrt{x}}$ with domain $x\in[0,\infty).$ Since the extremal points are both zero, namely $F(0)=\lim_{x\rightarrow \infty}F(x)=0$, the maximum will be one of the stationary points, we therefore want the solutions to 
\begin{equation}
\frac{d}{dx} F(x)= \frac{e^{\sqrt{x}-\sqrt{6 x+27}} x^3 \left(-2 \sqrt{3} x+\sqrt{x} \sqrt{2 x+9}+8 \sqrt{2 x+9}\right)}{2 \sqrt{2 x+9}}=0.
\end{equation}
The only solution to $-2 \sqrt{3} x+\sqrt{x} \sqrt{2 x+9}+8 \sqrt{2 x+9}=0$ can be found analytically by hand (or using Mathematica's \textit{Solve} routine) giving 
\ba
\begin{split}
&x_0:=\\
&\frac{1}{2} \Bigg(-\frac{\sqrt[3]{1719926784 \sqrt{1149814}+17320375304957}}{60\ 5^{2/3}}-\frac{1}{300} \sqrt[3]{86601876524785-8599633920 \sqrt{1149814}}+ \frac{354075648}{625}\times\\
&\times \sqrt{\frac{3}{25 \sqrt[3]{5 \left(1719926784 \sqrt{1149814}+17320375304957\right)}+25 \sqrt[3]{86601876524785-8599633920 \sqrt{1149814}}+2754793}} \\
&+\frac{2754793}{3750}\Bigg)^{1/2}+\frac{941}{100}+\frac{1}{100 \sqrt{\frac{3}{25 \sqrt[3]{5 \left(1719926784 \sqrt{1149814}+17320375304957\right)}+25 \sqrt[3]{86601876524785-8599633920 \sqrt{1149814}}+2754793}}},
\end{split}
\ea 
since $x_0$ is within the end points, namely $0<x_0<\infty$ and $F$ evaluated at $x_0$ is larger than at the end points, i.e. $F(x_0)>F(0)=0$, and $F(x_0)>\lim_{x\rightarrow \infty}F(x)=0$; it must be a global maximum. Thus to conclude the proof, we use Eq. \eqref{eq:F(x)} to find
\begin{equation}
\frac{x^4}{5^6}\, \me^{-\sqrt{\ln^2(96 D)+6x} +\ln(96)+\ln(D)} < \frac{96}{5^6} F(x_0)=0.707818 \ldots  <1,
\end{equation}
for all $D\geq 2$ and for all $x \geq  \ln^2 (D)$.
\end{proof}

\begin{conj}
\label{conj:Hconvex}
$\hcal_{\alpha}(p|q)$ is convex in $\alpha$ for $\alpha <0$ and $\alpha>1$ and concave in $\alpha$ for $\alpha\in(0,1)$.
\end{conj}
{\it Remark.} From the plot it follows at least for $\hcal_\alpha(p |I/2)$ for binary distributions, which is enough for us.
We have not proven the conjecture, but we are able to prove the following 
\begin{lemma}
\label{lem:HvsD}
For arbitrary binary probability distribution  $p\in\cP_2.$ 
\bei
\item   For $\alpha\in(0,1)$
\be
\label{eq:H2x}
\hcal_\alpha(p|I/2) \geq \alpha D(p|I/2) \geq  \frac{\alpha}{8}\delta^2
\ee
\item For $\alpha>1$
\be
\label{eq:H3x}
\hcal_\alpha(p|I/2) \geq D(p|I/2). 
\ee
\begin{proof}
The inequality \eqref{eq:H2x} comes from the lemma \ref{lem:Hbound01} below and Pinsker inequality \eqref{eq:pinsker}. The inequality \eqref{eq:H3x}
comes from  \eqref{eq:HD} and Pinsker inequality. 
\end{proof}

\eei
\end{lemma}

\begin{lemma}
\label{lem:Hbound01}
For any probability distributions $p,q\in\cP_d$ and for all $\alpha\in[0,1)$ and  $\delta\in(0,1)$ we have 
\be
H_\alpha(p|I/2) \geq \alpha D(p|I/2)
\ee
Equivalently for all $\alpha\in(0,1)$  and $\delta<1$ we have 
\be
\frac{(1+\delta)^\alpha+(1-\delta)^\alpha-2}{\alpha-1}\geq \alpha ((1+\delta)\ln(1+\delta)+(1-\delta)\ln(1-\delta) )
\ee
\end{lemma}
\begin{proof}
We will prove that 
\be
G(\alpha,\delta):=(1+\delta)^\alpha+(1-\delta)^\alpha-2   - \alpha(\alpha-1) \bigl((1+\delta)\ln(1+\delta)+(1-\delta)\ln(1-\delta) 
\bigr) \geq 0.
\ee
For $|x|<1$ we have 
\be
(1+x)^\alpha=1 +\alpha x  + \sum_{n=2}^\infty\frac{a_n(\alpha)}{n!} x^n, \quad \ln(1+x)=\sum_{n=1}^\infty\frac{(-1)^{n+1}}{n}x^n
\ee
where 
\be
a_n(\alpha)=\alpha (\alpha-1) \ldots (\alpha-(n-1)).
\ee
One then finds that 
\be
\label{eq:series}
(1+\delta)^\alpha + (1-\delta)^\alpha = 2+2 \sum_{n=2,\, even}^\infty \frac{a_n(\alpha)}{n!} \delta^n,\quad 
(1+\delta)\ln(1+\delta)+(1-\delta)\ln(1-\delta)  =   2\sum_{n=2,\,even}^\infty\frac{1}{n(n-1)}\delta^n
\ee

Then 
\be
G(\alpha,\delta)=2 \sum_{n=2,\,even}^\infty
\left(\frac{a_n(\alpha)}{n!} - \alpha(\alpha-1) \frac{1}{n(n-1)}\right)\delta^n.
\ee
Now, it is enough to  show that 
\be
(\alpha-2) \ldots (\alpha-(n-1)) \geq (n-2)!,
\ee
for all $n\geq 2$ and even. Since there is even number of terms on the left hand side, all negative, the above inequality can be rewritten as 
\begin{equation}
(2-\alpha)(3-\alpha)\ldots (n-1-\alpha) \geq (n-2)!
\end{equation}
Since $\alpha\in [0,1)$, left hand side  is no smaller than $(n-2)!$ and therefore the inequality holds.

\end{proof}

\subsection{Tsalis continuity Theorem} 	

\begin{theorem}[Tsalis uniform continuity]\label{thm:Tsalis continuity} Let $p,p'\in\cP_d$ have entries denoted $p_k$. For the following parameters, we have the following Tsalis entropy (Eq. \ref{eq:Tsalis def}) continuity bounds:\\
	\begin{itemize}
		\item [0)]
		For $\alpha\in (0,1)\cup(1,\infty]$, 
		\be\label{eq:Tsalis discontinous at 1 bound}
		\left| T_\alpha(p)- T_\alpha(p') \right| \leq \frac{2 \lceil \alpha \rceil}{|\alpha-1|} d^{1-\alpha/\lceil \alpha \rceil} \left( \|p-p'\|_1 \right)^{\alpha/\lceil \alpha \rceil} \leq 2\frac{ ( \alpha +1)}{|\alpha-1|} d \left(\frac{ \|p-p'\|_1}{d} \right)^{\alpha/(1+\alpha)}.
		\ee 
		\item [1)]
		For $\alpha\in (0,1]$ and $\|p-p'\|_1 \leq d\,\big(\frac{1 }{2\me \,d}\big)^{1/\alpha}$,
		\be\label{eq:Tsalis uniform around 1 in theorem alpha 0 1}
		\left| T_\alpha(p)- T_\alpha(p') \right|  \leq  4\, d^{1-\alpha} \left[ \left(\frac{3}{2 \alpha}+1 \right) \big(\|p-p'\|_1\big)^\alpha \ln d - \big(\|p-p'\|_1\big)^\alpha \ln \|p-p'\|_1 \right] .
		\ee
		\item [2)]
		For $\alpha\in [1,\infty)$ and $\|p-p'\|_1 \leq \frac{1}{2\me\lceil \alpha\rceil d^\alpha}$,
		\be\label{eq:Tsalis uniform around 1 in theorem alpha 1 infty}
		\left| T_\alpha(p)- T_\alpha(p') \right|    \leq  8 \left[ \|p-p'\|_1 \ln \left(\frac{d^{3/2}}{4}\right) - \|p-p'\|_1 \ln \big( \|p-p'\|_1 \big)\right].
		\ee
	\end{itemize}
\end{theorem}
\begin{remark}
	Eqs. \ref{eq:Tsalis discontinous at 1 bound}, \eqref{eq:Tsalis uniform around 1 in theorem alpha 0 1}, \eqref{eq:Tsalis uniform around 1 in theorem alpha 1 infty}, provide continuity bounds for the Tsalis entropy, which are uniform in $\alpha>0$ bounded away from zero. For $\alpha$ near zero Eq. \eqref{eq:Tsalis discontinous at 1 bound} is best, while for $\alpha$ in the vicinity of 1, Eqs.  \eqref{eq:Tsalis uniform around 1 in theorem alpha 0 1}, \eqref{eq:Tsalis uniform around 1 in theorem alpha 1 infty} are optimal. For large $\alpha$ one can either use Eq. \ref{eq:Tsalis discontinous at 1 bound} or Eq. \eqref{eq:Tsalis uniform around 1 in theorem alpha 1 infty} depending on the circumstances. If the condition $\|p-p'\|_1 \leq \frac{1}{2\me\lceil \alpha\rceil d^\alpha}$ is fulfilled (which becomes more stringent the larger $\alpha$ is), then it is likely to be preferable to use Eq. \eqref{eq:Tsalis uniform around 1 in theorem alpha 1 infty} which only grows logarithmically with $d$. On the other hand, if $\|p-p'\|_1 \leq \frac{1}{2\me\lceil \alpha\rceil d^\alpha}$ cannot be guaranteed to be fulfilled (such as in the limiting case $\alpha\rightarrow \infty$), then Eq. \ref{eq:Tsalis discontinous at 1 bound}, with sub-linear scaling with $d$, is the only option. For related, but less explicit, continuity bounds see \cite{hanson2017tight}.
\end{remark}
\begin{proof}
	We start by proving  Eq. \eqref{eq:Tsalis discontinous at 1 bound}. From the definition of the Tsalis entropy, Eq. \eqref{eq:Tsalis def}, it follows
	\be 
	\left| T_\alpha(p)- T_\alpha(p') \right| = \frac{1}{|1-\alpha|} \bigg| \sum_i p_i^\alpha - p_i^{\prime \alpha} \bigg|\leq  \frac{1}{|1-\alpha|}  \sum_i \left| p_i^\alpha - p_i^{\prime \alpha} \right|.
	\ee 
	We now apply Lemma \ref{Lemm:sum diff up bound} to find
	\be
	\left| T_\alpha(p)- T_\alpha(p') \right| \leq   \frac{  2  \lceil \alpha \rceil }{|1-\alpha|} \, d ^{(1-\alpha/ \lceil \alpha \rceil)} \left( \|p-p'\|_1 \right)^{\alpha/ \lceil \alpha \rceil},
	\ee 	
	from which the bound follows by noting $\lceil \alpha\rceil\leq \alpha+1$, and $\|p-p'\| /d \leq 1$.
	We will now prove Eqs. \eqref{eq:Tsalis uniform around 1 in theorem alpha 0 1}, \eqref{eq:Tsalis uniform around 1 in theorem alpha 1 infty}. To start with, 
	\be\label{eq:Tsalis cont}
	\left| T_\alpha(p)- T_\alpha(p') \right|= \frac{1}{|1-\alpha|} \bigg| \sum_i p_i^\alpha - p_i^{\prime \alpha} \bigg| \leq \frac{1}{|\alpha-1|} \left|G_\alpha (p,p')  \right|,
	\ee
	where we have defined 
	\be 
	G_\alpha (p,p')= \|p\|_\alpha^\alpha - \|p'\|_\alpha^\alpha.
	\ee 
	From the definition of $\|\cdot \|_p$ we see that for all $\alpha\in[0,\infty)$, $G_\alpha(p,p')$ is continuous. Furthermore, from Eq. \eqref{eq:div sudo p norm} we observe that $G_\alpha(p,p')$ is differentiable for $\alpha\in(0,\infty)$. As such, we can apply the mean value theorem to it as follows.
	
	Using the notation $a,b,c,$ from the mean value theorem \ref{lemm:MVT}, we have
	\begin{itemize}
		\item [1)] For $a=1$, $b=\alpha$, \,\,$\alpha> 1$ 
		\be \label{eq:g alpha p p 1 0}
		G_\alpha(p,p')=G_1(p,p') +G_c'(p,p') (\alpha-1)=G_c'(p,p') (\alpha-1), \quad \text{for some } c\in(1,\alpha).
		\ee 
		\item[2)] For $b=1$, $a=\alpha$, \,\,$0\leq \alpha< 1$
		\be\label{eq:g alpha p p 2 0}
		G_\alpha(p,p')=G_1(p,p') +G_c'(p,p') (-1+\alpha)= G_c'(p,p') (\alpha-1), \quad \text{for some } c\in(\alpha,1).
		\ee 
	\end{itemize}
	Where in both cases we have used $\|p\|_1=\|p'\|_1=1$.
	We thus conclude
	\be 
	g_\alpha(p,p')= g_c'(p,p')(\alpha-1),\quad \text{for some } c\in \begin{cases}
		(\alpha,1) &\text{ if } 0\leq \alpha < 1\\
		(1,\alpha) &\text{ if } \alpha  > 1.
	\end{cases}
	\ee 
	Plugging in to Eq. \eqref{eq:Tsalis cont}, we thus have for all $\alpha > 0$, 
	\be\label{eq:Tsalis cont 2}
	\left| T_\alpha(p)- T_\alpha(p') \right| \leq  \left| \frac{d}{d\alpha} \|p\|_\alpha^\alpha- \frac{d}{d\alpha} \|p'\|_\alpha^\alpha  \right|_{\alpha=c}= c \left| \|p\|_\alpha^{\alpha-1}\frac{d}{d\alpha} \|p\|_\alpha- \|p'\|_\alpha^{\alpha-1}\frac{d}{d\alpha} \|p'\|_\alpha  \right|_{\alpha=c}.
	\ee
	Now plugging in Eq. \ref{eq:div p-norm for x},
	\be 
	\left| T_\alpha(p)- T_\alpha(p') \right| \leq \frac{1}{c} \Big| \|p\|_\alpha^\alpha\, S_1\left(q_\alpha(p)\right) -\|p'\|_\alpha^\alpha\, S_1\left(q_\alpha(p')\right)\Big|_{\alpha=c},
	\ee
	where
	\be 
	[q_\alpha(x)]_i := \frac{|x_i|^\alpha}{\|x\|_\alpha^\alpha}, \quad i=1,2,3,\ldots ,d.
	\ee
	Thus we find, 
	\ba 
	\left| T_\alpha(p)- T_\alpha(p') \right| & \leq\frac{1}{\ccc} \|p\|_\ccc^\ccc \bigg|\,\left(\frac{\|p'\|_\ccc^\ccc}{\|p\|_\ccc^\ccc}-1\right)  S_1\left( q_\ccc(p') \right)+S_1\left( q_\ccc(p')\right) - S_1\left( q_\ccc(p) \right)\bigg| \\
	&\leq  \frac{1}{\ccc} \|p\|_\ccc^\ccc \bigg(\,\left|\frac{\|p'\|_\ccc^\ccc}{\|p\|_\ccc^\ccc}-1\right|  \left|S_1\left( q_\ccc(p') \right)\right|+\left|S_1\left( q_\ccc(p') \right) -  S_1\left( q_\ccc(p) \right)\right|\bigg)\\
	& = \frac{1}{\ccc} \bigg( \left|S_1\left( q_\ccc(p') \right)\right| \big|\|p'\|_\ccc^\ccc-\|p\|_\ccc^\ccc\big|  + \|p\|_\ccc^\ccc\big|S_1\left( q_\ccc(p') \right) -  S_1\left( q_\ccc(p) \right)\big| \bigg) \\
	& \leq  \frac{1}{\ccc} \bigg(  \left(  \max_{q\in\cP_d}\left|S_1(q)\right|\right) \Big|\|p'\|_\ccc^\ccc-\|p\|_\ccc^\ccc\Big|  + \|p\|_\ccc^\ccc\Big| S_1\left(q_\ccc(p')\right) -  S_1\left(q_\ccc(p)\right)\Big|  \bigg)  \\
	&= \frac{1}{\ccc} \bigg(  \ln d\, \Big|\|p'\|_\ccc^\ccc-\|p\|_\ccc^\ccc\Big|  + \|p\|_\ccc^\ccc\Big| S_1\left(q_\ccc(p')\right) -  S_1\left(q_\ccc(p)\right)\Big| \bigg).\label{eq:T -T inter 3}
	\ea
	Applying the Fannes inequality (Lemma \ref{Lemm:Fannes}), we find 
	\ba 
	\left| T_\alpha(p)- T_\alpha(p') \right| & \leq \frac{1}{\ccc} \bigg(  \Big|\|p'\|_\ccc^\ccc-\|p\|_\ccc^\ccc\Big|\ln d\  + \|p\|_\ccc^\ccc\Big(  \| q_\ccc(p)- q_\ccc(p')\|_1 \ln d  - \| q_\ccc(p)- q_\ccc(p')\|_1 \ln\left( \|q_\ccc(p)- q_\ccc(p')\|_1 \right)  \Big) \bigg)\\
	\ea
	We now pause a moment to bound $\|q_\alpha(p)- q_\alpha(p')\|_1 $. Using the definition of $q_\alpha(p)$, we have
	\ba
	\|q_\alpha(p)- q_\alpha(p')\|_1 & =\sum_{i=1}^d \left| \frac{p_i^\alpha}{\|p\|_\alpha^\alpha}+ \frac{p_i^{\prime\alpha}}{\|p'\|_\alpha^\alpha} \right| = \sum_{i=1}^d \frac{1}{\|p\|_\alpha^\alpha} \Bigg| p_i^\alpha -p_i^{\prime \alpha} + p_i^{\prime \alpha}\left( 1-\frac{\|p\|_\alpha^\alpha}{\|p'\|_\alpha^\alpha}  \right) \Bigg|\\
	& \leq   \sum_{i=1}^d \frac{1}{\|p\|_\alpha^\alpha} \Bigg( \left|p_i^\alpha -p_i^{\prime \alpha}\right|  + p_i^{\prime \alpha}\left| 1-\frac{\|p\|_\alpha^\alpha}{\|p'\|_\alpha^\alpha}  \right| \Bigg)  = \frac{1}{\|p\|_\alpha^\alpha} \Bigg(  \Big| \|p\|_\alpha^\alpha -  \|p'\|_\alpha^\alpha  \Big| +   \sum_{i=1}^d \left|p_i^\alpha -p_i^{\prime \alpha}\right|\Bigg)\\
	&\leq  \frac{2}{\|p\|_\alpha^\alpha} \Bigg(    \sum_{i=1}^d \left|p_i^\alpha -p_i^{\prime \alpha}\right|\Bigg) = \frac{\Delta_\alpha (p,p')}{\|p\|_\alpha^\alpha} ,\label{eq:q-q Tsalis}
	\ea 
	where in the last line, we have used Lemma \ref{Lemm:sum diff up bound} and defined,
	\be \label{eq: F beta up bound 2 0}
	\Delta_\alpha(p,p'):=2\sum_{i=1}^d \big| p_i^\alpha-p_i^{\prime \alpha}  \big|.
	\ee 
	Plugging this into Eq. \eqref{eq:T -T inter 3}, we find for $\| q_\ccc(p)- q_\ccc(p')\|_1 \leq 1/\me \approx 0.37$, 
	\ba 
	\left| T_\alpha(p)- T_\alpha(p') \right| & \leq \frac{1}{\ccc} \bigg(  \Big|\|p'\|_\ccc^\ccc-\|p\|_\ccc^\ccc\Big|\ln d\  + \|p\|_\ccc^\ccc\Big(  \frac{\Delta_\ccc (p,p')}{\|p\|_\ccc^\ccc} \ln d  - \frac{\Delta_\ccc (p,p')}{\|p\|_\ccc^\ccc} \ln\left( \frac{\Delta_\ccc (p,p')}{\|p\|_\ccc^\ccc} \right)  \Big) \bigg)\\
	& \leq \frac{1}{\ccc} \bigg(  \frac{\Delta_\ccc(p,p')}{2}\ln d\  +  \Delta_\ccc (p,p') \ln d  - \Delta_\ccc (p,p') \ln\left( \frac{\Delta_\ccc (p,p')}{\|p\|_\ccc^\ccc} \right)   \bigg),\\
	& = \frac{1}{\ccc} \left( \Delta_\ccc (p,p')\ln\left(d^{3/2} \|p\|_\ccc^\ccc  \right) -\Delta_\ccc (p,p')\ln \Delta_\ccc (p,p') \right).\label{eq:T-T  30}
	\ea
	We now find bounds for $\Delta_\ccc (p,p')$. To start with, from Lemma \ref{Lemm:sum diff up bound} we have
	\be 
	\Delta_\ccc (p,p') \leq 4  \lceil \ccc \rceil \, d  \left(\frac{ \|p-p'\|_1}{d} \right)^{\ccc/ \lceil \ccc \rceil},\ee
	where $\|p-p'\|_1/d \leq 1$ since $d\geq 2$ and $\|p-p'\|_1\leq 2$ holds for all $p,p'\in\cP_d$. \\
	For $\alpha\in[0,1)$, $c\in(\alpha,1)$ we find 
	\be \label{eq:Delta for alpha 0 1}
	\Delta_\ccc (p,p') \leq 4  d  \left(\frac{ \|p-p'\|_1}{d} \right)^{\ccc} \leq 4  d  \left(\frac{ \|p-p'\|_1}{d} \right)^{\alpha} \quad \forall\, \ccc \in(\alpha,1).
	\ee 
	For $\alpha> 1$, $c\in(1, \alpha)$ and we find 
	\be \label{eq:Delta for alpha 1 infty}
	\Delta_\ccc (p,p')  \leq 4  \lceil \ccc \rceil \, d  \left(\frac{ \|p-p'\|_1}{d} \right)= 4  \lceil \ccc \rceil \,  \|p-p'\|_1 \quad \forall\, \ccc \in(1,\alpha).
	\ee 
	We will now upper bound $\|p\|_\ccc^\ccc$ using Eq. \eqref{eq:p r inequalities} from Lemma \ref{norm lem trangle inequality}.\\
	For $\alpha\in(0,1)$, $c\in(\alpha,1)$\,:
	\begin{itemize}
		\item [1)] Noting that $0\leq p_i\leq 1$ for all $i$, 
		\be \label{eq: norm c  for alpha 0 1 eq1}
		\|p\|_\ccc^\ccc =\sum_{i=1}^d p_i^\ccc \leq \lim_{\ccc \rightarrow 0^+} \sum_{i=1}^d p_i^\ccc \leq d\quad \forall \, c\in(\alpha,1)
		\ee  
		\item[2)] Setting $r=\ccc$, $p=1$, in Eq. \eqref{eq:p r inequalities}, we have
		\be \label{eq:alpha 0 1 for constriant}
		1=\|p\|_1 \leq \|p\|_\ccc \implies  1 \leq \|p\|_\ccc^\ccc
		\ee
	\end{itemize} 
	For $\alpha\in(1,\infty)$, $c\in(1,\alpha)$\,:
	\begin{itemize}
		\item [1)]   Setting $r=1$, $p=\ccc$, in Eq. \eqref{eq:p r inequalities}, we have 
		\be \label{eq:norm case 1 alpah 1 infinity}
		\|p\|_\ccc \leq \|p\|_r=1 \quad\forall\, c\in(1,\alpha)  \implies \ln \|p\|_\ccc^\ccc = \ccc \ln \|p\|_\ccc \leq 0.
		\ee
		\item [2)]   Seetting $r=c$, $p=\infty$, in Eq. \eqref{eq:p r inequalities}, and noting that $\max_i\{p_i\}\geq 1/d$ we have 
		\be \label{eq:alpha 1 infinity for constriant} 
		\|p\|_\infty = \max_{i}\{p_i\} \leq \|p\|_c \quad\forall\, c\in(1,\alpha)  \implies  \left(\frac{1}{d}\right)^\ccc \leq \|p\|_\ccc^\ccc
		\ee
		
	\end{itemize}
	
	Plugging Eqs. \ref{eq:Delta for alpha 0 1}, \ref{eq: norm c  for alpha 0 1 eq1} into Eq. \eqref{eq:T-T  30}, we find for $\alpha\in(0,1)$ and for all $\ccc\in(\alpha,1)$,
	\ba
	\left| T_\alpha(p)- T_\alpha(p') \right|  & \leq \frac{1}{\ccc} \left[ 4 d \left(\frac{\|p-p'\|_1}{d}\right)^\alpha \ln\left(d^{3/2} d\right) -4 d \left(\frac{\|p-p'\|_1}{d}\right)^\alpha \ln\left( 4d \left(\frac{\|p-p'\|_1}{d}\right)^\alpha \right) \right]\\
	&\leq  4 \frac{d}{\alpha} \left(\frac{\|p-p'\|_1}{d}\right)^\alpha \ln\left(\frac{d^{3/2}}{4} \right) -4 d \left(\frac{\|p-p'\|_1}{d}\right)^\alpha \ln \left(\frac{\|p-p'\|_1}{d}\right)\\
	& \leq  4 \frac{d}{\alpha} \left(\frac{\|p-p'\|_1}{d}\right)^\alpha \ln\left(d^{3/2} \right) -4 d \left(\frac{\|p-p'\|_1}{d}\right)^\alpha \ln \left(\frac{\|p-p'\|_1}{d}\right)\\
	&= 4\, d^{1-\alpha} \left[ \left(\frac{3}{2 \alpha}+1 \right) \big(\|p-p'\|_1\big)^\alpha \ln d - \big(\|p-p'\|_1\big)^\alpha \ln \|p-p'\|_1 \right]  . \label{eq:T-T bound alpha}
	\ea 
	Since the above equation is continuous around $\alpha=1$, it must also hold for $\alpha\in(0,1]$\footnote{One can easily prove this by contradiction. Imagine that the bound Eq. \eqref{eq:T-T bound alpha} does not hold for $\alpha=1$. Then there must exist an $\epsilon>0$ such that it also does not hold for $\alpha=1-\epsilon$, but this would be a contradiction.}
 which is Eq. \eqref{eq:Tsalis uniform around 1 in theorem alpha 0 1} of the theorem. For $\alpha\in(0,1)$ we can also find an appropriate bound on $\|p-p'\|_1$ so that the constraint $\| q_\ccc(p)- q_\ccc(p')\|_1 \leq 1/\me$ is satisfied. Using Eqs, \eqref{eq:q-q Tsalis}, \eqref{eq:Delta for alpha 0 1}, \eqref{eq:alpha 0 1 for constriant}, we find
	\be 
	\| q_\ccc(p)- q_\ccc(p')\|_1 \leq \frac{\Delta_\ccc(p,p')}{\|p\|_\ccc^\ccc}\leq \frac{2 d\left(\|p-p'\|_1/d\right)^\alpha}{\|p\|_\ccc^\ccc}\leq 2 d\left(\|p-p'\|_1/d\right)^\alpha.
	\ee 
	Thus imposing the constraint $2 d\left(\|p-p'\|_1/d\right)^\alpha \leq 1/\me$ we achieve the condition
	\be 
	\|p-p'\|_1 \leq d \left(\frac{1}{2\me\,d}\right)^{1/\alpha},
	\ee 
	which is the condition in the text above Eq. \eqref{eq:Tsalis uniform around 1 in theorem alpha 0 1} in the theorem.\\
	
	Similarly, plugging Eqs. \ref{eq:Delta for alpha 1 infty}, \ref{eq:norm case 1 alpah 1 infinity} into Eq. \eqref{eq:T-T  30}, we find for $\alpha\in(1,\infty)$
	\ba 
	\left| T_\alpha(p)- T_\alpha(p') \right| & \leq \frac{1}{\ccc} \left(\Delta_\ccc(p,p') \ln d^{3/2} -\Delta_\ccc(p,p') \ln \Delta_\ccc(p,p') \right) \\
	& \leq \frac{4 \lceil \ccc \rceil}{\ccc} \left( \|p-p'\|_1 \ln d^{3/2} - \|p-p'\|_1 \ln \left(4 \lceil \ccc \rceil \|p-p'\|_1 \right)\right)\\
	&\leq \frac{4 ( \ccc +1)}{\ccc} \left( \|p-p'\|_1 \ln \left(\frac{d^{3/2}}{4}\right) - \|p-p'\|_1 \ln \left(  \|p-p'\|_1 \right) - \|p-p'\|_1 \ln \left( \lceil \ccc \rceil\right) \right)\\
	& \leq 8 \left( \|p-p'\|_1 \ln \left(\frac{d^{3/2}}{4}\right) - \|p-p'\|_1 \ln \left( \|p-p'\|_1 \right)\right),\quad \forall \ccc\,\in(1,\alpha)\label{eq:Y-T final in proof}
	\ea
	which is Eq. \eqref{eq:Tsalis uniform around 1 in theorem alpha 1 infty} in the theorem. Similarly to above, for $\alpha\in(1,\infty)$ we can also find an appropriate bound on $\|p-p'\|_1$ so that the constraint $\| q_\ccc(p)- q_\ccc(p')\|_1 \leq 1/\me$ is satisfied. Using Eqs, \eqref{eq:q-q Tsalis}, \eqref{eq:alpha 1 infinity for constriant}, we find
	\be 
	\| q_\ccc(p)- q_\ccc(p')\|_1 \leq \frac{\Delta_\ccc(p,p')}{\|p\|_\ccc^\ccc}\leq \frac{2 \lceil \ccc\rceil \|p-p'\|_1}{\left(\frac{1}{d}\right)^c}= 2 \lceil \ccc\rceil d^\ccc\, \|p-p'\|_1 \quad \forall\, \ccc\in(1,\alpha).
	\ee 
	Thus imposing the constraint $2 \lceil \ccc\rceil d^\ccc\, \|p-p'\|_1 \leq 1/\me$ for all $\ccc\in(1,\alpha)$. We observe that this is satisfied for all $\ccc\in(1,\alpha)$ if
	\be \label{eq:contraint for alpha greater than 1}
	\|p-p'\|_1 \leq \frac{1}{2\me\lceil \alpha\rceil d^\alpha} 
	\ee 
	which is the condition in the text above Eq. \eqref{eq:Tsalis uniform around 1 in theorem alpha 1 infty} in the theorem.

\end{proof}

\subsection{R\'enyi entropy continuity theorem}\label{Renyi Entropy Continuity Theorem}

\begin{theorem}[R\'enyi uniform continuity]\label{thm:renyi continuity} Let $p,p'\in\cP_d$ have entries denoted $p_k$. For the following parameters, we have the following R\'enyi entropy (Eq. \ref{eq:deef renyi entrpies}) continuity bounds:
	\begin{itemize}
		\item[0)] For $\alpha\in[-\infty,-1]$, we have
		\be \label{Eq: 01 cont thm 1}
		\frac{\alpha}{1-\alpha} \ln \left( \frac{\min_k\{ p_k \}}{\min_l \{p_l p_l'\}} \|p-p'\|_1+1 \right) \leq S_\alpha(p)-S_\alpha(p') \leq \frac{|\alpha|}{1-\alpha} \ln \left( \frac{\min_k\{ p_k' \}}{\min_l \{p_l p_l'\}} \|p-p'\|_1+1 \right)
		\ee 
		and 
		\be \label{Eq: 01 cont thm 2}
		\frac{\alpha}{1-\alpha} \ln \left( \frac{\|p-p'\|_1}{\min_l \{p_l' \}} +1 \right) \leq S_\alpha(p)-S_\alpha(p') \leq \frac{|\alpha|}{1-\alpha} \ln \left( \frac{\|p-p'\|_1}{\min_l \{p_l\}} +1 \right)
		\ee 
		\item[1)]
		For $\alpha\in(0,1)$, we have
		\be \label{eq:thm eq alpha between 0 and 1}
		|S_\alpha(p)-S_\alpha(p')| \leq  \frac{\me^{(\alpha-1) S_\alpha(p)}}{1-\alpha}  \, d^{(1-\alpha)} \Big( \|p-p'\|_1 \Big)^\alpha  .
		\ee		
		\item [2)]
		Let $\epsilon_0>1$. Then for all $\alpha\in[\epsilon_0,\infty]$, we have 
		\be\label{eq:thm eq alpha gre 1}
		|S_\alpha(p)-S_\alpha(p')| \leq \frac{\epsilon_0}{\epsilon_0-1} \ln\big(1+ \|p-p'\|_1 d\,\big).
		\ee 
		\item[3.1)]
		Let $\alpha\in[1/2,1)$, and $\|p-p'\|_1 \leq {1}/{[(4 \me)^2 d]}$, then
		\be\label{eq:thm eq alpha beibourghood 2}
		|S_\alpha(p)-S_\alpha(p')| \leq  8\, d \left( 6\ln d -\ln 2 \right) \sqrt{\|p-p'\|_1} - 4\, d \sqrt{\|p-p'\|_1} \ln \sqrt{\|p-p'\|_1} 
		\ee
		\item[3.2)]
		Let $\alpha\in[1,2]$, and $\|p-p'\|_1 \leq d/(8 \me)^2$, then
		\be\label{eq:thm eq alpha beibourghood 1}
		|S_\alpha(p)-S_\alpha(p')| \leq 2\sqrt{d} \left( \|p-p'\|_1 \ln d + \sqrt{\|p-p'\|_1}\, 4d\ln\left( d/64 \right)- 8d \sqrt{\|p-p'\|_1}\ln \sqrt{\|p-p'\|_1} \right).
		\ee	
	\end{itemize}

\end{theorem}
\begin{proof}
	This proof is divided into subsections, one for each $\alpha$ regime, 0), 1), 2), 3). 
	\subsubsection{Proof of 0). $\alpha\in[-\infty,-1]$}
	For $\alpha<0$, we have
	\be 
	S_\alpha(p)-S_\alpha(p')= \frac{-\alpha}{1-\alpha}\ln \|p\|_\alpha - \frac{-\alpha}{1-\alpha}\ln \|p'\|_\alpha = \frac{|\alpha|}{1-\alpha} \ln \left( \frac{\|p^{-1}\|_{|\alpha|}}{\|p^{\prime-1}\|_{|\alpha|}} \right),
	\ee 
	where $\|x^{-1}\|_\alpha:= \sum_{i=1}^d (1/x_i)^\alpha$. Therefore, 
	\be\label{eq:S-S neg}
	S_\alpha(p)-S_\alpha(p')\geq 0
	\ee 
	iff $\|p^{-1}\|_{|\alpha|}/ \|p^{\prime-1}\|_{|\alpha|}\geq 0$. \\
	We will now provide a proof for $\alpha\in[-\infty,- 1]$ and finalise the proof for the remaining negative interval afterwards. Assume $\|p^{-1}\|_{|\alpha|}/ \|p^{\prime-1}\|_{|\alpha|}\geq 0$,\footnote{Note that the upper bound Eq. \eqref{eq:S-S neg} is non-negative. Therefore, if assumption $\|p^{-1}\|_{|\alpha|}/ \|p^{\prime-1}\|_{|\alpha|}\geq 0$ does not hold, the bound will be trivially true since $S_\alpha(p)-S_\alpha(p')$ will be negative.} we then have using the $p$-norm triangle inequality (Lemma \ref{norm lem trangle inequality}), 
	\be 
	S_\alpha(p)-S_\alpha(p')=   \frac{|\alpha|}{1-\alpha} \ln \left( \frac{\|p^{-1}-p^{\prime-1}+p^{\prime-1}\|_{|\alpha|}}{\|p^{\prime-1}\|_{|\alpha|}} \right) \leq \frac{|\alpha|}{1-\alpha} \ln \left( \frac{\|p^{-1}-p^{\prime-1}\|_{|\alpha|}}{\|p^{\prime-1}\|_{|\alpha|}} +1 \right).
	\ee 
	Also from Lemma \ref{norm lem trangle inequality}, it follows
	\ba
	\|p^{-1}-p^{\prime-1}\|_{|\alpha|} & \leq \|p^{-1}-p^{\prime-1}\|_{1}   = \sum_{i=1}^d \frac{|p_i-p_i'|}{p_i p_i'} \leq \frac{1}{\min_k{p_k p_k'}} \|p-p'\|_1\\
	\|p^{\prime-1}\|_{|\alpha|} &\geq  \|p^{\prime-1}\|_{\infty}= \max_i \{1/p'_i\}=1/\min_i\{p'_i\}. 
	\ea 
	Therefore, by plugging in the above inequalities,
	\be\label{up bound for o1 neg 1} 
	S_\alpha(p)-S_\alpha(p') \leq \frac{|\alpha|}{1-\alpha} \ln \left( \frac{\min_k\{ p_k' \}}{\min_l \{p_l p_l'\}} \|p-p'\|_1+1 \right)
	\ee
	which is the R.H.S. of Eq. \eqref{Eq: 01 cont thm 1}. For the proof of the L.H.S. of Eq. \eqref{Eq: 01 cont thm 1}, we note that this term is negative. Therefore, it is trivially true if $S_\alpha(p)-S_\alpha(p')\geq 0$. When $S_\alpha(p)-S_\alpha(p')< 0$, we can write $S_\alpha(p)-S_\alpha(p')=-( S_\alpha(p')-S_\alpha(p))$, where the term in brackets in positive. Thus from Eq. \eqref{eq:S-S neg}, we can use upper bound Eq. \eqref{up bound for o1 neg 1} with $p\mapsto p'$, $p'\mapsto p$ to achieve the L.H.S. of Eq. \eqref{Eq: 01 cont thm 1}. Eqs. \eqref{Eq: 01 cont thm 2} follow from noting $\min_k\{p_kp_k'\}\geq \min_k\{p_k\}\min_k\{p_k'\}$.
	\subsubsection{Proof of 1). $\alpha\in(0,1)$}		
	From notes from old bound [delete old notes and just put that section here], we have 
	\be
	|S_\alpha(p)-S_\alpha(p')| \leq   \frac{1}{1-\alpha} \, \frac{1}{\sum_{i=1}^d  p_i^\alpha} \sum_{i=1}^d | p_i^\alpha -p_i^{\prime\alpha} |= \frac{1}{1-\alpha} \, \frac{1}{\| p_i\|_\alpha^\alpha} \sum_{i=1}^d | p_i^\alpha -p_i^{\prime\alpha} |.
	\ee 
	Eq. \eqref{eq:thm eq alpha between 0 and 1} now follows directly by bounding $ \sum_{i=1}^d | p_i^\alpha -p_i^{\prime\alpha} |$ using Lemma \ref{Lemm:sum diff up bound}, and the relationship between R\'enyi entropies, and $\| p\|_\alpha$, namely $(1-\alpha) S_\alpha(p)= \ln \|p\|_\alpha^\alpha$.

	\subsubsection{Proof of 2). $\alpha\in[\epsilon_0,\infty]$}\label{sec:of (2)}
	We start with proving Eq. \eqref{eq:thm eq alpha gre 1}.
	\ba\label{eq:reyi difference bound general}
	\left| S_\alpha(p)-S_\alpha(p') \right| &= \left|  \frac{\alpha}{1-\alpha} \ln \|p'\|_\alpha - \frac{\alpha}{1-\alpha} \ln \|p\|_\alpha   \right| = \left|  \frac{\alpha}{1-\alpha}\right|\, \left| \ln\left(\frac{ \|p'\|_\alpha}{\|p\|_\alpha}\right) \right| \\
	&=
	\begin{cases}
		\left|  \frac{\alpha}{1-\alpha}\right|\,  \ln\left(\frac{ \|p'\|_\alpha- \|p\|_\alpha}{\|p\|_\alpha}+1\right) ,\quad \text{for } \|p'\|_\alpha- \|p\|_\alpha \geq 0, \vspace{0.3cm}\\
		\left|  \frac{\alpha}{1-\alpha}\right|\,  \ln\left(\frac{ \|p\|_\alpha- \|p'\|_\alpha}{\|p'\|_\alpha}+1\right),\quad \text{for }  \|p\|_\alpha- \|p'\|_\alpha \geq 0,
	\end{cases}	\\
	&= 	\left|  \frac{\alpha}{1-\alpha}\right|\,  \ln\left(\frac{\Big| \|p\|_\alpha- \|p'\|_\alpha \Big|}{\|p''\|_\alpha}+1\right),
	\ea 
	where
	\be 
	\|p''\|_\alpha = \begin{cases}
		\|p\|_\alpha \quad \text{if } \|p'\|_\alpha-\|p\|_\alpha \geq 0 \\ 
		\|p'\|_\alpha \quad \text{otherwise}.
	\end{cases}
	\ee  
	For $\alpha > 1$, from Lemma \ref{Lemm:sum diff up bound}, we have $\Big| \|p\|_\alpha- \|p'\|_\alpha \Big|\leq \|p-p'\|_1$. Furthermore, using Eq. \eqref{eq:p r inequalities}, we have $\|p''\|_\alpha \geq \|p''\|_\infty$. However, we also have that $\|p''\|_\infty=\max_{i}\{p_i''\}\geq 1/d$. Thus Eq. \eqref{eq:thm eq alpha gre 1} follows since the logarithm is an increasing function.

	\subsubsection{Proof of 3.1) and 3.2). $\alpha\in[1/2, 2]$}
	
	We now move on to the proof of Eqs. \eqref{eq:thm eq alpha beibourghood 1},\eqref{eq:thm eq alpha beibourghood 2}. To start with, we define 
	the function $F$ to be any upper bound to 
	\be 
	\left| \frac{1}{\alpha-1} \frac{\|p'\|_\alpha -\|p\|_\alpha}{\|p''\|_\alpha}\right| \leq F_\alpha(p,p').
	\ee  
	Using Eq. \eqref{eq:reyi difference bound general}, we can now write
	\ba\label{eq:S alpha change}
	| S_\alpha(p)-S_\alpha(p')| & = \left|  \frac{\alpha}{1-\alpha}\right|\,  \ln\left(\frac{\Big| \|p\|_\alpha- \|p'\|_\alpha \Big|}{\|p''\|_\alpha}+1\right) \leq \frac{\alpha}{|\alpha-1|} \ln\left( |\alpha-1| F_\alpha(p,p')+1 \right)\\
	& \leq \alpha F_\alpha(p,p').
	\ea
	where in the last line we have used the inequality $\ln(x+1)\leq x$ for $x\geq 0$. 
	
	We now set out to find an appropriate expression for $F_\alpha(p,p')$. For this we will use the mean value theorem, Theorem \ref{lemm:MVT}. We start by finding two separate expressions for the function
	\footnote{Here it will be assumed that $g_\alpha(p,p')$ is differentiable w.r.t. $\alpha$ on the interval $(1,\alpha)$ for $\alpha\geq 1$ and $(\alpha,1)$ for $\alpha\leq 1$. Later we will calculate explicitly its derivative, thus verifying this assumption.}
	\be 
	g_\alpha(p,p'):=  {\|p'\|_\alpha -\|p\|_\alpha}.
	\ee 
	Using the notation $a,b,c,$ from the mean value theorem (Thm. \ref{lemm:MVT}), we have
	\begin{itemize}
		\item [1)] $a=1$, $b=\alpha$, $\alpha\geq 1$ 
		\be \label{eq:g alpha p p 1}
		g_\alpha(p,p')=g_1(p,p') +g_c'(p,p') (\alpha-1)=g_c'(p,p') (\alpha-1), \quad \text{for some } c\in(1,\alpha).
		\ee 
		\item[2)] $b=1$, $a=\alpha$, $\alpha\leq 1$
		\be\label{eq:g alpha p p 2}
		g_\alpha(p,p')=g_1(p,p') +g_c'(p,p') (-1+\alpha)= g_c'(p,p') (\alpha-1), \quad \text{for some } c\in(\alpha,1).
		\ee 
	\end{itemize}
	Where in both cases we have used $\|p\|_1=\|p'\|_1=1$. 
	We thus conclude
	\be 
	g_\alpha(p,p')= g_\beta'(p,p')(\alpha-1),\quad \text{for some } \beta\in \begin{cases}
		(\alpha,1) &\text{ if } \alpha <1\\
		(1,\alpha) &\text{ if } \alpha  \geq 1.
	\end{cases}
	\ee 
	We thus have 
	\be \label{eq:F beta def}
	\left| \frac{1}{\alpha-1} \frac{\|p'\|_\alpha -\|p\|_\alpha}{\|p''\|_\alpha}\right|=  \frac{|g_\beta'(p,p')|}{\|p''\|_\alpha}=\frac{1}{\|p''\|_\alpha}\left| \frac{d}{d\beta}  \bigg(\|p'\|_\beta -\|p\|_\beta\bigg) \right| := F_\beta(p,p').
	\ee 
	We thus have taking into account eq. \eqref{eq:S alpha change},
	\be\label{eq:S alpha up bound dev}
	|S_\alpha(p)-S_\alpha(p')|  \leq \alpha\,\sup_{\beta\in\mathcal{I}} F_\beta(p,p'),
	\ee
	where $\mathcal{I}=[1/2,1]$ when $\alpha\leq 1$, and $\mathcal{I}=[1,2]$ when $\alpha \geq 1$.
	We have so-far managed to remove the singularity at $\alpha=1$ in our upper bound to the R\'enyi entropies. We will now set out to prove a relationship between this upper bound and the distance $\|p-p'\|_1$. We start by find the derivative. For convenience, note that for $x\in\cc^n$, 
	\be 
	\|x\|_{\alpha}  = \exp \left({\frac 1 {\alpha}  \ln \sum_{i=1}^\infty |x_i|^{\alpha}}\right).
	\ee 
	Hence the derivative is
	\ba
	{\frac{d}{d\alpha}} \|x\|_{\alpha} &=  \|x\|_{\alpha} {\frac{d}{d\alpha}} \left({\frac 1 {\alpha}  \ln \sum_{i=1}^\infty {|x_i|}^{\alpha} }\right) = 
	\left( { \left({ {\frac{d}{d\alpha}} \frac 1 {\alpha} } \right) \ln \sum_{ i =1 }^\infty {|x_i|}^{\alpha} +
		\frac 1 {\alpha} {\frac{d}{d\alpha}} \ln \sum_{i=1}^\infty {|x_i|}^{\alpha}} \right) \\
	&=\|x\|_{\alpha}  \left( { \frac {-1} {{\alpha}^2} \ln \sum_{ i =1 }^\infty {|x_i|}^{\alpha} +
		\frac 1 {\alpha} \frac 1 { \sum_{i=1}^\infty {|x_i|}^{\alpha} } {\frac{d}{d\alpha}} \sum_{i=1}^\infty {|x_i|}^{\alpha} }\right).
	\ea 
	However, ${\frac{d}{d\alpha}} \sum_{i=1}^\infty {|x_i|}^{\alpha} = {\frac{d}{d\alpha}} \sum_{i=1}^\infty e^{{\alpha} \ln {|x_i|}} = \sum_{i=1}^\infty {|x_i|}^{\alpha} \ln {|x_i|}$, therefore
	\ba
	{\frac{d}{d\alpha}} \|x\|_{\alpha} &=\|x\|_{\alpha} \left( { \frac {-1} {{\alpha}^2} \ln \sum_{ i =1 }^\infty {|x_i|}^{\alpha} +
		\frac 1 {\alpha} \frac 1 { \sum_{i=1}^\infty {|x_i|}^{\alpha} } \sum_{i=1}^\infty {|x_i|}^{\alpha} \ln {|x_i|} }\right)\\
	&=\frac{\|x\|_{\alpha}}{\alpha}  \left( - \ln \|x\|_{\alpha}  +\frac{1}{\|x\|_{\alpha}^\alpha} \sum_{i=1}^d |x_i|^\alpha \ln |x_i| \right). \label{eq:div sudo p norm}
	\ea
	Now, by direct calculation we observe that the above line can be written in terms of the $S_1$, Shannon entropy for a probability distribution which depends on $\alpha$, namely\footnote{Recall that $S_1(x) = -\sum_{i=1}^d |x_i| \ln |x_i|$.}
	\be 
	S_1\left(q_\alpha(x)\right)= -\alpha \left( - \ln \|x\|_{\alpha}  +\frac{1}{\|x\|_{\alpha}^\alpha} \sum_{i=1}^d |x_i|^\alpha \ln |x_i| \right),
	\ee 
	where the components of the normalised probability vector $q_\alpha(x)$ are 
	\be 
	[q_\alpha(x)]_i := \frac{|x_i|^\alpha}{\|x\|_\alpha^\alpha}, \quad i=1,2,3,\ldots ,d.
	\ee 
	Therefore, 
	\be\label{eq:div p-norm for x}
	{\frac{d}{d\alpha}} \|x\|_{\alpha} = -\frac{\|x\|_\alpha}{\alpha^2} S_1\left( q_\alpha(x) \right).
	\ee 
	From Eq. \eqref{eq:F beta def},
	\be 
	\beta^2 \,{\|p''\|_\alpha} F_\beta(p,p') =\beta^2\left| \frac{d}{d\beta}  \bigg(\|p'\|_\beta -\|p\|_\beta\bigg) \right| = \Big|\|p'\|_{\beta}\, S_1\left( q_\alpha(p') \right) - \|p\|_{\beta} \, S_1\left( q_\alpha(p) \right)\Big|,
	\ee
	Thus,
	\ba
	\beta^2\,{\|p''\|_\alpha} F_\beta(p,p') & =\|p\|_{\beta} \bigg|\,\left(\frac{\|p'\|_{\beta}}{\|p\|_{\beta}}-1\right)  S_1\left( q_\alpha(p') \right)+S_1\left( q_\alpha(p')\right) - S_1\left( q_\alpha(p) \right)\bigg| \\
	&\leq  \|p\|_{\beta} \bigg(\,\left|\frac{\|p'\|_{\beta}}{\|p\|_{\beta}}-1\right|  \left|S_1\left( q_\alpha(p') \right)\right|+\left|S_1\left( q_\alpha(p') \right) -  S_1\left( q_\alpha(p) \right)\right|\bigg)\\
	& = \left|S_1\left( q_\alpha(p') \right)\right| \big|\|p'\|_{\beta}-\|p\|_{\beta}\big|  + \|p\|_{\beta}\big|S_1\left( q_\alpha(p') \right) -  S_1\left( q_\alpha(p) \right)\big|\\
	& \leq   \left(  \max_{q\in\cP_d}\left|S_1(q)\right|\right) \Big|\|p'\|_{\beta}-\|p\|_{\beta}\Big|  + \|p\|_{\beta}\Big| S_1\left(q_\beta(p')\right) -  S_1\left(q_\beta(p)\right)\Big|.
	\ea
	
	Therefore, noting that the Shannon entropy is maximized for the uniform distribution and applying the Fannes inequality (Lemma \ref{Lemm:Fannes}), we achieve
	\ba
	\beta^2 \,{\|p''\|_\alpha} F_\beta(p,p') & \leq {\ln d} \, \Big|\|p'\|_{\beta}-\|p\|_{\beta}\Big|  +{\|p\|_{\beta}}  \,      \Big( \| q_\beta(p)- q_\beta(p')\|_1 \ln d  - \| q_\beta(p)- q_\beta(p')\|_1 \ln\left( \|q_\beta(p)- q_\beta(p')\|_1 \right)    \Big).
	\ea
	We now pause a moment to bound $\|q_\beta(p)- q_\beta(p')\|_1 $. Using the definition of $q_\alpha(p)$, we have
	\ba
	\|q_\beta(p)- q_\beta(p')\|_1 & =\sum_{i=1}^d \left| \frac{p_i^\beta}{\|p\|_\beta^\beta}+ \frac{p_i^{\prime\beta}}{\|p'\|_\beta^\beta} \right| = \sum_{i=1}^d \frac{1}{\|p\|_\beta^\beta} \Bigg| p_i^\beta -p_i^{\prime \beta} + p_i^{\prime \beta}\left( 1-\frac{\|p\|_\beta^\beta}{\|p'\|_\beta^\beta}  \right) \Bigg|\\
	& \leq   \sum_{i=1}^d \frac{1}{\|p\|_\beta^\beta} \Bigg( \left|p_i^\beta -p_i^{\prime \beta}\right|  + p_i^{\prime \beta}\left| 1-\frac{\|p\|_\beta^\beta}{\|p'\|_\beta^\beta}  \right| \Bigg)  = \frac{1}{\|p\|_\beta^\beta} \Bigg(  \Big| \|p\|_\beta^\beta -  \|p'\|_\beta^\beta  \Big| +   \sum_{i=1}^d \left|p_i^\beta -p_i^{\prime \beta}\right|\Bigg)\\
	&\leq  \frac{2}{\|p\|_\beta^\beta} \Bigg(    \sum_{i=1}^d \left|p_i^\beta -p_i^{\prime \beta}\right|\Bigg) = \frac{\Delta(p,p')}{\|p\|_\beta^\beta} ,
	\ea 
	where in the last line, we have used Lemma \ref{Lemm:sum diff up bound} and defined,
	\be \label{eq: F beta up bound 2}
	\Delta(p,p'):=2\sum_{i=1}^d \big| p_i^\beta-p_i^{\prime \beta}  \big|.
	\ee 
	Therefore, for $\| q_\beta(p)- q_\beta(p')\|_1 \leq 1/\me$, 
	\be\label{eq:up bound final before deivde}
	\beta^2  \,{\|p''\|_\alpha} F_\beta(p,p')  \leq  \, \Big|\|p'\|_{\beta}-\|p\|_{\beta}\Big| {\ln d} +\|p\|_\beta^{1-\beta}  \Big( \Delta(p,p') \ln d-\Delta(p,p') \ln \Delta(p,p') +\beta\,\Delta(p,p') \ln \|p\|_\beta \Big).
	\ee
	We will now proceed to bound Eq. \eqref{eq: F beta up bound 2} separately for $\beta\in [1/2,1)$, and $\beta\in[1,2]$. We start with the easier of the two.
	\begin{itemize}
		\item [For] $\beta\in[1,2]$:\\
		Setting $r=1$, $p=\beta$ in Eq. \eqref{eq:p r inequalities}, and recalling $\|p\|_1=1$, it follows
		\be \label{eq:up low for beta norm 1 to 2}
		d^{1/\beta -1} \leq \| p\|_\beta \leq 1.
		\ee 
		Similarly, from Lemma \ref{Lemm:sum diff up bound}, and assuming $\|p-p'\|_1\leq d$, we have 
		\ba 
		\Delta(p,p') &\leq 8\, d^{1-\beta/2} \Big(\|p-p'\|_1\Big)^{\beta/2} \leq 8\, d \Big(\frac{\|p-p'\|_1}{d}\Big)^{\beta/2} \leq 8 \sqrt{d\, \|p-p'\|_1}\\
		\Big|\|p'\|_{\beta}-\|p\|_{\beta}\Big| & \leq \|p-p'\|_1.
		\ea
		Furthermore more, from Eq. \eqref{eq:up low for beta norm 1 to 2} if follows
		\be 
		\|p\|_\beta^{1-\beta}\leq \frac{1}{\left(d^{1/\beta-1}\right)^{\beta-1}} \leq \sqrt{d}.
		\ee 
		We will now see which of the two constraints, $\|p-p'\|_1\leq d$, and $\| q_\beta(p)- q_\beta(p')\|_1 \leq 1/\me$ is more demanding. 
		\be 
		\| q_\beta(p)-q_\beta(p')\|_1 \leq \frac{8}{d}\sqrt{d \|p-p'\|_1} \leq \frac 1 \me \implies \|p-p'\|_1 \leq \frac{d}{(8 \me)^2} \leq d.
		\ee 
		therefore  $\|p-p'\|_1\leq d$, and $\| q_\beta(p)- q_\beta(p')\|_1 \leq 1/\me$ are both satisfied if $\|p-p'\|_1 \leq {d}/{(8 \me)^2}$. From these bounds, Eq. \eqref{eq:thm eq alpha beibourghood 1} follows.
		\item [For] $\beta\in[1/2,1)$: \\
		Setting $r=\beta$, $p=1$ in Eq. \eqref{eq:p r inequalities}, and recalling $\|p\|_1=1$, it follows
		\be \label{eq:up low for beta norm 1/2 to 1}
		1\leq \|p\|_\beta\leq d^{1/\beta-1}.
		\ee 
		From Eq. \eqref{eq:up low for beta norm 1/2 to 1}
		\be \label{eq:up for 1/2 to 1}
		\| p\|_\beta^{1-\beta} \leq \left( d^{1/\beta-1} \right)^{1-\beta}= d^{-2+\beta+1/\beta} \leq \sqrt{d}.
		\ee 
		
		Similarly, from Lemma \ref{Lemm:sum diff up bound}, and assuming $\|p-p'\|_1\leq d$, we have 
		\ba 
		\Delta(p,p') &\leq 4\, d^{1-\beta} \Big(\|p-p'\|_1\Big)^{\beta} \leq 4 d \left( \frac{\|p-p'\|_1}{d} \right)^\beta \leq 4\sqrt{d\, \|p-p'\|_1},\\
		\Big|\|p'\|_{\beta}-\|p\|_{\beta}\Big| & \leq 8\, d^{1/\beta-1/(2\beta)+1/2-\beta} \left(\|p-p'\|_1\right)^{1/2} \leq 8\, d \sqrt{\|p-p'\|_1} .
		\ea
		
		We will now see which of the two constraints, $\|p-p'\|_1\leq d$, and $\| q_\beta(p)- q_\beta(p')\|_1 \leq 1/\me$ is more demanding. 
		\be 
		\| q_\beta(p)-q_\beta(p')\|_1 \leq \frac{\Delta(p,p)}{\|p\|_\beta^\beta}\leq 4\sqrt{d \|p-p'\|_1} \leq \frac 1 \me \implies \|p-p'\|_1 \leq \frac{1}{(4 \me)^2d} \leq d.
		\ee 
		therefore  $\|p-p'\|_1\leq d$, and $\| q_\beta(p)- q_\beta(p')\|_1 \leq 1/\me$ are both satisfied if $\|p-p'\|_1 \leq {1}/{(4 \me)^2 d}$.
		
		Plugging this all into Eq. \eqref{eq:up bound final before deivde} and simplifying the resultant expression, followed by plugging int Eq. \eqref{eq:S alpha up bound dev}, we arrive at Eq. \eqref{eq:thm eq alpha beibourghood 2}.
		
	\end{itemize}

\end{proof}

\subsection{Miscellaneous Lemmas and Theorems used in the proofs to the entropy continuity Theorems \ref{thm:Tsalis continuity} and \ref{thm:renyi continuity}.}

\begin{lemma}\label{x beta - y beta lem}
	Let $0<\alpha<1$. Then $\forall$\, $x,y\geq 0,$ $\epsilon>0$,
	\begin{equation}
	|x^\alpha-y^\alpha|\leq \epsilon^\alpha+\epsilon^{\alpha-1}|x-y|
	\end{equation} 
\end{lemma}
\begin{proof}
	See Lemma 5 in \cite{Tvan}. 
\end{proof}


\begin{lemma}[$p$-norm inequalities]\label{norm lem trangle inequality}
	For $x\in\cc^n$ and $p\in(0,\infty]$, define
	\begin{equation}\label{eq:seudo p-norm}
	\|x\|_p:=\left( \sum_{q=1}^n |x_q|^p \right)^{1/p}.
	\end{equation}
	For $p\geq 1$, this is a norm, known as the \textit{p-norm}.\\
	For $0<r\leq p$ we have the \textup{$p$-norm interchange inequalities}
	\begin{equation}\label{eq:p r inequalities}
	\|x\|_p\leq \|x\|_r\leq n^{\left(\frac{1}{r}-\frac{1}{p}\right)} \|x\|_p.
	\end{equation}
	Furthermore, for $p\in [1,\infty]$, we have the \textup{$p$-norm triangle inequality},
	\be \label{eq:p-norm triangle inequality}
	\|x+y\|_p\leq \|x\|_p+\|y\|_p
	\ee 
\end{lemma}
\begin{proof}
	See \cite{Bourbaki}.
\end{proof}

\begin{lemma}[Fannes inequality \cite{Fannes}] \label{Lemm:Fannes}
	For any $p,p'\in\cP_d$, the following continuity bounds hold.
	\be 
	|S_1(p)- S_1(p')| \leq
	\begin{cases}
		\|p -p'\|_1 \ln(d)- \|p -p'\|_1\ln\left(\|p -p'\|_1\right),\quad \text{ for all } \|p -p'\|_1\\
		\|p -p'\|_1 \ln(d)+1/(\me \ln(2)),\quad \text{ if } \|p -p'\|_1\leq 1/(2\me),
	\end{cases}
	\ee 
	where $S_1$ is the Shannon entropy.
\end{lemma}
\begin{proof}
	See \cite{Fannes}.
\end{proof}
\begin{remark}
	Also see \cite{Koenraad} for a nice tightening of the bound and \cite{Winter2016} for bounds for the infinite dimensional case.
\end{remark}

\begin{theorem}[Mean Value Theorem]\label{lemm:MVT}
	Let $f:[a,b]\rightarrow \rr$ be a continuous function on the closed interval $[a,b]$, and differentiable on the open interval $(a,b)$, where $a<b$. The there exists some $c\in(a,b)$ such that
	\be 
	f'(c) = \frac{f(b)-f(a)}{b-a},
	\ee 
	where $f'(c):=\frac{d}{dx} f(x)\Bigg|_{x=c}$.
\end{theorem}
\begin{proof}
	See any introductory book to calculus.
\end{proof}

\begin{lemma}[Sum difference upper bounds]\label{Lemm:sum diff up bound}
	Let $p, p'\in \cP_d$. 
	\begin{itemize} 
		\item [1)]  For all $\alpha>0$: 
		\be \label{eq:sum diff up bound}
		\big| \|p\|_\alpha^\alpha- \|p'\|_\alpha^{ \alpha} \big| \leq 
		\sum_{i=1}^d \left|p_i^\alpha -p_i^{\prime \alpha} \right| \leq 
		2  \lceil \alpha \rceil \Big( \|p-p'\|_{\alpha/ \lceil \alpha \rceil}\Big)^{\alpha/ \lceil \alpha \rceil} \leq 2  \lceil \alpha \rceil \, d ^{(1-\alpha/ \lceil \alpha \rceil)} \left( \|p-p'\|_1 \right)^{\alpha/ \lceil \alpha \rceil}  
		\ee
		\item[2)] 
		\be \label{eq: p - p prime alpha}
		\big| \|p\|_\alpha- \|p'\|_\alpha\big| \leq 
		\begin{cases}
			4 \lceil \alpha^{-1} \rceil\, d^{( \alpha^{-1} -\alpha^{-1}/ \lceil\alpha^{-1} \rceil +1/ \lceil \alpha^{-1} \rceil -\alpha)} \big( \|p-p'\|_1 \big)^{1/\lceil \alpha^{-1}\rceil}\quad  &\text{ for } 0< \alpha < 1 \vspace{0.2cm}\\ 
			\|p-p'\|_1 \quad  &\text{ for } \alpha \geq  1. 
		\end{cases}
		\ee    
	\end{itemize}
\end{lemma}
\begin{proof} It will be partitioned into two subsections, one for Eq. \eqref{eq:sum diff up bound}, the other for Eq. \eqref{eq: p - p prime alpha}.
	
	\begin{center}
		Proof of 1)
	\end{center}
	
	The first line in Eq. \eqref{eq:sum diff up bound} follows directly from the triangle inequality. The remainder of this subsection will refer to the proof of \eqref{eq:sum diff up bound} for $\alpha >0$.
	For the second line, start by defining $\alpha_1:=\alpha/ \lceil \alpha \rceil \leq 1$ where $\lceil x \rceil:= \min y\in \mathbb{Z}$ s.t. $y>x$. Now using the identity $(x^n-y^n)=(x-y)\sum_{k=0}^{n-1}x^k y^{n-1-k}$ for $n=1,2,3,\ldots$, we have
	\begin{align}
	\left|{p_k}^\alpha-{p'_k}^\alpha\right|=\left|\big({p_k}^{\alpha_1}\big)^{\lceil \alpha \rceil}-\big({p'_k}^{\alpha_1}\big)^{\lceil \alpha \rceil}\right|=\left|{p_k}^{\alpha_1}-{p'_k}^{\alpha_1}\right|\,\left|\sum_{n=0}^{\lceil \alpha \rceil-1} \big( p_k^{\alpha_1}\big)^n \big({p'_k}^{\alpha_1}\big)^{\lceil \alpha \rceil-1-n}\right|\leq \lceil \alpha \rceil  \left|{p_k}^{\alpha_1}-{p'_k}^{\alpha_1}\right|.
	\end{align}
	If $\alpha\in\nn^+$ the proof of the second inequality in Eq. \eqref{eq:sum diff up bound} is complete. Otherwise $\alpha_1<1$ and we can employ Lemma \ref{x beta - y beta lem} with $\epsilon= \left|{p_k}-{p'_k}\right|$ to achieve 
	\begin{align}
	\left|{p_k}^\alpha-{p'_k}^\alpha\right| \leq 2\, \lceil \alpha \rceil  \left|{p_k}-{p'_k}\right|^{\alpha_1} = 2 \lceil \alpha \rceil  \left|{p_k}-{p'_k}\right|^{\alpha/   \lceil \alpha \rceil },
	\end{align}
	from which the proof of second inequality in Eq. \eqref{eq:sum diff up bound}  follows. To achieve the third inequality, we employ Lemma \ref{norm lem trangle inequality} with $r=\alpha/\lceil \alpha \rceil$, $p=1$. \\
	
	\begin{center}
		Proof of 2)
	\end{center}
	For $\alpha \geq 1$, this is easy. Using the $p$-norm triangle inequality, Eq. \eqref{eq:p-norm triangle inequality}, twice we have 
	\be 
	- \|p-p'\|_\alpha\leq \|p\|_\alpha -\|p'-p+p\|_\alpha= \| p\|_\alpha-\|p'\|_\alpha = \|p-p'+p'\|_\alpha -\|p'\|_\alpha \leq \|p-p'\|_\alpha.
	\ee 
	Therefore, from the monotonicity of the $p$-norm, Eq. \eqref{eq:p r inequalities}, we find
	\be 
	\big| \|p\|_\alpha- \|p'\|_\alpha\big| \leq \|p-p'\|_\alpha\leq \|p-p'\|_1.
	\ee  
	For $\alpha\in(0,1)$, we have to do a bit more work since the $p$-triangle inequality does not apply. Define $\beta_1:= \beta^{-1}/\lceil \beta^{-1} \rceil \leq 1$. We can write
	\ba 
	\Big|\|p\|_\alpha -\|p'\|_\alpha\Big| &= \Bigg|\left[\left(\|p\|_\alpha^\alpha\right)^{\beta_1} \right]^{\lceil \alpha^{-1}\rceil}-\left[\left(\|p'\|_\alpha^\alpha\right)^{\beta_1} \right]^{\lceil \alpha^{-1}\rceil}  \Bigg| \\
	&= \Bigg|\left(\|p\|_\alpha^\alpha\right)^{\beta_1} -\left(\|p'\|_\alpha^\alpha\right)^{\beta_1}   \Bigg|\, \Bigg| \sum_{n=0}^{\lceil \alpha^{-1} \rceil -1} \left[\left(\|p\|_\alpha^\alpha\right)^{\beta_1} \right]^n \left[\left(\|p'\|_\alpha^\alpha\right)^{\beta_1} \right]^{\lceil \alpha^{-1}\rceil-1 -n}  \Bigg|,
	\ea 
	where in the last line we have applied the identity $(x^n-y^n)=(x-y)\sum_{k=0}^{n-1}x^k y^{n-1-k}$ for $n=1,2,3,\ldots$.
	Applying Eq. \eqref{eq:p r inequalities} for $r=\alpha$, $p=1$, and noting $\|p\|_1=1$, we find 
	\be 
	\left( \|p\|_\alpha^\alpha \right)^{\beta_1} \leq \left(d^{\alpha^{-1}-1}\right)^{\alpha \beta_1} = d^{(\alpha^{-1}-1)/\lceil \alpha^{-1}\rceil}.
	\ee 
	Therefore, plugging in this upper bound we find
	\ba 
	\Big|\|p\|_\alpha -\|p'\|_\alpha\Big| &\leq \Bigg|\left(\|p\|_\alpha^\alpha\right)^{\beta_1} -\left(\|p'\|_\alpha^\alpha\right)^{\beta_1}   \Bigg|\, \Bigg| \sum_{n=0}^{\lceil \alpha^{-1} \rceil -1} \left[  d^{(\alpha^{-1}-1)/\lceil \alpha^{-1}\rceil}\right]^n \left[  d^{(\alpha^{-1}-1)/\lceil \alpha^{-1}\rceil} \right]^{\lceil \alpha^{-1}\rceil-1 -n}  \Bigg| \\
	& \leq \Bigg|\left(\|p\|_\alpha^\alpha\right)^{\beta_1} -\left(\|p'\|_\alpha^\alpha\right)^{\beta_1}   \Bigg|\,  \lceil \alpha^{-1}\rceil  \left[  d^{(\alpha^{-1}-1)/\lceil \alpha^{-1}\rceil} \right]^{\lceil \alpha^{-1}\rceil-1} \\
	& \leq \Bigg|\left(\|p\|_\alpha^\alpha\right)^{\beta_1} -\left(\|p'\|_\alpha^\alpha\right)^{\beta_1}   \Bigg|\,  \lceil \alpha^{-1}\rceil\,  d^{\alpha^{-1}-\alpha^{-1}/\lceil \alpha^{-1}\rceil+1 /\lceil \alpha^{-1}\rceil -1 }.
	\ea 
	Now for $\beta_1<1$ apply Lemma \ref{x beta - y beta lem}, with $\epsilon=\| p-p'\|_\alpha$, $x=\|p\|_\alpha^\alpha$, $y=\|p'\|_\alpha^\alpha$ to achieve
	\ba 
	\Big|\|p\|_\alpha -\|p'\|_\alpha\Big| & \leq \Bigg|\left(\|p\|_\alpha^\alpha\right)^{\beta_1} -\left(\|p'\|_\alpha^\alpha\right)^{\beta_1}   \Bigg|\,  \lceil \alpha^{-1}\rceil\,  d^{\alpha^{-1}-\alpha^{-1}/\lceil \alpha^{-1}\rceil+1 /\lceil \alpha^{-1}\rceil -1 }\\
	& \leq \Big|\|p\|_\alpha^\alpha -\|p'\|_\alpha^\alpha  \Big|^{\beta_1}    \, 2 \lceil \alpha^{-1}\rceil\,  d^{\alpha^{-1}-\alpha^{-1}/\lceil \alpha^{-1}\rceil+1 /\lceil \alpha^{-1}\rceil -1 }.
	\ea 
	By inspection, we observe that the inequality also holds when $\beta_1=1$. We now plug in Eq. \eqref{eq: p - p prime alpha}, to find 
	\ba 
	\Big|\|p\|_\alpha -\|p'\|_\alpha\Big|& \leq \left(\|p-p'\|_1\right)^{1/(\lceil \alpha\rceil \lceil \alpha^{-1}\rceil)}  \left(2 \lceil \alpha \rceil\right)^{\beta_1}  \, 2 \lceil \alpha^{-1}\rceil\,  d^{\alpha^{-1}-\alpha^{-1}/\lceil \alpha^{-1}\rceil+1 /\lceil \alpha^{-1}\rceil -1 }.
	\ea 
	Thus noting that $(2 \lceil \alpha \rceil)^{\beta_1}\leq 2$, we achieve Eq. \eqref{eq: p - p prime alpha} for $0<\alpha<1$.
	
\end{proof}

\subsection{Refinement of Theorem \ref{thm:noemb}.}

The following corollary allows one to write Theorem \ref{thm:noemb} in the form of Theorem \ref{thm:noemb physical}.
\begin{corollary}\label{rem:tigher bound on ep emb}
	For all fixed $d_\Sy$, and in the limits $\epemb\to 0^+$, $D_\cat\to\infty$, Theorem \ref{thm:noemb} holds under the replacement of Eq. \eqref{eq:long ep emb form} by
	\begin{align}\label{eq:old ep res bound 2}
		\epres(\epemb,d_\Sy, D_\cat)&= 
		15\sqrt{ \frac{d_\Sy  \ln D_\cat}{\ln\left(1/\epemb\right) } \left(1 + D_\cat \epemb^{1/7} \right)}
	\end{align}
\end{corollary}
\begin{proof}
	The proof consists in upper bounding line \ref{eq:last line} in the proof of Theorem \ref{thm:noemb} before setting $\epemb$ equal to it. Continuing from line  \ref{eq:last line} but now taking $D$ to be large and $\epemb$ to be small followed by neglecting higher order terms, one finds
	\begin{align}
		&5\sqrt{\frac{d_\Sy^{5/3} + 4(\ln D) \ln d_\Sy}{\ln (1/\epemb)}+ D \epemb^{1/6} + 5\left( D^2 \sqrt{\frac{\epemb}{D}} \ln\sqrt{\frac{D}{\epemb}}  \right)^\frac23}\\
		\leq & \, 5\sqrt{\frac{4(\ln D_\cat) \ln d_\Sy}{\ln (1/\epemb)}+ D \epemb^{1/6} + 5 D\left( \ln\sqrt{D}\right)^{2/3} \left( \sqrt{\epemb}+\sqrt{\epemb} \ln\sqrt{\frac{1}{\epemb}}  \right)^\frac23}\\
		\leq &\, 5\sqrt{\frac{4(\ln D_\cat) \ln d_\Sy}{\ln (1/\epemb)}+ D \epemb^{1/6} + 5 D\left(\frac12\right)^{2/3}\left( \ln D\right)^{2/3} \left( \sqrt{\epemb} \ln\frac{1}{\epemb}  \right)^\frac23}\\
		\leq &\, 5\sqrt{\frac{4(\ln D_\cat) \ln d_\Sy}{\ln (1/\epemb)}+ D \epemb^{1/7}\epemb^{1/42} + 4 D\left( \ln D\right) \epemb^{1/7}  \epemb^{4/21} \left(  \ln\frac{1}{\epemb}  \right)}\\
		\leq &\, 5\sqrt{\frac{4(\ln D_\cat) \ln d_\Sy}{\ln (1/\epemb)}+ D \frac{\epemb^{1/7}}{\ln (1/\epemb)} + \frac{4 D\left( \ln D\right) \epemb^{1/7}}{\ln (1/\epemb)} }\label{line:withcomment}\\
		\leq &\, 5\sqrt{\frac{4(\ln D_\cat)  d_\Sy}{\ln (1/\epemb)}+\frac{5 D\left( \ln D\right) \epemb^{1/7}}{\ln (1/\epemb)}  }\\
		\leq &\, 5\sqrt{\frac{4(\ln D_\cat)  d_\Sy}{\ln (1/\epemb)}+\frac{5 d_\Sy D_\cat \left( \ln d_\Sy +\ln D_\cat\right) \epemb^{1/7}}{\ln (1/\epemb)}  }\\
		\leq &\, 15\sqrt{ \frac{d_\Sy  \ln D_\cat}{\ln\left(1/\epemb\right) } \left(1 + D_\cat \epemb^{1/7} \right)},
	\end{align}
	where in line \ref{line:withcomment}, we have used the observation \begin{align}
		\lim_{x\to 0^+} \frac{x^p} {\ln^2 (1/x)}=\lim_{x\to 0^+} \frac{x^p} {\ln (1/x)}=0
	\end{align}
	for all $p>0$ and thus $x^{p} \leq \ln^2 (1/\epemb)$ for all fixed $p>0$ and sufficiently small $\epemb$.
	Thus if we set $\epres$ (defined via Eq. \ref{eq:ep res}) to 
	\begin{align}
		15\sqrt{ \frac{d_\Sy  \ln D_\cat}{\ln\left(1/\epemb\right) } \left(1 + D_\cat \epemb^{1/7} \right)},
	\end{align}
	we conclude the proof.
\end{proof}





\section{Miscelaneous lemmas for the proof of Theorem \ref{Thm:Implementation with Quasi-Idela clock}}

\begin{lemma}
	\label{lem:ak-phik}
	Let $a_k\geq 0$ satisfy $1-\ep_2\leq \sum_k a_k\leq 1 $, and $\phi_k$ satisfy $|\phi_k-\phi|\leq \ep_1$ for some $\phi$. Let $\ep_1,\ep_2\leq 1$. Then we have 
	\begin{equation}
		\left| \sum_k e^{-i \phi_k} a_k \right| \geq 
		1- \ep_2 - \ep_2
	\end{equation}
\end{lemma}
\begin{proof}
	\begin{align}
		&\left| \sum_k e^{-i \phi_k} a_k \right| 
		= 
		\left| \sum_k e^{-i (\phi+ \phi_k-\phi)} a_k \right| 
		=
		\left| \sum_k e^{-i(\phi- \phi_k)} a_k \right| 
		\geq Re \sum_k e^{-i(\phi- \phi_k)} a_k 
		=\sum_k \cos(\phi- \phi_k) a_k
		\geq  \nonumber \\
		& \sum_k\left(1-\frac{\ep_1^2}{2}\right) a_k
		=\sum_ka_k -  \frac{\ep_1^2}{2} \sum_k a_k 
		\geq 1 - \ep_2 -  \frac{\ep_1^2}{2} \geq 1 
		-\ep_1- \ep_2
	\end{align}
	where we have used $\cos x \geq 1-x^2/2$ 
	and $\ep_2\leq 1$.
\end{proof}

\begin{lemma}
	\label{lem:gaussian-tail}
	For $\delta>0$, and $d\geq 1/(2\delta)$, and for  $\sigma=1/\sqrt{d} $  we have 
	\begin{align}
		\sum_{|k|>\delta d} |(\psi(0,k)|^2
		\leq poly(d) e^{-\delta^2 d}
	\end{align}
\end{lemma}

\begin{proof}
	We have 
	\begin{align}
		\sum_{|k|>\delta d} |(\psi(0,k)|^2= 
		A \sum_{|k|>\delta d} e^{- \frac{2\pi}{\sigma^2}k^2}=2 A\sum_{k\geq \delta d }
		e^{- \frac{2\pi}{\sigma^2}k^2} \leq 
		2 A \int_{\delta d -1}^\infty
		e^{- \frac{2\pi}{\sigma^2}k^2}.
	\end{align}
	We now use exponential bound for gaussian tails:
	\begin{equation}
		\int_r^\infty \dt x  e^{-b^2  x^2 }\leq \frac{\sqrt{\pi}}{b} e^{-b^2  r^2 }
	\end{equation}
	and the assumption $d\geq 1/(2\delta)$ to bound our expression  further by 
	\begin{align}
		2 A \sigma e ^{-\frac{ 2\pi}{\sigma^2}(\delta d -1)^2}\leq A \sigma 
		e ^{-\frac{\pi\delta^2 }{2 \sigma^2}d^2}.
	\end{align}
	Now using estimate on $A$ and putting $\sigma=\sqrt{d}$ we get 
	\begin{align}
		\sum_{|k|>\delta d} |(\psi(0,k)|^2
		\leq poly(d) e^{-\delta^2 d}.
	\end{align}
\end{proof}

Here we prove lemma which bounds potential tail, giving rise to the estimate \eqref{eq:potential-tail-1}:
\begin{lemma}
	\label{lem:tail-potential}
	For $V_0$ given by 
	\begin{align}
		V_0(x)=A_c \cos^{2n}\left(\frac{x}{2}\right),
		\quad \text{with } \quad 
		A_c=\frac{2^{2n}}{2\pi{2n \choose n}}.
	\end{align}
	We have the following estimate for  potential tail for $0<\delta\leq 1/2$
	\begin{align}
		\label{eq:lem-potential-tail}
		\int_0^{\pi-2\pi \delta }
		V_d(x) \dt x \leq \frac{1}{\delta} e^{-n\delta^2}.
	\end{align}
\end{lemma}
\begin{proof}
	We write
	\begin{align}
		&\int_{2 \delta \pi }^\pi V_0(x) \dt x=A_c \int_{2 \delta \pi }^\pi
		\cos^{2n}(x/2) \dt x
		\leq A_c \frac{1}{\sin(\delta \pi)}
		\int_{2 \delta \pi }^\pi 
		\cos^{2n}(x/2) \sin(x/2)\dt x=
		\nonumber\\
		&=A_c  \frac{1}{\sin(\delta \pi)}\frac{2}{2n+1} \cos^{2n+1}(\delta \pi)
		\leq \frac{1}{\sin(\delta \pi)} \frac{1}{\pi} \cos^{2n+1}(\delta \pi)
	\end{align}
	where in the last inequality we have used 
	\begin{align}
		{2n \choose n} \geq 2^{2n} 
		\frac{2^{2n}}{2n+1}.
	\end{align}
	Finally we use $\sin x\geq x/\pi$
	as well as $\cos x\leq 1 - x^2/(2 \pi)$ for $0\leq x\leq \pi$ 
	to arrive at 
	\begin{align}
		\int_{2 \delta \pi }^\pi V_0(x) \dt x \leq 
		\frac{1}{\pi \delta} \left(1 - \frac{\pi}{2}\delta^2\right)^{2n+1} \leq \frac{1}{\delta} e^{-n\delta}, 
	\end{align}
	where in the second inequality we have used $e^{-x}\geq 1-x$ for $x\geq 0$, and have dropped some unnecessary terms.
\end{proof}
In the following lemma, we reproduce bound for $\varepsilon_\nu$ the quantity reporting for deviation of pointer's evolution 
from the simple picture of acquiring phase. 
\begin{lemma}
	\label{lem: ep-nu-estimate}
	For the potential  $V_0$ as in Eq. \eqref{eq:potential} with $n=\lceil d^{1/4}\rceil$
	we have the following estimate on $\||\varepsilon_\nu\>\|$ given by
	\eqref{eq:non free evollution of Quasi ideal}
	\begin{align}
		\||\varepsilon_\nu\>\|\sim
		t\, poly(d) e^{-c d^{1/4}}
	\end{align}
	where $c=1/(32 \pi)$.
\end{lemma}
\begin{proof}
	In \cite{WSO} it is proven that (see Section \emph{4.3. Quasi-canonical Commutation}; page 135)
	\begin{align}
		\| \ket{\varepsilon_\nu}\|_2 &=\varepsilon_\nu(t,d_\cl)=\, \bo\left( t\; poly(d)\; e^{-\frac{\pi}{4} \frac{d}{\zeta}} \right)\,\textup{as } d\rightarrow \infty,\label{eq:ep nu normaisation}
	\end{align}
	where $\zeta \geq 1$ is a measure of the size of the derivatives of $V_0(x)$,
	\begin{align}\label{eq:b zeta main text}
		\begin{split}
			\zeta &=\left( 1+\frac{0.792\, \pi}{\ln(\pi d)}b \right)^2,\quad\text{for any } \\
			b &\geq\; \sup_{k\in\nn^+}\left(2\max_{x\in[0,2\pi]} \left|\Omega_n  V_0^{(k-1)}(x) \right|\,\right)^{1/k},
		\end{split}
	\end{align}
	where $ V_0^{(k)}(x)$ is the $k^\textup{th}$ derivative with respect to $x$ of $ V_0(x)$.
	(Depending on how one chooses the potential $V_0$, the lower bound on $b$ may or may not depend on $d_\cl$).
	
	For the specialised potential $V_0$ of 
	the form \eqref{eq:potential}  
	in \cite{WSO} the value of $b$ was calculated (in Section 9.2. \emph{Examples of Potential Functions: The Cosine Potential}; page 177), and it follows that 
	that   $b$ satisfies
	\begin{align}
		b\leq n\sqrt{n}.
	\end{align}
	We also see that  for $b\geq \ln(d)$,
	and $d\geq3$
	\begin{align}
		\zeta \leq \left(\frac{2\pi}{b \ln d}\right)^2
		\leq (2 \pi b)^2=4\pi^2 n^3
	\end{align}
	Plugging the last estimate into \eqref{eq:ep nu normaisation} we obtain
	\begin{align}
		\varepsilon_\nu(t,d_\cl) \sim
		t\, poly(d) e^{-\frac{\pi}{4}\frac{d}{4 \pi^2 n^3}}
	\end{align}
	Choosing $n=\lceil d^\frac{1}{4}\rceil$ so that $n\leq 2 d^{1/4}$ we obtain the required estimate. 

\end{proof}
\cblack
 
\section{Miscellaneous lemma for the proof of Theorem \ref{thm:noemb physical t CTO}}

Here is the lemma used in the proof of Theorem \ref{thm:noemb physical t CTO} located in Section \ref{Sec:proof of thm 3, t-CTOs generic bound}.
\begin{lemma}
\label{lem:c}
We have 
\begin{align}\label{eq:c bounds}
	&1-|c|^2 \leq |1-c^2| \nonumber \\
	&|1-c| \leq |1-c^2|,
\end{align}
for all $c\in\cc$ satisfying $|c|\leq 1$ and $|1-c|\leq 1$.
\end{lemma}
\begin{proof}
Write $c$ in terms of real and imaginary parts: $c=a+\mi b$. The constraints $|c|\leq 1$ and $|1-c|\leq 1$ imply that $a$ satisfies $0\leq a\leq 1$. We find $(1-|c|^2)^2 = (1-a^2)^2-2(1-a^2)b^2+b^4$ while $|1-c^2|^2= (1-a^2)^2+2(1-a^2) b^2+b^4+4 a^2 b^2$, thus concluding the 1st inequality in \eqref{eq:c bounds}. For the second inequality in Eq. \ref{eq:c bounds}, we start by observing that $(1-a^2)^2 \leq (1-a)^2$ and $b^2 \leq (2+2 a^2) b^2 +b^4= 2(1-a^2) b^2+4 a^2 b^2 + b^4$, and thus $(1-a^2)^2+b^2 \leq (1-a)^2 + 2(1-a^2) b^2+4 a^2 b^2 + b^4$. On the other hand, we find $|1-c|^2= (1-a)^2+b^2$ and $|1-c^2|^2= (1-a^2)^2 + 2(1-a^2) b^2+4 a^2 b^2 + b^4$, thus proving the 2nd inequality in Eq. \eqref{eq:c bounds}.
\end{proof}

\section{Continuity for perturbations: proof of Proposition \ref{prop: H Int to epsig}}
In this section we evaluate 
the norm of the difference between $\me^{\mi(H+V)}$ and $\me^{\mi H}$ for Hermitian $H$ and $V$, along the lines of  \cite{VybornyTaylorReminder1981}.  
For completeness we shall prove 
the result basing our proof solely on the version of mean value theorem for vector valued functions of \cite{Wazewski}, in the form taken from  \cite{McLeod}:
\begin{prop}
	\label{prop:mean-value}
	Let $f$ be defined on interval $[a,b]\subset \rr$  with values in a $d$-dimensional linear space. Let $f$ be continuous on $[a,b]$ and differentiable on $(a,b)$. Then there exist numbers $\{c_k\}_{k=1}^d$ with $c_k\in (a,b)$ 
	and $\{\lambda_k\}_{k=1}^d$ satisfying $\sum_k\lambda_k=1$ such that 
	such that 
	\begin{eqnarray}
	f(b)-f(a)=(b-a) \sum_{k=1}^d \lambda_k f'(c_k).
	\end{eqnarray}
\end{prop}
Using this we shall prove a version of 
Taylor's reminder theorem:
\begin{prop}
	\label{prop:Taylor-reminder}
	Let $F$ a function defined on interval $[0,1]$ with values in a $d$-dimensional linear space. Let $F$ be $n+1$ times differentiable on interval $(0,1)$ and continuous on $[0,1]$.  We then have 
	\begin{eqnarray}
	F(1) - \sum_{k=0}^n \frac{F^{(k)}(0)}{k!}=\sum_{l=1}^d \lambda_l\frac{(1-t_l)^n}{n!} F^{(n+1)}(t_l),
	\end{eqnarray}
	for some $\lambda_l$'s satisfying $\sum_{l=1}^d\lambda_l=1$ and some $t_l\in(0,1)$, where $F^{(k)}(x):=\frac{\textup{d}^k}{\textup{d} x^k} F(x)$. 
\end{prop}
\begin{proof}
	Following \cite{VybornyTaylorReminder1981}
	we consider function $G$ defined as 
	\begin{eqnarray}
	G(t) = F(t) +\sum_{k=1}^n \frac{(1-t)^k}{k!} F^{(k)}(t).
	\end{eqnarray}
	We see that 
	\begin{eqnarray}
	G(1)=F(1), \quad G(0)=\sum_{k=0}^n \frac{F^{(k)}(0)}{k!}
	\end{eqnarray}
	and $G$ satisfies assumptions of Prop. \ref{prop:mean-value}. 
	Applying this Proposition to $G$ for $a=0$, $b=1$, we obtain the desired result. 
\end{proof}

We can apply the above proposition to the case $n=1$ and get 
\begin{eqnarray}
F(1)-F(0) = F'(0) + \sum_{l=1}^d \lambda_l\frac{(1-t_l)^2}{2} F''(t_l)
\end{eqnarray}
which implies (by convexity of norm, and triangle inequality):
\begin{eqnarray}
\label{eq:Fbound}
\|F(1)-F(0)\|_\infty \leq \| F'(0) \|_\infty + \frac12 \max_{t\in (0,1)}\|F''(t)\|_\infty.
\end{eqnarray}
Consider now  $F(t)=\me^{\mi(H+tV)}$.  We obtain the following.

\begin{lemma}
	\label{lem:Fprimbis}
	Let  $F(t)=\me^{\mi(H+tV)}$ for Hermitian matrices $H$, $V$.  We then have 
	\begin{equation}
	\label{eq:Fprimbis}
	\|F'(t)\|_\infty \leq \|V\|_\infty, \quad	\|F''(t)\|_\infty \leq \|V\|_\infty^2
	\end{equation}
	for $t\in(0,1)$. 
\end{lemma}
\begin{proof}
	We use the following general formula 
	\begin{eqnarray}
	\frac{\dt}{\dt t} \me^{X(t)} = \int_0^1 \me^{\alpha X(t)} \frac{\dt X(t)}{\dt t}  \me^{(1-\alpha) X(t)} \dt \alpha.
	\end{eqnarray}
	For $X=\mi(H+Vt)$ 	we get 
	\begin{eqnarray}
	\frac{\dt}{\dt t} \me^{\mi(H+V t)} = \int_0^1 \me^{\alpha \mi(H+ Vt)} V  \me^{(1-\alpha) \mi(H+Vt)} \dt \alpha
	=  \int_0^1 U_1 V  U_2 \dt \alpha
	\end{eqnarray}
	with $U_1$,$U_2$ unitaries. Using convexity and multiplicativity of operator norm, and $\|U\|_\infty=1$ 
	for unitaries we get 
	\begin{eqnarray}
	F'(t)\leq \int_0^1 \|V\|_\infty\dt \alpha=\|V\|_\infty.
	\end{eqnarray}
	Similarly we  have 
	\begin{eqnarray}
	F''(t)=\frac{\dt^2}{\dt t^2} \me^{\mi(H+ Vt)} &=& \int_0^1 \frac{\dt}{\dt t} \left(\me^{\alpha \mi(H+ Vt)}\right) \mi V  \me^{(1-\alpha) \mi(H+Vt)} + \me^{\alpha \mi(H+ Vt)} \mi V  \frac{\dt}{\dt t} \left(\me^{(1-\alpha) \mi(H+Vt)}\right) \dt \alpha=
	\nonumber \\
	&=& \int_0^1 \left\{
	\int_0^1
	\me^{\beta \alpha \mi(H+V t)}   \alpha \mi V 
	\, \me^{(1-\beta) \alpha \mi ( H+Vt)} \, \mi V \, \me^{(1-\alpha)\mi ( H+Vt)}\dt \beta  
	\right. 
	\nonumber \\
	&+& 
	\left. 
	\int_0^1
	\me^{\alpha \mi (H+V t)}  \,  \mi V \,  \me^{\beta \alpha \mi ( H+Vt)} \alpha \mi V \, \me^{(1-\beta)\alpha \mi ( H+Vt)}\dt \beta
	\right\} \dt \alpha
	\nonumber \\
	&=&\int_0^1 \dt \alpha
	\left\{
	\int_0^1 \mi ^2 \alpha V_1 V V_2 V V_3 \dt \beta +
	\int_0^1 \mi ^2 \alpha W_1 V W_2 V W_3 \dt \beta   
	\right\}
	\end{eqnarray}
	with $V_i$ and $W_j$ being unitary. 
	As before, passing to norms, using convexity of norm, multiplicativity of norm and triangle inequality, we arrive at 
	\begin{eqnarray}
	\|F''(t)\|_\infty\leq 2 \|V\|_\infty^2  \int_0^1\alpha\, \dt \alpha = \|V\|_\infty^2.
	\end{eqnarray}
\end{proof}

\begin{remark}
	Similarly one can prove that 	$\|F^{(k)}(t)\|_\infty \leq \|V\|_\infty^k$, $k\in\nn$.
\end{remark}

Now, combining Lemma \ref{lem:Fprimbis} with formula \eqref{eq:Fbound}
we obtain 
\begin{prop}
	\label{prop:UU0}
	We have 
	\begin{eqnarray}
	\|\me^{\mi(H+V)}-\me^{\mi H}\|_\infty\leq \|V\|_\infty + \frac12 \|V\|_\infty^2.
	\end{eqnarray}
\end{prop}
\begin{proof}
	We obtain the above equation by noting that for $F(t)=\me^{\mi(H+ tV)}$ we have
	$F(1)=\me^{\mi(H+V)}$, $F(0)=\me^{\mi H}$, and  inserting  these into \eqref{eq:Fbound} and using 
	\eqref{eq:Fprimbis}.
\end{proof}
To obtain bounds on states, we need the following well known fact (a special case of H\"older type inequalities \cite{HornJohnson}).
\begin{lemma} 
	\label{lem:AB}
	For arbitrary operators $A$,$B$ in finite dimensional Hilbert space we have 
	\begin{equation}
	\|AB\|_1 \leq \|A\|_1 \|B\|_\infty,\label{eq:AB lemma}
	\end{equation}
	where $\|\cdot\|_\infty$ denotes the infinity norm and $\|\cdot\|_1$ the one norm.
\end{lemma}

For the following proposition, we need to recall Eqs. \eqref{eq:2nd condition of robust embezzelment Ham 2}, \eqref{eq:2nd condition of robust embezzelment Ham}, \eqref{eq:epsig def} from the main text. We reproduce them here for convenience:

\begin{align}
\rho_{\Sy}^1\otimes\rho_{\cat}^0 &=\tr_\G\left[ \me^{- \mi\hat I_{\Sy\cat\G}^\textup{int} } (\rho_\Sy^0\otimes\rho_\cat^0\otimes\tauGibb)\, \me^{\mi\hat I_{\Sy\cat\G}^\textup{int} } \right],\label{eq:2nd condition of robust embezzelment Ham 3}\\
\sigma_{\Sy\cat}^1&:=\tr_\G\left[ \me^{- \mi\left(\hat I_{\Sy\cat\G}^\textup{int}+\delta \hat I_{\Sy\cat\G}^\textup{int}\right) } (\rho_\Sy^0\otimes\rho_\cat^0\otimes\tauGibb)\, \me^{\mi \left(\hat I_{\Sy\cat\G}^\textup{int}+\delta \hat I_{\Sy\cat\G}^\textup{int}\right) } \right],\label{eq:2nd condition of robust embezzelment Ham 1}\\
\epsig & := \|\sigma_{\Sy\cat}^1- \rho_\Sy^1\otimes\rho_\cat^0\|_1.\label{eq:epsig def 1}
\end{align}

\begin{prop}\label{prop: H Int to epsig}
	For all states $\rho^0_\Sy, \sigma_\Sy^1, \rho_\cat^0$ and Gibbs states $\tau_\G$, and for all Hermitian operators $\hat I_{\Sy\cat\G}^\textup{int}$, $\delta\hat I_{\Sy\cat\G}^\textup{int}$ satisfying Eq. \eqref{eq:2nd condition of robust embezzelment Ham 3}, the following bound on $\epsig$, defined in Eq. \eqref{eq:epsig def 1}, holds:
	\begin{align}\label{eq:ep bound app}
	\epsig \leq 2 \| \delta\hat I_{\Sy\cat\G}^\textup{int} \|_\infty+  \| \delta\hat I_{\Sy\cat\G}^\textup{int} \|_\infty^2.
	\end{align}
\end{prop}
\begin{proof} Let $U=\me^{\mi (H+V)}$ and $U_0=\me^{\mi H}$. Then for any state $\rho$ we have
	\begin{eqnarray}
	&&\|	U\rho U^\dagger - U_0 \rho U_0^\dagger\|_1=
	\|	U\rho U^\dagger -U \rho U_0^\dagger + U \rho U_0^\dagger - U_0 \rho U_0^\dagger\|_1\leq 
	\|	U\rho U^\dagger -U \rho U_0^\dagger\|_1 + \|U \rho U_0^\dagger - U_0 \rho U_0^\dagger\|_1\leq \nonumber \\
	&&\leq \|\rho(U^\dagger - U_0^\dagger)\|_1+\|(U-U_0)\rho\|_1
	\leq \|\rho\|_1 \|U^\dagger -U_0^\dagger\|_\infty +\|\rho\|_1 \|U-U_0\|_\infty =2\|U-U_0\|_\infty.
	\end{eqnarray}
	We have here used triangle inequality for first inequality, invariance of trace norm under unitaries for the second one, and Eq. \eqref{eq:AB lemma} for the third one. Next, using Prop. \ref{prop:UU0} we 
	obtain the needed relation
	\begin{equation}
	\| U\rho U^\dagger -U_0 \rho U_0^\dagger\|_1\leq 2\|V\|_\infty + \|V\|_\infty^2.
	\end{equation}
	Now in the above equation, set $\rho=\rho_\Sy^0\otimes\rho_\cat^0\otimes\tau_\G$ and $U_0=\me^{- \mi \hatf(t) \hat I_{\Sy\cat\G}^\textup{int} }$, $U=\me^{- \mi \left(\hat I_{\Sy\cat\G}^\textup{int}+ \delta\hat I_{\Sy\cat\G}^\textup{int}\right) }$. The bound Eq. \eqref{eq:ep bound app} now follows by applying the data processing inequality.
\end{proof}

\section{Proof of Theorem \ref{thm:4}}
\label{subsec:deltaxy}

Since the proof of Theorem \ref{thm:4} is similar to that of Theorem \ref{Thm:Implementation with Quasi-Idela clock}, we have relegated it to this \supp. For its proof, we need the following proposition. Furthermore, the following proposition will require calculations which are also performed in the proof of Proposition \ref{prop:clock-error}
We will refer the reader to the proof of 
Proposition \ref{prop:clock-error}
and not replicate them here.  
\begin{prop}\label{prop:intermedia to proof of thm 4} 
	Consider the potential $\hat H_\cl^\textup{int}=\hat V_d$ of Eq. \eqref{eq:potential-op} and Quasi-Ideal clock described in Subsection \ref{Overview of the quasi-ideal clock}. Set $k_0=0$. Then 
	for $t\in[0,t_1]\cup[t_2,T_0]$, and 
	provided that  $|\Delta(t;\omega,\omega')|\leq 2$, $\tilde\epsilon_V\leq1$ and $\varepsilon_{\nu}\leq 1$, we have 	
	\begin{align}
		\max_{\omega,\omega'\in[-\pi,\pi]}|1-\Delta(t;\omega,\omega')^2|\leq 
		12(\varepsilon_{LR}+\varepsilon_{\nu}+12 \tilde\epsilon_V)
	\end{align}
	where $\tilde\epsilon_V$ is given by Eq. \eqref{eq:potential-tail-1}, $\varepsilon_{LR}$ by Eq. \eqref{eq:epsilon-LR}  and $\varepsilon_\nu=\varepsilon_\nu(t,d_\cl)$ by Eq. \eqref{eq:epsilon-nu-epsilon-c}.
\end{prop}

\begin{proof}
	Let us repeat here useful definitions 
	from Eqs. \eqref{eq:clock plus Ham def}:
	\begin{align}
		\Delta(t;\omega,\omega')&:= \braket{\rho_\cl^0|\hat\Gamma_\cl^\dag(\omega,t)\hat \Gamma_\cl(\omega',t)|\rho_{\cl}^0},\\
		\hat\Gamma_\cl(\omega,t)&:= \me^{-\mi t\hat H_\cl+ \mi \omega \left(\hatf(t) \id_\cl- t\hat H_\cl^\textup{int}\right)}, \quad \omega,t\in\rr.\label{eq:clock plus Ham def2}
	\end{align} 
	Recall that $\theta(t)$ is defined as 
	\begin{align}\label{eq:def target S Cat G2}
		\hatf(t)=
		\begin{cases}
			0 &\mbox{ for }  t \in[0, t_1]\\
			1 &\mbox{ for }  t\in[t_2,T_0].
		\end{cases}
	\end{align}
	The  initial state of the clock is the 
	state of Eq. \eqref{gaussianclock}
	\begin{align}
		\ket{\rho_\cl^0} = 
		\ket{\Psi_\textup{nor}(k_0)}.
	\end{align}
	we then have \begin{align}
		\Delta(t;\omega,\omega')=
		\me^{-\mi(\omega'-\omega) \hatf(t)}
		\bra{\Psi_\textup{nor}(k_0)}
		\me^{-\mi t( \omega \hat V_d+\hat H_\cl)}
		\me^{\mi t( \omega' \hat V_d+\hat H_\cl)}
		\ket{\Psi_\textup{nor}(k_0)}.
	\end{align}
	Using \eqref{eq:non free evollution of Quasi ideal} we write
	\begin{align}
		\label{eq:Delta-ooprim}
		\Delta(t;\omega,\omega')=
		\me^{-\mi(\omega'-\omega) \hatf(t) }\left(
		\braket{\psi_\omega|\psi_{\omega'}}+
		\braket{\psi_\omega|\varepsilon_{\nu}}
		+\braket{\varepsilon_{\nu}|\psi_{\omega'}}+ \braket{\varepsilon_{\nu}|\varepsilon_{\nu}}\right)
	\end{align}
	with
	\begin{align}
		\ket{\psi_\omega}=\ket{\bar\Psi_\textup{nor}(k_0+td/T_0, t d/T_0)}
	\end{align}
	where in definition of  
	$\ket{\bar\Psi_\textup{nor}(k_0+td/T_0, t d/T_0)}$ we take $\omega$ in place of $\Omega_n$. 
	Due to \eqref{eq:epsilon-nu-epsilon-c}
	the last three terms in \eqref{eq:Delta-ooprim} are small, hence we need to consider
	\begin{align}
		\tilde\Delta:= \me^{-\mi(\omega'-\omega) \hatf(t) }\braket{\psi_\omega|\psi_{\omega'}}=
		\me^{-\mi(\omega'-\omega) \hatf(t)}
		\sum_{k:|k- td/T_0|\leq  d/2}  \me^{-\mi (\omega-\omega') \int_{k}^{k+td/T_0} dy { V_d}(y)}\left| \psi_\textup{nor}(td/T_0;k)   \right|^2,
	\end{align}
	where we have set $k_0=0$ as in the proof of Prop. \ref{prop:clock-error}. 
	As in \eqref{eq:split-delta} we write 
	\begin{align}
		\label{eq:Delta-tilde-LCR}
		\tilde\Delta =
		\tilde\Delta_C + \tilde\Delta_{LR}.
	\end{align}
	and as in \eqref{eq:tail-for-Delta} we obtain 
	\begin{align}
		|\tilde \Delta_{LR}|\leq 2 \varepsilon_{LR}
	\end{align}
	So we have to estimate 
	\begin{align} \label{eq:Delta C def2}
		\tilde\Delta_C:= \me^{-\mi(\omega'-\omega)\hatf(t) } \sum_{k: |k- td/T_0|\leq \psidelta d }  \me^{-\mi (\omega-\omega') \int_{k}^{k+td/T_0} dy { V_d}(y)}\left| \psi_\textup{nor}(td/T_0;k)   \right|^2.
	\end{align}
	As in proof of Proposition \ref{prop:clock-error} 
	we  introduce real numbers $\phi_k$, $a_k$ given by 
	\begin{equation}
		\phi_k = 
		\begin{cases}
			(\omega-\omega') \int_{k}^{k+td/T_0} dy { V_d}(y) &\mbox{ for }  t \in[0, t_1]\\
			(\omega-\omega') \left( 1- \int_{k}^{k+td/T_0} dy { V_d}(y) \right) &\mbox{ for }  t\in[t_2,T_0].
		\end{cases}
	\end{equation} 
	and 
	\begin{align}
		a_k=\left| \psi_\textup{nor}(td/T_0;k)   \right|^2,
	\end{align}
	so that 
	\begin{align}
		\tilde\Delta_C = \sum_{k:|k-td/T_0|\leq \psidelta d} a_k e^{-i \phi_k}.
	\end{align}
	Since the state \eqref{gaussianclock} was normalized, and due to estimates \eqref{eq:tail-for-Delta} we know that 
	\begin{align}
		\label{eq:akbound}
		1-2\varepsilon_{LR} \leq  
		\sum_{{k:|k-td/T_0|\leq \psidelta d}} a_k \leq 1. 
	\end{align}
	Similarly as in Prop. \ref{prop:clock-error}
	we obtain that  for 
	$k$ satisfying $|k-t d/T_0|\leq \psidelta d$ we have 
	\begin{align}
		&\int_{k}^{k+td/T_0} dy { V_d}(y)\leq \tilde\epsilon_V \quad \text{for}\quad t\in[0,t_1]
		\nonumber \\
		& 1- 
		\int_{k}^{k+td/T_0} dy { V_d}(y)\leq \tilde\epsilon_V \quad \text{for}\quad t\in[t_2,T_0].
	\end{align}
	hence, since $\omega,\omega'\in[-\pi,\pi]$, we have
	\begin{align}
		\label{eq:phikbound}
		|\phi_k|\leq 2 \pi \tilde\epsilon_V
	\end{align}
	Thus, using $\cos x \geq 1-x^2$ and $|\sin x| \leq |x|$ as well as plugging in estimates \eqref{eq:akbound},\eqref{eq:phikbound} we obtain bounds for real and imaginary parts of $\tilde \Delta_C$:
	\begin{equation}
		\label{eq:bound-Re-Im}
		\Re (\tilde\Delta_C) \geq (1- \varepsilon_{LR})(1- 4\pi^2 \tilde\epsilon_V)
		\geq 1 -\varepsilon_{LR}- 4\pi^2 \tilde\epsilon_V,
		\quad |\Im (\tilde \Delta_C)|
		\leq 2 \pi \tilde\epsilon_V.
	\end{equation}
	We are now in position to deal with 
	$\Delta(t;\omega,\omega')$. Using \eqref{eq:Delta-ooprim} 
	and \eqref{eq:Delta-tilde-LCR} we obtain 
	\begin{equation}
		\Delta(t;\omega,\omega')=
		\Re(\tilde\Delta_C) + z,
	\end{equation}
	where 
	\begin{equation}
		z= i \Im(\Delta_c) + \tilde\Delta_L+ \tilde\Delta_R+
		\me^{-\mi(\omega'-\omega) \hatf(t) }\biggl(
		\braket{\psi_\omega|\varepsilon_{\nu}}
		+\braket{\varepsilon_{\nu}|\psi_{\omega'}}+ \braket{\varepsilon_{\nu}|\varepsilon_{\nu}}\biggr).
	\end{equation}
	Since $\psi_\omega,\psi_{\omega'}$ are normalized, 
	we have 
	\begin{equation}
		\left|\me^{-\mi(\omega'-\omega) \hatf(t) }\biggl(
		\braket{\psi_\omega|\varepsilon_{\nu}}
		+\braket{\varepsilon_{\nu}|\psi_{\omega'}}+ \braket{\varepsilon_{\nu}|\varepsilon_{\nu}}\biggr)\right|
		\leq 2 \varepsilon_{\nu} +  |\varepsilon_{\nu}|^2
	\end{equation}
	so that 
	\begin{align}
		\label{eq:bound-z}
		|z|\leq 2 \pi \tilde\epsilon_V + 2 \varepsilon_{LR}  + 
		2 \varepsilon_{\nu} +  |\varepsilon_{\nu}|^2.
	\end{align}
	Now, using assumption that $\Delta(t;\omega,\omega')\leq 2$ and estimates \eqref{eq:bound-Re-Im} and \eqref{eq:bound-z} we get
	\begin{equation}
		|1-\Delta(t;\omega,\omega')^2|\leq 
		3|1-\Delta(t;\omega,\omega')|\leq 3(|1-\Re(\tilde\Delta_C)|+|z|) 
		\leq  3\left(2 \varepsilon_{LR}+4 \pi^2 \tilde\epsilon_V^2 + 2 \pi \tilde\epsilon_V + 2 \varepsilon_{LR}  + 
		2 \varepsilon_{\nu} +  |\varepsilon_{\nu}|^2\right)
	\end{equation}
	Using the assumption that $\tilde\epsilon_V\leq1$ and $\varepsilon_{\nu}\leq 1$ we obtain
	\begin{equation}
		|1-\Delta(t;\omega,\omega')^2|\leq 12(\varepsilon_{LR}+\varepsilon_{\nu}+12 \tilde\epsilon_V).
	\end{equation}
	We have thus obtained a bound that is independent of the choice of $\omega$ and $\omega'$; this proves the required estimate. 
\end{proof}
\

We can now prove Theorem \ref{thm:4}: 
\begin{proof}
	Due to Prop. \ref{prop:intermedia to proof of thm 4} 
	we have to provide an upper bound for 
	\begin{align}
		\label{eq:eps}
		12(\varepsilon_{LR}+\varepsilon_{\nu}+12 \tilde\epsilon_V).
	\end{align}
	The appropriate bounds for the above epsilons 
	are given by Eqs. \eqref{eq:epsilons-bounds}
	hence the needed scaling follows
\end{proof}


\end{document}